\documentclass[11pt]{amsart}

\usepackage{graphicx}
\usepackage{pstricks}

\newtheorem{thm}{Theorem}[section]
\newtheorem{prop}[thm]{Proposition}
\newtheorem{defn}[thm]{Definition}
\newtheorem{lemma}[thm]{Lemma}

\newtheorem{cor}[thm]{Corollary}
\newtheorem{remark}[thm]{Remark}
\newtheorem{example}[thm]{Example}

\usepackage{amsmath}
\usepackage{amsxtra}
\usepackage{amscd}
\usepackage{amsthm}
\usepackage{amsfonts}
\usepackage{amssymb}
\usepackage{eucal}
\newcommand{\bmb}{\left( \begin{array}{rr}}
\newcommand{\enm}{\end{array}\right)}

\newcommand{\cM}{\mathcal M}
\newcommand{\wy}{\widehat{y}}

\newcommand{\cP}{\mathcal P}
\newcommand{\cA}{\mathcal A}

\newcommand{\bZ}{{\mathbf Z}}

\newcommand{\bP}{{\mathbf P}}

\newcommand{\Z}{{\mathbb Z}}

\newcommand{\bF}{{\mathbf F}}
\newcommand{\bG}{{\mathbf G}}
\newcommand{\bT}{{\mathbf T}}
\newcommand{\ba}{{\mathbf a}}
\newcommand{\bc}{{\mathbf c}}
\newcommand{\bd}{{\mathbf d}}
\newcommand{\be}{{\mathbf e}}
\newcommand{\bbf}{{\mathbf f}}
\newcommand{\bb}{{\mathbf b}}
\newcommand{\bm}{{\mathbf m}}

\newcommand{\bx}{{\mathbf x}}
\newcommand{\by}{{\mathbf y}}
\newcommand{\bu}{{\mathbf u}}
\newcommand{\bw}{{\mathbf w}}
\newcommand{\bR}{{\mathbf R}}
\newcommand{\bS}{{\mathbf S}}

\newcommand{\bX}{{\mathbf X}}
\newcommand{\bJ}{{\mathbf J}}
\newcommand{\bC}{{\mathbf C}}
\newcommand{\bK}{{\mathbf K}}
\newcommand{\bL}{{\mathbf L}}
\newcommand{\bz}{{\mathbf z}}
\newcommand{\bbw}{{\mathbf w}}
\newcommand{\bW}{{\mathbf W}}
\newcommand{\bone}{{\mathbf 1}}
\newcommand{\al}{{\alpha}}

\numberwithin{equation}{section}

\begin{document}
\title{Noncommutative integrability, paths and quasi-determinants}
\author{Philippe Di Francesco} 
\address{Department of Mathematics, University of Michigan,
530 Church Street, Ann Arbor, MI 48190, USA
and Institut de Physique Th\'eorique du Commissariat \`a l'Energie Atomique, 
Unit\'e de Recherche associ\'ee du CNRS,
CEA Saclay/IPhT/Bat 774, F-91191 Gif sur Yvette Cedex, 
FRANCE. e-mail: philippe.di-francesco@cea.fr}

\author{Rinat Kedem}
\address{Department of Mathematics, University of Illinois MC-382, Urbana, IL 61821, U.S.A. e-mail: rinat@illinois.edu}
\date{June 2010}
\begin{abstract}
In previous work, we showed that the solution of certain systems of discrete integrable equations, 
notably $Q$ and $T$-systems, is given in terms of partition functions of positively weighted paths,
thereby proving the positive Laurent phenomenon of Fomin and Zelevinsky for these cases. 
This method of solution is amenable to generalization to non-commutative weighted paths. 
Under certain circumstances, these describe solutions of discrete evolution equations in 
non-commutative variables: Examples are the corresponding quantum cluster 
algebras \cite{BZ}, the Kontsevich evolution \cite{DFK09b} and the $T$-systems 
themselves \cite{DFK09a}. In this paper, we formulate certain non-commutative 
integrable evolutions by considering paths with non-commutative weights, together with 
an evolution of the weights that reduces to cluster algebra mutations in the 
commutative limit. The general weights are expressed as Laurent monomials
of quasi-determinants of path partition functions, allowing for a non-commutative 
version of the positive Laurent phenomenon.
We apply this construction to the known systems, and obtain Laurent positivity results for
their solutions in terms of initial data.
\end{abstract}

\maketitle
\tableofcontents

\section{Introduction}

\subsection{Discrete integrable systems and positivity}
Discrete integrable systems are evolution equations in a discrete time
variable, which possess sufficiently many conserved quantities
(discrete integrals of motion) so that they can be solved explicitly in terms of
some discrete Cauchy initial data.

Cluster algebra mutations \cite{FZI} can be viewed as discrete
evolution equations for certain commutative variables, where at each
step, the evolution
is a rational transformation.
The evolution is known to result in variables which are Laurent polynomials as
functions of the initial data, even though the evolution equations are
rational transformations \cite{FZLaurent}.  

Cluster algebras are algebras over commutative variables (the cluster
variables) in which sets of variables (the clusters) are associated
with the nodes of a regular $n$-tree, and cluster variables on
connected nodes are related by an evolution called a mutation.  One
can choose a cluster at any fixed node to be the initial data for the
evolution. Fomin and Zelevinsky conjectured that in terms of any
initial data, the cluster variables at any node are positive Laurent
polynomials \cite{FZI}. This is known as the positivity conjecture.

The $A_r$ $Q$-systems and $T$-systems are discrete integrable
evolution equations which are a subset of a cluster algebra structure
\cite{Ke07,DFK08}. These systems were
solved, and thus their positivity property proved, in \cite{DFK09},
\cite{DFK09a} (for the $T$-systems, see \cite{DFT} for generic boundary
conditions).  The approach of \cite{DFK09} and \cite{DFK09a} involves
the interpretation of the solutions of the systems as partition
functions for paths on some discrete target space, with weights
explicitly related to the (discrete Cauchy) initial data of the
evolution equations.  

The $T$-system can be considered as a particular {\it non-commutative}
version of the $Q$-system, in which we adjoin a non-commuting element
to the commutative algebra generated by $Q$-system solutions \cite{DFK09a}.  
The $T$-system equations are matrix elements of the cluster algebra
mutations twisted by the non-commutative element.  The solutions are
then expressed in terms of partition functions for paths as for the
$Q$-system, but with non-commutative weights.

In all these examples, the positive Laurent
property follows from the fact that the weights are explicit positive
Laurent monomials of the initial data.

The aim of the present paper is to formalize this underlying path
structure in the more general case of arbitrary non-commutative
weights, with the idea of defining non-commutative (integrable)
cluster algebra mutations with
a built-in non-commutative
positive Laurent property.

We note that there exists an example of a non-commutative cluster
algebra with these properties: the quantum cluster algebra
\cite{BZ}, where the cluster variables $q$-commute. This is the mildest
form of non-commutativity, and in this case, much more can be said
about the structure of the solutions.

\subsection{The $A_r$ $T$-system and $Q$-system}
The $A_r$ $T$-system is a system of discrete evolution equations for
the commuting variables $\{T_{i,j,k}\}$. It can be written in the
following form:\footnote{The relative minus sign which appears in the
  original version \cite{BazResh,KNS} is normalized away in this paper.}
\begin{equation}\label{Tsystem}
T_{i,j,k+1}T_{i,j,k-1}=T_{i,j+1,k}T_{i,j-1,k}+T_{i+1,j,k}T_{i-1,j,k}, \quad 1\leq i \leq r, j,k\in \Z,
\end{equation}
with the convention that $T_{0,j,k}=T_{r+1,j,k}=1$.

Equation \eqref{Tsystem} (with a sign change on the right-hand side) was first introduced \cite{BazResh,KNS} in the
context of the generalized Heisenberg spin chain in statistical
mechanics. With appropriate boundary conditions, it is the ``fusion relation'' satisfied by its transfer
matrices. Its solutions are the
$q$-characters of $U_q(\mathfrak{sl}_{r+1})$  \cite{FR,Nakajima}.

A closely related system is the $A_r$ $Q$-system \cite{KR},
\begin{equation}\label{Qsystem}
R_{i,k+1}R_{i,k-1}=R_{i,k}R_{i,k}+R_{i+1,k}R_{i-1,k},\quad 1\leq i \leq r, k\in \Z,
\end{equation}
with $R_{0,k}=R_{r+1,k}=1$.  
This system is obtained from the $T$-system \eqref{Tsystem} by
dropping the index $j$, for example, by considering solutions
of \eqref{Tsystem} which are $2$-periodic in $j$,
$T_{i,j+2,k}=T_{i,j,k}$ for all $i,j,k$. Then $R_{i,k}=T_{i,i+k\, {\rm
    mod}\, 2,k}$ is a solution to \eqref{Qsystem}.

\begin{remark} Given the appropriate boundary conditions
($R_{i,0} = (-1)^{\lfloor i/2\rfloor} $) this system is satisfied by
the characters $Q_{i,n}=(-1)^{\lfloor i/2\rfloor}R_{i,n}$ ($n>0$) of
the irreducible $A_r$-module with special highest weights which are
multiples of one of the fundamental weights. Using Equation
\eqref{Qsystem}, these are expressed as polynomials in the characters
of the fundamental modules, identified as $\{Q_{i,1}:\ 1\leq i \leq r\}$.
\end{remark}

In reference to the origin of the equation, we refer to Equation
\eqref{Qsystem} as the $A_r$ $Q$-system throughout this paper, without
imposing this special boundary condition.

To solve such a system means to express the variables for any $n$ in
terms of the initial data set. For example, valid initial data sets
for the $Q$-system are sets of $2r$ variables
$\{Q_{i,m_i},Q_{i,m_{i}+1}: 1\leq i \leq r\}$ such that
$|m_i-m_{i+1}|\leq 1$. For fixed values of $(i,k)$ or $(i,j,k)$,
Equations \eqref{Tsystem},\eqref{Qsystem} are mutations of cluster
variables in a cluster algebra.  We will therefore investigate the
systems from the point of view of local mutations of their initial
data.

\subsection{Example: Commutative and non-commutative $A_1$ $Q$-systems}
\label{cononcosec}

The simplest case of the $Q$-system occurs for $r=1$. This example was
previously considered from a different point of view by several
authors \cite{CZ,SZ,PM,FV}.
Let us recall
the solutions of both the commutative and non-commutative versions of the
$A_1$ $Q$-system \cite{DFK3,DFK09b}. We wish to stress the similarity
of their path solutions, which relies on the existence of a sufficient
number of conserved quantities and a linear recursion relation
involving these.

Consider the recursion relation for the commutative variables $\{R_{n}\}$,
\begin{equation} \label{aoneQsystem}
R_{n+1}R_{n-1}=(R_n)^2+1
\end{equation}
The system has a manifest translational invariance. We thus compute
the solution $R_n$, $n\geq 0$ in terms of the initial data
$(R_0,R_1)$.

The system is integrable: The quantity $c={R_{n+1}\over R_n}+{1\over
  R_nR_{n+1}}+{R_n\over R_{n+1}}$ is independent of $n$. Moreover, the
solutions of \eqref{aoneQsystem} obey the linear recursion relation
$R_{n+1}-cR_n+R_{n-1}=0$.  Let 
$$y_1=R_1/R_0,\ 
y_2=1/(R_0R_1),\ y_3=1/y_1=R_0/R_1.$$ 
and define the
generating function
\begin{equation}
\label{cfaoneco}
F_0(t)=\sum_{n=0}^\infty t^n {R_n\over R_0}= {1\over 1- {ty_1 \over 1-{t y_2\over 1-ty_3}}}.
\end{equation}
This is the partition function for weighted paths on the integer segment
$[1,4]$, from $1$ to $1$, and with steps $\pm 1$, with weights $1$ for step
$i\to i+1$ and $ty_i$ for step $i+1\to i$, $i=1,2,3$. As mentioned above,
this interpretation makes positivity manifest, as the $y_i$'s are
explicit Laurent monomials of the initial data $(R_0,R_1)$.

The non-commutative $A_1$ $Q$-system is in a class of
wall crossing formulas introduced by
Kontsevich, and involves the variables
$\bR_n\in\cA$, a skew field of rational fractions in the variables
$\bR_0^{\pm 1},\bR_1^{\pm1}$. It can be written as \cite{DFK09b} 
\begin{equation} \label{aoneNCQsystem}
\bR_{n+1}(\bR_n)^{-1}\bR_{n-1}=\bR_n+(\bR_n)^{-1}
\end{equation}
The goal is to express $\bR_n$, $n\geq 0$ in terms of the
initial data $(\bR_0,\bR_1)$.

The system is integrable, in the sense that there are two
conserved quantities \cite{DFK09b}, independent of $n$:
\begin{equation*}
\bC=(\bR_{n+1})^{-1}\bR_n\bR_{n+1}(\bR_n)^{-1}, \quad 
\bK=\bR_{n+1}(\bR_n)^{-1}+(\bR_{n+1})^{-1}(\bR_n)^{-1}+(\bR_{n+1})^{-1}\bR_n
\end{equation*}
and the solutions to \eqref{aoneQsystem} obey the linear recursion
relation $\bR_{n+1}-\bK\, \bR_n+\bC\bR_{n-1}=0$. Define
$\by_1=\bR_1(\bR_0)^{-1}$, $\by_2=(\bR_{1})^{-1}(\bR_0)^{-1}$ and
$y_3=(\bR_{1})^{-1}\bR_0$, with $\by_3\by_1=\bC$. Then
the generating function
\begin{equation}
\label{cfaonenc}
\bF_0(t)=\sum_{n=0}^\infty t^n \bR_n(\bR_0)^{-1}= 
\left( 1- \big( 1- (1-t \by_3)^{-1}t\by_2\big)^{-1}t\by_1\right)^{-1}.
\end{equation}
This can be interpreted as the partition function of weighted paths on
the integer segment $[1,4]$, from $1$ to $1$, and with steps $\pm 1$,
with non-commutative weights $1$ for $i\to i+1$ and $t\by_i$ for
$i+1\to i$, $i=1,2,3$. In this case, the weight of a path is the
product of step weights, ordered according to the order in
which they are taken. 

We see that the non-commutative $A_1$ $Q$-system, despite its
complexity, has essentially the {\it same} solution as the commutative
one: it is expressed in terms of paths on the same graph, with
similar but non-commutative weights related to the initial data, and
with the condition that the order in which weights are multiplied is
that in which the corresponding steps are taken.

We recall that in the commutative case, the mutation of initial data in
terms of which $R_n$ is expressed
may be implemented via local {\it
  rearrangements} of the finite continued fraction expression
\eqref{cfaoneco} for $F_0(t)$ \cite{DFK3}. This property was shown to
carry over in the non-commutative case \cite{DFK09b}, where mutations of
the initial data are obtained via local non-commutative rearrangements
of the non-commutative finite continued fraction expression
\eqref{cfaonenc} for $\bF_0(t)$.

\subsection{Plan of the paper}

The example above suggests the following strategy for exploration
of the non-commutative evolutions.
Instead of starting from an evolution equation, we start from the ``solution", namely some
path models on the {\it same} target graphs as those considered in the solution of commutative
systems, but with non-commutative weights. It turns out that the partition functions of these paths
have a structure similar to that observed in the commutative case. In particular, we will define
non-commutative mutations connecting them via non-commutative finite continued fraction
rearrangements. Moreover, by investigating
the relation between the coefficients of these continued fractions
(i.e. the weights of the path models) and the coefficients of their series expansion, we 
will be led naturally to express results in terms of the quasi-determinants introduced 
in \cite{GGRW}. Our main result is the construction of a chain
of path models with non-commutative weights, all related via mutations, and whose
partition functions involve a basic set of non-commuting variables, 
in a manifestly positive Laurent polynomial way.

The paper is organized as follows. 

We first review in Section \ref{sectiontwo} the solution of the
(commutative) $A_r$ $Q$-system, in a new path formulation related to that of \cite{DFK3} (the equivalence is described in detail in Appendix \ref{appendixa}). Each initial data, and therefore
each graph and associated set of weights is coded by a Motzkin path of length $r-1$.
The corresponding partition functions of weighted paths display
a simple finite continued fraction structure, and the mutations are shown to be expressed through
local rearrangements.

Section \ref{section3} is devoted to the non-commutative generalization of this structure, and involves
paths on the same graphs but with non-commutative weights. Mutations are shown to be expressed
again as some non-commutative local rearrangements of the finite continued fraction expressions
for the partition functions of paths, giving rise to simple evolution equations 
for the corresponding non-commutative weights. 
We then define non-commutative variables by discrete 
quasi-Wronskian determinants of the coefficients of the series expansion of the partition functions.
These variables are shown to satisfy an evolution equation (the discrete non-commutative Hirota equation,
Theorem \ref{dischironc})
playing the role of the $Q$- or $T$-system in the non-commutative setting. Finally, we express the weights
of all the mutated path models in terms of these discrete quasi-Wronskian variables (Theorem \ref{ncweightcalc}),
as manifestly positive Laurent monomials.

Sections \ref{section4} and \ref{section5} are applications respectively to the $T$-system, as an example of
non-commutative $Q$-system, and to the quantum $Q$-system, in relation with the corresponding
quantum cluster algebra \cite{BZ}. In both cases, we compute the discrete quasi-Wronskian variables in terms of
the initial data, and obtain Laurent positivity theorems for the solutions of the corresponding systems.

In Section 6, we give some further examples of non-commutative evolutions which have a Laurent positivity property. Section 7 considers the problem of generalizing of the Lindstr\"om-Gessel-Viennot theorem to the non-commutative setting. We conclude with some remarks about the possible generalization of cluster algebra relations and variables to the generic non-commutative setting.

\medskip
\noindent{\bf Acknowledgments.} We would like to thank S. Fomin for discussions, and L. Faddeev
for pointing out Refs.\cite{FKV,FV} to us.
We thank the Mathematisches ForschungsInstitut Oberwolfach, Germany, Research in Pairs program
(Aug. 2009) during which this work was initiated.
PDF also thanks S. Fomin for hospitality 
at the Dept. of Mathematics of the University of Michigan, Ann Arbor (Spring 2010).
RK thanks the Institut Henri Poincar\'e, Paris, France, for hospitality during the
semester ``Statistical Physics, Combinatorics and Probability" (Fall 2009), 
and the Institut de Physique Th\'eorique du CEA Saclay, France.
PDF received partial support from the ANR Grant GranMa, the
ENIGMA research training network MRTN-CT-2004-5652,
and the ESF program MISGAM. RK is supported by NSF grant DMS-0802511.

\section{Continued fraction solutions to Q-systems}\label{sectiontwo}
\subsection{The $A_r$ $Q$-system}
We consider the set of discrete evolution equations
\begin{equation}\label{Qsys}
  R_{i,k+1} R_{i,k-1} = R_{i,k}^2 + R_{i-1,k}R_{i+1,k},
\quad 1\leq i\leq r,\ k\in Z
\end{equation}
with the boundary conditions $R_{0,k}=R_{r+1,k}=1$ for all
$k\in \Z$.

\subsection{Initial data sets}
For fixed $\al$ and $n$, the variable $R_{\al,n+1}$ is a rational
function of $R_{\beta,n}$ with $|\al-\beta|\leq 1$ and
$R_{\al,n-1}$. Therefore the recursion relation \eqref{Qsys} has a
solution for a given set of initial data, which must include the $2r$
variables $R(\mathbf m)=\{R_{\al,m_\al},R_{\al,m_{\al}+1} : 1\leq \al
\leq r\}$ where the sequence $\bm =
(m_1,...,m_r)$ is such that $|m_i-m_{i+1}|\leq 1$. That is, initial data sets
of the $Q$-system are in bijection with Motzkin paths of length $r-1$.

Any variable $R_{\al,k}$ which is a solution of \eqref{Qsys} can then
be expressed as a function of the initial data set $R(\bm)$. Moreover,
due to translational invariance of Eq.\eqref{Qsys}, the expression
for $R_{\al,n}$ as a function of $R(\bm)$ is the {\it same} as that of
$R_{\al,n+k}$ as a function of $R(k+\bm)$, with the translated Motzkin
path $k+\bm=(k+m_1,...,k+m_r)$. We may therefore restrict ourselves to
a fundamental domain of Motzkin paths modulo translations, $\cM=\{\bm\vert
{\rm Min}(m_i)=0\}$.

In \cite{Ke07}, we showed that the relations \eqref{Qsys} are
mutations in a coefficient-free cluster algebra, defined from the seed
cluster variables $(\bx_0,B)$, where 
$$\bx_0=(R_{1,0},...,R_{r,0};R_{1,1},...,R_{1,r}),\quad
B=\bmb 0 & -C \\ C & 0 \enm,$$ where $C$ is the Cartan matrix of the Lie
algebra $A_r$. This seed corresponds to the fundamental Motzkin path,
$\bm = (0,...,0)$.  

Not all mutations in the cluster algebra are of
the form \eqref{Qsys}. However, all the variables $\{R_{\al,n}:\ 1\leq
\al\leq r, \ n\in\Z\}$ are contained in the subalgebra made up of the
union of seeds parametrized by the set of Motzkin paths of length
$r-1$. Thus, we consider the corresponding sub-cluster algebra, restricting
to the nodes of the Cayley tree that are labeled by Motzkin paths, and to the mutations
that connect them.

We note that mutations in the sub-cluster algebra under consideration
here have the following property: $\mu_i$ is given by one of the
$Q$-system equations if and only if it maps $R(\bm)$ to $R(\bm')$,
where $m_j'=m_j \pm \delta_{i,j}$. That is, it maps a Motzkin path to
another Motzkin path which differs from it in only one node.

\subsection{Path partition functions}

\subsubsection{Positivity from weighted paths on graphs}

In \cite{DFK3}, we first eliminated from the $Q$-system \eqref{Qsys}
the variables $R_{i,k}$ $i>1$, $k\in \Z$,
in terms of the $R_{1,n}$, $n\in \Z$, by noting that the discrete Wronskians
\begin{equation}\label{discwronskco}
W_{i,n}=\det_{1\leq a,b\leq i} \left( R_{1,n+i+1-a-b} \right) \end{equation}
actually satisfy the same recursion relation \eqref{Qsys}, as a consequence of
the Desnanot-Jacobi identity for the minors of any given matrix (also a particular case
of the Pl\"ucker relations, see \cite{DFK3} for details). As moreover $W_{0,n}=1$
and $W_{1,n}=R_{1,n}$, we conclude that $W_{i,n}=R_{i,n}$ for all $i\geq 0$, $n\in \Z$.

The $Q$-system therefore boils down to a single non-linear recursion relation for $R_{1,n}$,
in the form $W_{r+1,n}=1$.
In Ref.\cite{DFK3}, we explicitly solved this equation, by using its integrability. The latter
implied the existence of a linear recursion relation for the $R_{1,n}$'s, with constant coefficients
(independent of $n$, and expressible explicitly in terms of initial data). 
Exploiting this, we further showed that $R_{1,n}$ could be related to
the partition functions of certain weighted path models on target graphs 
$\Gamma_\bm$ (see definition in
Appendix \ref{appendixa}),
indexed by Motzkin paths $\bm$ of length $r-1$,
the weights being directly expressed as positive Laurent
monomials of the corresponding initial data $R(\bm)$. 

Finally, the quantity $R_{i,n}$,
expressed as the determinant $W_{i,n}$ \eqref{discwronskco}
was shown to also be related to path partition functions,
namely that of collections of $i$ strongly non-intersecting 
paths on the same target graphs $\Gamma_\bm$, with the
same weights. This followed from a slight generalization of the 
Lindstr\"om-Gessel-Viennot theorem \cite{LGV1} \cite{LGV2}.

Moreover, we obtained an explicit expression for the
generating function of the variables $\{R_{1,n}: n\geq 0\}$ as the
partition function of paths on the graph $\Gamma_\bm$, in the form of
an explicit finite continued fraction.

In this section, we present a similar but simpler description in terms
of paths on another, more compact graph $G_\bm$, with related weights.
The equivalence of the two formulations is proved in Appendix
\ref{appendixa}.

\subsubsection{The weighted graphs $G_\bm$}\label{Gmintro}

For fixed $r$, consider the weighted graph $G_\bm$. The graph itself
depends only on $r$. It has $r+1$ vertices labeled $1,...,r+1$ and
oriented edges connecting vertex $i$ to itself, to $i-1$ and to $i+1$.
This is pictured below,
\begin{center}

\psset{unit=2mm,linewidth=.4mm,dimen=middle}
\begin{pspicture}(-10,-2)(10,22)
\psline(5,0)(5,20)
\multips(.25,-2.5)(0,5){5}{\pscircle(2.5,2.5){2}}
\multips(5,0)(0,5){5}{\pscircle*[linecolor=red](0,0){.7}}
\rput(7,0){1}
\rput(7,5){2}
\rput(9,20){$r+1$}
\rput(-5,10){$G_\bm=$}

\end{pspicture}\end{center}
where each edge connecting two different vertices is to be considered
as a doubly oriented edge.  The map $\bm\mapsto \Gamma_\bm\mapsto
G_\bm$ is described in Appendix \ref{appendixa}.

To make the contact between the $G_\bm$ path formulation and that on
$\Gamma_\bm$, we need a few definitions. We start from a collection of $2r+1$
formal variables,
$y(\bm)=(y_1(\bm),...,y_{2r+1}(\bm))$. These will eventually be
identified with the skeleton weights of the graph $\Gamma_\bm$ (see Appendix
\ref{appendixa}).
\begin{defn}\label{weightdefco}
  To $y(\bm)$ we associate bijectively the collection $\widehat{
    y}(\bm)=(\wy_1(\bm),...,\wy_{2r+1}(\bm))$. This is defined according to the
  local structure of the Motzkin path $\bm$ as follows:
\begin{eqnarray}
\wy_{2i}&=& t y_{2i} (t y_{2i+1})^{m_{i+1}-m_{i}},\nonumber\\
\wy_{2i-1} &=& \left\{ \begin{array}{ll}
    t y_{2i-1}-\frac{y_{2i}}{y_{2i+1}},& m_{i+1}=m_i-1;\\
    t y_{2i-1} + t y_{2i}, & m_{i+1} = m_i+1; \\
    t y_{2i-1} ,& \hbox{otherwise}.\end{array} \right.\label{Gweights}
\end{eqnarray}
\end{defn}

\begin{example} The three basic examples are as follows:
\begin{itemize}
\item If $\bm=\bm_0=(0,...,0)$, the fundamental Motzkin path, we have
  $\wy_i=ty_i$, $i=1,2,...,2r+1$.  
\item If $\bm=\bm_1=(0,1,2,...,r-1)$, the strictly
  ascending Motzkin path, we have 
$$\wy_{2i-1}=t(y_{2i-1}+y_{2i}), \quad
  \wy_{2i}=t^2y_{2i}y_{2i+1} (1\leq i \leq r-1),$$ with
  $\wy_{a}=ty_{a}$ for $a\geq 2r-1$.  
\item If $\bm=\bm_2=(r-1,r-2,...,0)$, the strictly
  descending Motzkin path,  
$$\wy_{2i-1}=ty_{2i-1}-{y_{2i}\over
    y_{2i+1}},\quad \wy_{2i}={y_{2i}\over y_{2i+1}}, (1\leq i \leq
  r-1),$$
 with $\wy_{a}=ty_{a}$ for $a\geq 2r-1$.
\end{itemize}
\end{example}

We now attach weights to the edges of the graph $G_\bm$.  The edge
connecting vertex $i$ to vertex $j$ carries a weight $w_{i,j}:=
w_{i,j}(\bm)$, where
$$
w_{i,i+1}=1;\quad w_{i,i}=\wy_{2i-1}(\bm);\quad w_{i+1,i}=\wy_{2i}(\bm),
$$
with $\wy_i(\bm)$ as in \eqref{Gweights}.

\subsubsection{Partition functions of paths on 
$G_\bm$ and continued
  fractions}  
\label{commupf}
Each path $p$ on the graph $G_\bm$ from node $i$ to node $j$ carries a
weight $wt(p)$ which is the product of the weights $w_{i,j}$ of the
edges $i\to j$ traversed by the path. The partition function of paths
from $i$ to $j$ is the sum of the weights of all paths $p\in \mathcal
P_{i,j}(G_\bm)$ which start at vertex $i$ and end at vertex $j$:
$$
Z_{i,j}(G_\bm) = \sum_{p\in\mathcal P_{i,j}(G_\bm)} wt(p).
$$

Once expressed in terms of $\{y_i\}$, the weights of Definition \ref{weightdefco}
introduce terms with negative signs.
The following lemma shows why this does not introduce minus signs in the
generating function of paths.
\begin{lemma}\label{positinvol}
The partition function $Z_{i,j}(G_\bm)$ as a power series in $t$, has
coefficients which are positive Laurent polynomials in $\{y_{i}\}_i$.
\end{lemma}
\begin{proof}
Let $I_\pm=\{i: m_i=m_{i-1}\pm 1\}$.
The partition function for paths on $G_\bm$ is the same as that on the
graph $G_\bm'$, where each loop at vertex $i$ such that $i\in I_-$
is decomposed into two loops as follows:

\begin{center}
\psset{unit=2mm,linewidth=.4mm,dimen=middle}
\begin{pspicture}(0,3)(10,20)
\psline[linecolor=black](5,10)(5,15)
\psline[linestyle=dotted](5,5)(5,10)
\psline[linestyle=dotted](5,15)(5,20)
\rput(2.75,10){\pscircle(0,0){2}}

\rput(6.5,10){ $i$}
\rput(8,15){ $i+1$}
\rput(4,13){\blue $b$}
\rput(-2.5,10){\blue $a-b$}
\multips(5,10)(0,5){2}{\pscircle*[linecolor=red](0,0){.7}}

\psline{->}(10,11)(13,11)
\psline[linecolor=black](20,10)(20,15)
\psline[linestyle=dotted](20,15)(20,20)
\psline[linestyle=dotted](20,5)(20,10)
\rput(17.75,10){\pscircle(0,0){2}}
\rput(22.25,10){\pscircle(0,0){2}}
\rput(21,10){ $i$}
\rput(23,15){$i+1$}
\multips(20,10)(0,5){2}{\pscircle*[linecolor=red](0,0){.7}}
\rput(14.5,10){\blue $a$}
\rput(19,13){\blue $b$}

\rput(26,10){\blue $-b$}

\end{pspicture}
\end{center}

The corresponding piece of $G'_\bm$ is depicted on the right. In this
case, $a=t y_{2i-1}$ and $b=y_{2i}/y_{2i+1}$. We call the two
corresponding loops
on the right-hand side $\ell_i^+$ (with weight $a$) and $\ell_i^-$
(with weight $-b$). Note
that the weight of $\ell_i^-$ is equal to 
 $-w_{i+1,i}$, as shown in the picture.

We define an involution $\varphi$ on paths $p\in \cP_{i,j}(G_\bm')$ as
follows. 
\begin{itemize}
\item $\varphi(p)=p$ if $p$ does not pass through
any $\ell_i^-$ ($i\in I_-$) nor contains any sequence of steps $i\to i+1\to
i$ ($i\in I_-$);
\item Otherwise, consider $s$, which is either the first step of $p$
  passing through some 
  $\ell_i^-$, or the first sequence of the form $i\to i+1\to i$ ($i\in I_-$) in $p$,
whichever comes first.
The path $\varphi(p)$ is obtained from $p$ by replacing $s$ by 
the other sequence (a step through $\ell_i^-$ by $i\to i+1\to i$ and vice
versa).
\end{itemize} 
Clearly, $\varphi$ is an involution on
$\cP_{i,j}(G_\bm')$. Its effect is to reverse the sign of the weight,
$wt(\varphi(p))=-wt(p)$ if $\varphi(p)\neq p$. There is
therefore a pairwise cancellation of
contributions of paths such that
$p\neq \varphi(p)$. The terms which are left in the partition function
are the fixed points $p=\varphi(p)$ of the involution. These
have no step through any loop $\ell_i^-$ ($i\in I_-$) nor any sequence
$i\to i+1\to i$ ($i\in I_-$).

The remaining weights are all
positive Laurent polynomials of the $y_i$'s, and the lemma follows.
\end{proof}

Given the graph $G_\bm$ with weights $w_{i,j}(\bm)$, the partition
function of paths from vertex $1$ to itself can be written as a finite
continued fraction as follows.
\begin{defn}
The Jacobi-type (finite) continued fraction is defined inductively as
\begin{equation}\label{inducJ}
J^{(r)}(x_1,...,x_{2r+1}) = \frac{1}{1-x_1 - x_2 J^{(r-1)}(x_3,...,x_{2r+1})},
\quad J^{(1)}(x)=\frac{1}{1-x}.
\end{equation}
\end{defn}

\begin{lemma}
  The generating function of weighted paths on $G_\bm$ from and to the
  vertex $1$ can be written as
\begin{equation}
Z_{1,1}(G_\bm) = J^{(r)}(\wy(\bm)).
\end{equation}
\end{lemma}
\begin{proof} By induction.
\end{proof}
We regard $J^{(r)}(\wy(\bm))$ as a power series in $t$. As we have seen in
Lemma \ref{positinvol}, the coefficients of $t^n$ are positive Laurent
polynomials in $y(\bm)$.

  Alternatively, the positivity of the generating function
  $Z_{1,1}(G_\bm)$ may be recovered by a simple rearrangement of the
  continued fraction $J^{(r)}(\wy(\bm))$. Starting with the bottom term in
  the continued fraction and moving upwards, one iteratively
  rearranges the fraction for each $i\in I_-$ using the identity
\begin{equation}\label{identoneco}
a-{b\over c} +{{b\over c}\over 1-c -u}=a+{b+{b\over c}u\over 1-c-u},
\end{equation}
with $a=ty_{2i-1}$, $b=ty_{2i}$ and $c=ty_{2i+1}$. (Here, $u$ denotes the
already rearranged remainder of the fraction). This eliminates
any negative weights in the fraction, at the cost of
a modification of the branching structure of the fraction (the
fraction $u$
appears in two places), and the introduction of new, positive weights
coming from the term ${b\over c} u$.




The fraction $J^{(r)}(\wy(m))$ can be re-expressed in a {\it canonical}, manifestly
  positive form as follows.  Going though the fraction iteratively
  from bottom to top, whenever $i\in I_-$, 
 apply the identity
  \eqref{identoneco}. Otherwise, if
  $m_i=m_{i-1}+1$,  apply the identity
\begin{equation}\label{identwoco}
a+b +{bc\over 1-c-u}=a+{b\over 1-{c\over 1-u}}
\end{equation}
with $a=ty_{2i-1}$, $b=ty_{2i}$ and $c=t y_{2i+1}$. (Here,
$u$ is the already rearranged remainder of the fraction.)

The resulting rearranged form of $J^{(r)}(\wy(\bm))$ has the property that
all its coefficients are Laurent monomials of $y(\bm)$ with
coefficient $t$. 

We call this form canonical.  
\begin{defn}
  The canonical form of the fraction $J^{(r)}(\wy(\bm))$ is the fraction
  resulting from the use of the rearrangements \eqref{identoneco} and
  \eqref{identwoco}, such that the final fraction has
  coefficients which are all Laurent monomials of $y(\bm)$, with
  coefficient $t$.
\end{defn}

Another important generating function is the partition function of
paths from $1'$ to $1'$ on the graph
$\Gamma_\bm$ of Appendix \ref{appendixa}. It can be written as follows:
\begin{equation}\label{defgeneF}
F_\bm(t)=1 + t y_1(\bm) J^{(r)}(\wy(\bm)).
\end{equation}

\begin{remark}
The canonical form of $F_\bm(t)$ itself (with the identity \eqref{identwoco}
applied also at the top level, with $a=0$, $b=1$ and $c=y_1(\bm)$) is related
directly to the path interpretation of $F_\bm(t)$ as generating
weighted paths on the graph $\Gamma_\bm$ (see Appendix
\ref{appendixa}, Eq.\eqref{resolvent}); while the $y_i$'s are the
skeleton weights of $\Gamma_\bm$, the redundant weights have been
produced by the rearrangements \eqref{identoneco}.
\end{remark}

\begin{example}
The fraction corresponding to the Motzkin path $(0,1,0)$ for $r=3$ reads:
\begin{eqnarray*}
F_{0,1,0}(t)&=&1+t{y_1\over 1-t(y_1+y_2)-t^2 {{y_2 y_3} \over 1-ty_3
+{y_4\over y_5}-{{y_4\over y_5}\over 1-t y_5 -t {y_6\over 1-t y_7}}}}\\
&=&{1\over 1-t {y_1\over 1-t {y_2\over 
1-t {y_3\over 1-t {y_4+{{y_4y_6\over y_5}\over 1-t y_7}\over 1-t y_5-t{y_6\over 1-t y_7}}}}}}
\end{eqnarray*}
Note the occurrence of a new ``redundant" weight $t y_4y_6/y_5$.
\end{example}

\subsection{Mutations}\label{mutacosec}

Without loss of generality we now restrict ourselves to forward 
mutations of $\bm$,
$\mu_i:\ \bm\mapsto \bm'$, with $m'_j=m_j+\delta_{j,i}$. 
We also restrict ourselves to two
special cases:
\begin{itemize}
\item{\rm  Case (i):} $\mu_i:\ (m_{i-1},m_i,m_{i+1})=(a,a,a+1)\mapsto
  (a,a+1,a+1)$
\item{\rm  Case (ii):} $\mu_i: (m_{i-1},m_i,m_{i-1})=(a,a,a)\mapsto
  (a,a+1,a)$. 
\end{itemize}
These two cases are sufficient to obtain any Motzkin path, up to a
global translation by an overall integer,
from the fundamental Motzkin path $\bm_0=(0,0,...,0)$.

We also define the action of a mutation on $\Z(y)$:
\begin{defn} \label{mutationony}
Let $\mu_i(\bm)=\bm'$. Define $y':=y(\bm')$ as a function
  of $y:=y(\bm)$ 
  as follows:
\begin{eqnarray}
{\rm\bf Case}\  {\bf (i)}\ :&& \left\{ \begin{matrix} 
y_{2i-1}' & = & y_{2i-1}+y_{2i} \\
y_{2i}' & = &  y_{2i+1}y_{2i}/(y_{2i-1}+y_{2i})\\
y_{2i+1}' & = & y_{2i+1}y_{2i-1}/( y_{2i-1}+y_{2i})
\end{matrix} \right. \label{caseone}\\
{\rm\bf Case}\  {\bf (ii)}\ :&&\left\{ \begin{matrix} 
y_{2i-1}' & = & y_{2i-1}+y_{2i} \\
y_{2i}' & = & y_{2i+1}y_{2i}/(y_{2i-1}+y_{2i}) \\
y_{2i+1}' & = & y_{2i+1}y_{2i-1}/(y_{2i-1}+y_{2i})\\
y_{2i+2}' & = & y_{2i+2}y_{2i-1}/(y_{2i-1}+y_{2i})
\end{matrix} \right. \label{casetwo}
\end{eqnarray}
while all other weights are left unchanged, 
\end{defn}
We then have the natural definitions of $\wy(\bm')$ in terms of $y(\bm')$,
using Equations \eqref{weightdefco}.

\begin{thm}\label{defymut}
Under the mutation $\mu_i$, the function $J^{(r)}(\wy(\bm))$ is transformed as
\begin{eqnarray}
J^{(r)}(\wy(\bm)) = \left\{ \begin{array}{ll}
J^{(r)}(\wy(\bm')), & i>1; \\
1 + t y_1(\bm') J^{(r)}(\wy(\bm'), & i=1.\end{array}\right.
\end{eqnarray}
\end{thm}
This has an immediate corollary.
\begin{cor}\label{Ffunction}
Under the same mutation $\mu_i$, $F_\bm(t)$ is transformed
as:
\begin{eqnarray}\label{alphanotonemut}\label{alphaonemut}
F_\bm(t) = \left\{ \begin{array}{ll}
F_{\bm'}(t), & i>1; \\ 
1+ty_1(\bm)F_{\bm'}(t),& i=1. \end{array}\right.
\end{eqnarray}
\end{cor}

To prove the theorem, we use two rearrangement lemmas in addition to
Equations \eqref{identoneco} and \eqref{identwoco}:
\begin{lemma}\label{rear1lem}
For all $a,b$ we have
\begin{equation}\label{rear1}
1+{a\over 1-a-b} ={1\over 1-{a\over 1-b}}
\end{equation}
\end{lemma}

\begin{lemma}\label{rear2lem}
For all $a,b,c,d$, with $a+b$ invertible, we have
\begin{equation}\label{rear2}
a+{b\over 1-{c\over 1-d}}=
{a'\over 1-{b'\over 1-c'-d}}
\end{equation}
where $a',b',c'$ are defined by
\begin{equation}\label{defabc}
a'=a+b \qquad b'={b c\over a+b} \qquad c'={ac\over a+b}
\end{equation}
\end{lemma}

We now turn to the proof of Theorem \ref{defymut}.
\begin{proof}
Let $J_i(\bm)$ be defined inductively as follows: 
$$
J_i(\bm)=\wy_{2i-1}+{\wy_{2i}\over 1- J_{i+1}(\bm)},\ (i\geq 0),\qquad
J_{r+1}(\bm) = \wy_{2r+1}(\bm),
$$
with the convention that $\wy_{-1}=0$ and $\wy_0=1$. Note that $J_0(\bm) =
J^{(r)}(\wy(\bm))$.
We refer below to the ``level'' $i$ of the fraction to denote the part
of the fraction $J^{(r)}$ where
$J_i$ appears.
Let
$$\begin{array}{llll}
a=ty_{2i-3}(\bm), & b=ty_{2i-2}(\bm),
&  c=ty_{2i-1}(\bm),  &  d=ty_{2i}(\bm),\\  
e=ty_{2i+1}(\bm), &
f=ty_{2i+2}(\bm), & u={1\over 1-J_{i+2}(\bm)}. & \end{array}
$$
together with their primed counterparts, which are functions of $\bm'$.

We now consider the effects of the mutations of Def. \ref{mutationony} on
$J^{(r)}(\wy(\bm))$.
\begin{itemize}
\item{\rm \bf Case (i):}
We first transform the Jacobi-like form into the local canonical form 
for $J_{i-1}(\bm)$ by applying \eqref{identwoco} at the level $i$, namely:
\begin{equation*}
J_{i-1}(\bm)=a +{b\over 1-c-d-{d e\over 1-e-f u}}=a+{b\over 1-c-{d\over 1-{e\over 1-f u}}}
\end{equation*}
Next, we apply Lemma \ref{rear2lem} to the second level of this fraction, namely:
\begin{equation*}
J_{i-1}(\bm)=a +{b\over 1- {c'\over 1-{d'\over 1-e'-f u}}}, \ {\rm with}\ c'=c+d,\ d'=de/c', \ {\rm and}\  e'=ce/c'
\end{equation*}
Finally, we put this back into Jacobi-like form, by applying \eqref{identwoco} backwards at the level $i-1$.
The complete transformation reads:
\begin{equation*}
J_{i-1}(\bm)=a +{b\over 1-c-d-{d e\over 1-e-f u}}=J_{i-1}(\bm')=a+b +{bc'\over 1-c'-{d'\over 1-e'-f u}}
\end{equation*}
\par
\begin{tabular}{ccc}
\psset{unit=1.5mm,linewidth=.4mm,dimen=middle}
\begin{pspicture}(0,0)(30,14)
\psline(0,0)(0,5)(5,10)
\rput(0,0){\pscircle*[linecolor=red](0,0){.7}}
\rput(0,5){\pscircle*[linecolor=red](0,0){.7}}
\rput(5,10){\pscircle*[linecolor=red](0,0){.7}}
\rput(4,0){$i-1$}
\rput(2,5){$i$}
\rput(9,10){$i+1$}
\psline(20,0)(20,12)
\psline[linestyle=dotted](20,12)(20,14)
\rput(24,0){$i-1$}
\rput(22,5){$i$}
\rput(24,10){$i+1$}
\multips(17.75,0)(0,5){3}{\pscircle(0,0){2}}
\multips(20,0)(0,5){3}{\pscircle*[linecolor=red](0,0){.7}}
\rput(15,0){\blue\small $a$}
\rput(12,5){\blue \small $c+d$}
\rput(15.75,5){$\times$}
\rput(15,10){\blue \small $e$}
\rput(21,2.5){\blue \small $b$}
\rput(22,7.5){\blue\small $de$}
\rput(21,12.5){\blue\small $f$}
\rput(20,7.5){$\times$}
\end{pspicture}
&\psset{unit=1.5mm,linewidth=.4mm,dimen=middle}
\begin{pspicture}(0,0)(10,14)
\psline{->}(2,6)(7,6)
\end{pspicture}
 &
\psset{unit=1.5mm,linewidth=.4mm,dimen=middle}
\begin{pspicture}(0,0)(30,14)
\psline(0,0)(5,5)(5,10)
\rput(0,0){\pscircle*[linecolor=red](0,0){.7}}
\rput(5,5){\pscircle*[linecolor=red](0,0){.7}}
\rput(5,10){\pscircle*[linecolor=red](0,0){.7}}
\rput(4,0){$i-1$}
\rput(7,5){$i$}
\rput(9,10){$i+1$}
\psline(20,0)(20,12)
\psline[linestyle=dotted](20,12)(20,14)
\rput(24,0){$i-1$}
\rput(22,5){$i$}
\rput(24,10){$i+1$}
\multips(17.75,0)(0,5){3}{\pscircle(0,0){2}}
\multips(20,0)(0,5){3}{\pscircle*[linecolor=red](0,0){.7}}
\rput(12.5,0){\blue\small $a+b$}
\rput(15.75,0){$\times$}
\rput(15,5){\blue \small $c'$}
\rput(15,10){\blue \small $e'$}
\rput(22,2.5){\blue \small $bc'$}
\rput(20,2.5){$\times$}
\rput(21.5,7.5){\blue\small $d'$}
\rput(21,12.5){\blue\small $f$}
\end{pspicture}
\end{tabular}
\vskip .5in
\item{\rm \bf Case (ii):}
We apply directly Lemma \ref{rear2lem} to the level $i$,
namely:
\begin{equation*}
J_{i-1}(\bm)=a +{b\over 1-c-{d\over 1-e-f u}}=a +{b\over 1-{c'\over 1- {d'+{d'f'\over e'}u \over 1-e'-f'u}}}, \ {\rm with}
\left\{ \begin{matrix}
c'=c+d\\
d'=de/c'\\
e'=ce/c'\\
f'=cf/c'
\end{matrix}\right.
\end{equation*}
We bring this back to Jacobi-like form by applying first \eqref{identoneco} backward on its first level, and
then \eqref{identwoco} backward on its level zero, namely:
\begin{equation*}
J_{i-1}(\bm)=a +{b\over 1-{c'\over 1+{d'\over e'} -{{d'\over e'}\over 1- e'-f' u}}}
=a +b +{bc'\over 1-c'+{d'\over e'}-{{d'\over e'}\over 1- e'-f' u}}
\end{equation*}
The complete transformation reads:
\begin{equation*}
J_{i-1}(\bm)=a +{b\over 1-c-{d\over 1-e-f u}}=J_{i-1}(\bm')=a +b +{bc'\over 1-c'+{d'\over e'}-{{d'\over e'}\over 1- e'-f' u}}
\end{equation*}
\par
\begin{tabular}{ccc}

\psset{unit=1.5mm,linewidth=.4mm,dimen=middle}
\begin{pspicture}(0,0)(30,18)
\psline(0,0)(0,10)
\multips(0,0)(0,5){3}{\pscircle*[linecolor=red](0,0){.7}}

\rput(4,0){$i-1$}
\rput(2,5){$i$}
\rput(4,10){$i+1$}
\psline(20,0)(20,15)
\psline[linestyle=dotted](20,15)(20,18)
\rput(24,0){$i-1$}
\rput(22,5){$i$}
\rput(24,10){$i+1$}
\multips(17.75,0)(0,5){3}{\pscircle(0,0){2}}
\multips(20,0)(0,5){3}{\pscircle*[linecolor=red](0,0){.7}}
\rput(15,0){\blue\small $a$}
\rput(15,5){\blue \small $c$}
\rput(15,10){\blue \small $e$}
\rput(21,2.5){\blue \small $b$}
\rput(21,7.5){\blue\small $d$}
\rput(21,12.5){\blue\small $f$}
\end{pspicture}
&
\psset{unit=1.5mm,linewidth=.4mm,dimen=middle}
\begin{pspicture}(0,0)(10,18)
\psline{->}(2,6)(7,6)
\end{pspicture}
&
\psset{unit=1.5mm,linewidth=.4mm,dimen=middle}
\begin{pspicture}(0,0)(30,18)
\psline(0,0)(5,5)(0,10)
\rput(0,0){\pscircle*[linecolor=red](0,0){.7}}
\rput(5,5){\pscircle*[linecolor=red](0,0){.7}}
\rput(0,10){\pscircle*[linecolor=red](0,0){.7}}
\rput(4,0){$i-1$}
\rput(7,5){$i$}
\rput(4,10){$i+1$}
\psline(20,0)(20,15)
\psline[linestyle=dotted](20,15)(20,18)
\rput(24,0){$i-1$}
\rput(22,5){$i$}
\rput(24,10){$i+1$}
\multips(17.75,0)(0,5){3}{\pscircle(0,0){2}}
\multips(20,0)(0,5){3}{\pscircle*[linecolor=red](0,0){.7}}
\rput(12.5,0){\blue\small $a+b$}
\rput(15.75,0){$\times$}
\rput(12,5){\blue \small $c'-\frac{d'}{e'}$}
\rput(15.75,5){$=$}
\rput(15,10){\blue \small $e'$}
\rput(22,2.5){\blue \small $bc'$}
\rput(20,2.5){$\times$}
\rput(22,7.5){\blue\small $\frac{d'}{e'}$}
\rput(20,7.5){$=$}
\rput(21,12.5){\blue\small $f'$}
\end{pspicture}\end{tabular}
\vskip .5in
\end{itemize}

Note that for $i>1$, we always have $J_0(\bm')=J_0(\bm)$,
or equivalently $F_\bm(t)=F_{\bm'}(t)$.
When $i=1$, we have $a=0$, $b=1$ and in all cases $c'=y_1(\bm')$, so the final formula
amounts to $J_0(\bm)=1+t J_0(\bm')y_1(\bm')$, or equivalently
$F_\bm(t)=1+tF_{\bm'}(t)y_1(\bm)$. 

This proves the theorem and its corollary.
\end{proof}

\subsubsection{Explicit expression for the weights $y_{i}(\bm)$}

We can now identify the definitions of the mutations
on the initial data variables $\mu_i:\ R(\bm)\mapsto R(\bm')$ and the action on
weights $\mu_i:y(\bm)\to y(\bm')$.

Corollary \ref{Ffunction}
may be used as a definition of the series $F_\bm(t)$, by using induction
over mutations, once we specify the initial condition. We can do so by
giving the value of $F_{\bm_0}(t)$ as a function of $R(\bm_0)$, 
the variables corresponding to the fundamental Motzkin path. This 
will define the initial values $y_i(\bm_0)$ of the weights. Using
Definition \ref{mutationony}, we can then determine $y(\bm)$ for any
$\bm$ inductively.

In Ref. \cite{DFK3}, we used the integrability of the $Q$-system
\eqref{Qsys} to show that the generating function
$$F_0(t)=\sum_{n\geq 0} t^n R_{1,n}/R_{1,0}$$
is expressed explicitly in terms of the variables $R(\bm_0)$ as the
finite continued fraction:
$$ F_0(t)=1+t{y_1\over 1-ty_1-t{y_2\over 1-t y_3-t{y_4\over \ddots {{} \over 1-ty_{2r-1}-t{y_{2r}\over y_{2r+1}}}}}}$$
where 
\begin{equation}\label{initweights}
 y_{2i-1}={R_{i,1}R_{i-1,0}\over R_{i,0}R_{i-1,1}} \qquad y_{2i}={R_{i+1,0}R_{i-1,1}\over R_{i,0}R_{i,1}} 
 \end{equation}
Comparing this expression with that for $F_\bm(t)$ \eqref{defgeneF},
one sees that the choice of weights $y_i(\bm_0)=y_i$ as in \eqref{initweights}
gives
$$ F_{\bm_0}(t)=F_0(t)$$

We can now derive explicit expressions for the variables $y(\bm)$ for
any $\bm$ using Definition \ref{mutationony}.

\begin{thm}\label{linkmuta}
For any Motzkin path $\bm=(m_1,m_2,...,m_r)$, 
the weights $y_i(\bm)$ defined by  \eqref{caseone}-\eqref{casetwo} with the initial
conditions 
$y_i(\bm_0)$ as in \eqref{initweights} are expressed in terms of the
variables $R(\bm)$ as:
\begin{equation}\label{yweights}
y_{2i-1}(\bm) = \frac{c_{i,m_i+1}}{c_{i-1,m_{i-1}+1}},\qquad
y_{2i}(\bm)=d_{i,m_i+1} L_{i,i+1}
(L_{i,i-1})^{-1} ,
\end{equation}
where
$$
L_{i,i+\varepsilon} = \left\{ \begin{array}{ll}
    \frac{c_{i+\varepsilon,m_{i+\varepsilon}+1}}{c_{i+\varepsilon,m_{i}+1}},& {\rm if} \  m_{i}=m_{i+\varepsilon}+\varepsilon; \\
 1, & \hbox{\rm otherwise}.\end{array}\right. \nonumber
$$
and
$$
c_{i,m}=\frac{R_{i,m}}{R_{i,m-1}},\quad d_{i,m} = \frac{R_{i+1,m}R_{i-1,m-1}}{R_{i,m}R_{i,m-1}}.
$$
\end{thm}
\begin{proof}
By induction under mutations, from $\bm_0$ to any Motzkin path $\bm$.
We explicitly check the relations \eqref{caseone}-\eqref{casetwo} on the expressions
\eqref{yweights}, by use of the $Q$-system. 
\end{proof}
\begin{cor}
  The series $F_\bm(t)$ is the following generating function
 for solutions of the $Q$-system:
\begin{equation}\label{genefone}
F_\bm(t)=\sum_{n\geq 0} t^n R_{1,n+m_1}/R_{1,m_1}
\end{equation}
\end{cor}
\begin{proof}
  Whenever a mutation $\mu_i:\bm\mapsto \bm'$, $i>1$ is applied, the
  function $F_{\bm}(t)=F_{\bm'}(t)$ is preserved, hence its series
  coefficient remain the same. When $i=1$ however, we have the
  transformation $F_\bm(t)=1+ty_1(\bm)F_{\bm'}(t)=1+t
  {R_{1,m_1+1}\over R_{1,m_1}}F_{\bm'}(t)$, and therefore the
  coefficients of the series $F_{\bm'}(t)$ change according to
  \eqref{genefone}.
\end{proof}

\section{Non-commutative weights and evolutions}\label{section3}
In this section, we introduce a non-commutative version of the path partition functions and weight evolutions which appear in
the $A_r$ Q-system. Most of the
key ideas of the previous section generalize to the non-commutative setting because paths are
themselves non-commutative objects, and respect ordering of
non-commutative weights.

\subsection{Paths with non-commutative weights and 
non-commutative
  continued fractions} 

Let $\{\by_i^{\pm1}:\ 1\leq i
\leq 2r+1\}$ be non-commutative variables and let  $\mathcal A$ be the skew field of rational functions in $
\{\by_1^{\pm1},...,\by_{2r+1}^{\pm1}\}$.
Let $G_\bm$ denote the  graph used in Section \ref{Gmintro}.
We associate to each edge of $G_\bm$ a
non-commutative weight, which is an element in $\cA[t]$. 

 A path on  a graph with non-commutative weights is defined to have
a weight equal to the product of weights of the edges traversed by the
path, written {\it in the order in which the steps are taken}, from
left to right. For example, if the path is of the form
$p=(i_1,i_2,...,i_n)$, then its weight is
${\mathbf wt}(p)=\bw_{i_1,i_2}\bw_{i_2,i_3}\cdots \bw_{i_{n-1},i_n}$, where
$\bw_{i,j}$ is the weight of edge from vertex $i$ to vertex $j$.

\subsubsection{Non-commutative weights}
For each Motzkin path in the fundamental domain, $\bm\in\cM$, let
$\by(\bm)\equiv (\by_1(\bm),...,\by_{2r+1}(\bm))\in \cA^{2r+1}$ be a
collection of $2r+1$ non-commutative variables. Let $I_\pm = \{i: m_i =
  m_{i-1}\pm1\}$ as before.
To $\by(\bm)$ we associate
the variables
${\widehat{\by}}(\bm)=({\widehat{\by}}_1(\bm),...,{\widehat
  \by}_{2r+1}(\bm))\in \cA[t]^{2r+1}$ defined as follows:
\begin{equation}\label{noncoweights}
\begin{array}{lll}
{\widehat{\by}}_{2i-1}=t(\by_{2i-1}+\by_{2i}),&
{\widehat{\by}}_{2i}=t^2\by_{2i+1}\by_{2i},
&\hbox{if $i\in I_+$}; \\
{\widehat{\by}}_{2i-1}=t\by_{2i-1}-\by_{2i+1}^{-1}\by_{2i},&
{\widehat{\by}}_{2i}=\by_{2i+1}^{-1}\by_{2i},
&\hbox{if $i\in I_-$;}\\
{\widehat{\by}}_{2i-1}=t\by_{2i-1},& {\widehat{\by}}_{2i}=t\by_{2i},
&\hbox{otherwise}.\hfill \\
\end{array}
\end{equation}
Here, we have omitted the reference to the Motzkin path $\bm$. Note
that these are non-commutative versions of the variables $\widehat{y}$
defined in Equation \eqref{Gweights}.

Finally, we define the weight on the edge from $i$ to $j$ of $G_\bm$ to be
$\bw_{i,j}$, where
\begin{equation*}
\bw_{i,i}={\widehat{\by}}_{2i-1}(\bm), \quad 
\bw_{i,i+1}=1, \quad 
\bw_{i+1,i}={\widehat{\by}}_{2i}(\bm).
\end{equation*}

\subsubsection{Non-commutative partition functions}
Let $\mathcal P_{a,b}$ be the set of paths on $G_\bm$ from vertex $a$
to vertex $b$. Each path $p\in \cP_{a,b}$ has a non-commutative weight ${\mathbf wt}(p)\in
\mathcal A$ defined like in Section \ref{commupf} as the product of step weights, in the 
order in which the steps are taken. The partition function, or generating
function, of paths from $a$ to $b$ is defined to be
\begin{equation}
\bZ(G_\bm)_{a,b}=\sum_{p\in \cP_{a,b}} {\mathbf wt}(p).
\end{equation}

\subsubsection{Non-commutative continued fractions}
\begin{defn}
For any $\bx=(\bx_1,...,\bx_{2r+1})$ with $\bx_i\in \cA[t]$, we 
define the non-commutative Jacobi-like finite continued
fraction $\bJ^{(r)}(\bx)$ inductively by $\bJ^{(1)}(\bx_{1})=(\bone -\bx_1)^{-1}$ and
\begin{equation}\label{NCJfrac}
\bJ^{(k)}(\bx_1,\bx_2,...,\bx_{2k+1})=
\left({1}-\bx_1 -\bJ^{(k-1)}(\bx_3,\bx_4,...,\bx_{2k+1})\bx_2\right)^{-1}.
\end{equation}
\end{defn}

As in the commutative case, we may consider such a
fraction as an 
element of $\cA[[t]]$, that is, as a power series in $t$.

It is clear that $\bZ(\bm)_{1,1}=\bJ^{(r)}({\widehat{\by}}(\bm))$
in terms of the weights \eqref{noncoweights}.
We also define the generating function $\bF_\bm(t)$, in analogy with
the commutative case, to be
\begin{equation} \label{ncJfrac} \bF_\bm(t)=1+t\,
  \bJ^{(r)}\big({\widehat{\by}}_1(\bm),...,{\widehat{\by}}_{2r+1}(\bm)\big)\,
  \by_1(\bm)
\end{equation}
\begin{lemma}\label{ncpositivity}
The coefficient of $t^n$ in the power series expansion of
$\bF_\bm(t)$ is a Laurent polynomial in $\by(\bm)$ with
non-negative integer coefficients.  
\end{lemma}
\begin{proof}
  The proof is identical to the proof in the commutative case
  (c.f. Lemma \ref{positinvol}), since weights respect the order of steps 
in the paths.   Suppose $i\in I_-$, so that $\bw_{i,i}$
  has a negative term.  Decompose the loop $i\to i$ with weight
  $\bw_{ii}={\widehat{\by}}_{2i-1}$ into two separate loops from $i$ to $i$,
  one, $l_i^+$, with weight $t\by_{2i-1}$ and the other, $l_i^-$, with weight
  $-\by_{2i+1}^{-1}\by_{2i}=-\bw_{i+1,i}$. Using the involution
  $\varphi$ defined in Lemma \ref{positinvol} on paths on the
  resulting graph $G'_\bm$, we again conclude that we have a
  cancellation of all paths $p\neq \varphi(p)$ in the partition function.
\end{proof}

An alternative proof of Lemma \ref{ncpositivity} uses the structure of non-commutative continued
fractions. 
One may simply rearrange the continued fraction expression
\eqref{ncJfrac} (from bottom to top) whenever $i\in I_-$, by
use the following non-commutative version of the identity
\eqref{identoneco}.  For any $\ba,\bb,\bc,\bu \in \cA$:
\begin{equation}\label{rearponc} 
\ba -\bc^{-1}\bb +({1}-\bc -\bu)^{-1} \bc^{-1}\bb =\ba+({1}-\bc -\bu)^{-1}(\bb+\bu\bc^{-1}\bb)
\end{equation}

Once again, this is done at the cost of introducing multiple branching
of the new fraction, since $\bu$ now appears in two places. Moreover,
new ``redundant" weights appear, from the new term
$\bu\bc^{-1}\bb$. These are, however, positive Laurent monomials in
$\by(\bm)$.

For use below, we note that we also have an analogue of the identity
\eqref{identwoco}:
For any
$\ba,\bb,\bc,\bu\in \cA$,
\begin{equation}\label{otherearnc}
\ba+ \bb +({1}-\bc -\bu)^{-1} \bc \bb =\ba +\big({1}-({\bf
  1}-\bu)^{-1}\bc\big)^{-1} \bb
\end{equation}

\subsubsection{Canonical form of continued fractions}\label{canoninc}
As in the commutative case, the fraction $\bF_\bm(t)$ can be rearranged
into a ``canonical" form, a (finite, multiply branching) continued
fraction whose coefficients are all Laurent monomials in 
$\by$, with coefficients equal to $t$. 

This is done as follows. From the bottom to the top level of the
fraction ($i=r,r-1,...,1$), we rearrange the fraction whenever
$m_{i}\neq m_{i-1}$. Let
$(\ba,\bb,\bc)=t(\by_{2i-1},\by_{2i},\by_{2i+1})$ and use \eqref{rearponc} if $i\in I_-$,
\eqref{otherearnc} if $i\in I_+$, to rearrange the fraction. This includes a last
``level zero" step of the continued fraction as well (case $i=0$),
which corresponds to the conventions $\by_{-1}=0$ and $\by_0={1}$.

The resulting ``canonical" continued fraction is interpreted as the
generating function for paths from and to the root vertex on the graph
$\Gamma_\bm$ (see Appendix \ref{appendixa}), but now with
non-commutative skeleton weights $t\by_i(\bm)$. 

In particular, the non-commutative
``redundant" weights $\bw_{i,j}$, $i-j>1$, generalizing
\eqref{redunco} now read, as a consequence of iterations of
\eqref{rearponc}:
\begin{equation}\label{redunc}
\bw_{i,j}=t \left(\prod_{a=i-1}^{j+1} \by_{2a}
\by_{2a-1}^{-1}\right)\by_{2j}
\end{equation}
where the product is taken in decreasing order, $a=i-1,i-2,...,j+1$.
\begin{example}\label{fundnc}
In the case of the fundamental Motzkin path $\bm_0=(0,0,...,0)$, we have
\begin{equation}\label{nccfzero}
\bF_{\bm_0}(t)=1+ t \bJ^{(r)}(t\by_1,t\by_2,...,t\by_{2r+1})\by_1 
\end{equation}
with $\by_i\equiv \by_i(\bm_0)$. The canonical form is obtained by applying \eqref{otherearnc}
at level zero, with $\ba=0$, $\bb=1$ and $\bc=t\by_1$, with the result:
\begin{equation}\label{nczerocf}
\bF_{\bm_0}(t)=\left(1- t\big(1-t\big(1-t\by_3-t(\cdots (1-t\by_{2r-1}
-t(1-t\by_{2r+1})^{-1}\by_{2r})^{-1}\cdots\big)^{-1}\by_4\big)^{-1}\by_2\big)^{-1}\by_1\right)^{-1}
\end{equation}
This is the generating function for paths on the graph $\Gamma_{\bm_0}$ of Example
\ref{exzero}, from vertex $1'$ to
itself, with the non-commutative weights $\bw_{i,i'}=\bw_{i-1,i}=1$, $\bw_{i',i}=t\by_{2i-1}$, 
$\bw_{i,i-1}=t\by_{2i-2}$ for $i\geq 2$ and $\bw_{1',1}=1$, $\bw_{1,1'}=t\by_1$.
\end{example}

\begin{example}\label{stilncex}
  In the case $\bm=\bm_1=(0,1,2,...,r-1)$, we have $m_i=m_{i-1}+1$
  for all $i$.  The corresponding continued fraction reads, with
  $\by_i\equiv \by_i(\bm_1)$:
\begin{eqnarray}
&&\ \ \ \bF_{\bm_1}(t)= \label{stielncf} \\
&& \!\!\!\!\!\!\!\!\!\!\!\!\!
1+t\left({1}-t\by_1-t\by_2-t^2\big(..(1-t\by_{2r-1}-t\by_{2r}-t^2(1-t\by_{2r+1})^{-1}
\by_{2r+1}\by_{2r})^{-1}..\big)^{-1}\by_2\by_1\right)\by_1 \nonumber \\
&&\quad \qquad =\left(1-t\big(1 -t(..(1-t(1-t \by_{2r+1})^{-1}\by_{2r})^{-1}..)^{-1}\by_2\big)^{-1}\by_1\right)^{-1}
\nonumber 
\end{eqnarray}
where the first line is the Jacobi-like expression \eqref{ncJfrac} and
the second line has the canonical form.
\end{example}
The second identity in Example \ref{stilncex} is the natural non-commutative
generalization of the Stieltjes continued fraction expansion.  In the
following, we will refer to the case $\bm=\bm_1$ as the ``Stieltjes
point".

The graph $\Gamma_{\bm_1}$ associated to this example is a chain of
$2r+2$ vertices labeled $1,2,...,2r+2$ (see Appendix \ref{appendixa}),
with weights $\bw_{i,i+1}={1}$ and $\bw_{i+1,i}=t \by_i$,
$i=1,2,...,2r+1$.

\begin{example}
In the case $r=3$ and $\bm=(0,1,0)$, we have, with $\by_i\equiv \by_i(0,1,0)$:
\begin{eqnarray*}
&&\ \ \  \bF_{(0,1,0)}(t)=\\
&& \!\!\!\!\!\!\!\!\!\!\!\!\!1+t\left(1 -t(\by_1+\by_2)-t^2\big(1-t\by_3+\by_5^{-1}\by_4
-(1 -t\by_5-t(1 -t \by_7)^{-1}\by_6)^{-1}\by_5^{-1}\by_4 \big)^{-1}\by_3\by_2\right)\by_1\\
&&\!\!\!\!\!\!\!\!\!\!\!\!\! = \left(1 -t\big(1 -t\big(1 -t(1 -t (1 -t \by_5-t(1-t\by_7)^{-1}\by_6)^{-1}
(\by_4+(1-t\by_7)^{-1}\by_6\by_5^{-1}\by_4))^{-1}\by_3)^{-1}\by_2\big)^{-1}\by_1\right)^{-1}
\end{eqnarray*}
where the first expression is the Jacobi-like form \eqref{ncJfrac},
and the second is the canonical form, obtained by first applying
\eqref{rearponc} with $\ba=t\by_3$, $\bb=t\by_4$ and $\bc=t\by_5$
while $\bu=1 -t \by_5-t(1-t\by_7)^{-1}\by_6$, then \eqref{otherearnc}
with $\ba=t\by_1$, $\bb=t\by_2$, $\bc=t\by_3$, and finally
\eqref{otherearnc} at ``level zero", with $\ba=0$, $\bb=1$ and
$\bc=t\by_1$.  We see the emergence of the ``redundant" weight
$t\by_6\by_5^{-1}\by_4$. See the corresponding graph $\Gamma_{(0,1,0)}$ 
in Appendix \ref{appendixa}.
\end{example}

\subsection{Mutations}In the non-commutative case, we do
not have an evolution equation for cluster variables. Instead, we
start by defining mutations of the weights $\by$. We will then show
that these mutations have the same effect on
the generating functions $F_\bm(t)$, considered as continued
fractions, as the non-commutative analogues of the rearrangement
Lemmas \ref{rear1lem} and \ref{rear2lem}.

We consider the set of Motzkin paths $\bm\in \cM$. We use the same
notion of mutations on Motzkin paths as in the commutative
case. Recall that mutations act on Motzkin paths as follows:
$\mu_i(\bm) = \bm'$ where $\bm'$ differs from $\bm$ only in that
$m_i'=m_i+1$. In order to obtain any Motzkin path $\bm\in \cM$ via a
sequence of mutations, we
need only consider the two types of mutations described in cases (i)
and (ii) of Section \ref{mutacosec}.
\subsubsection{Mutations of weights} 

We now define the action of mutations on $\cA$, as follows.
\begin{defn}\label{mutationsonncy}
Let $\mu_i(\bm)=\bm'$. Define $\by':=\by(\bm')$ as a function of
$\by:=\by(\bm)$ as follows:
\begin{eqnarray}
\hbox{\bf Case (i):} & \left\{\begin{array}{l} 
\by'_{2i-1} = \by_{2i-1}+\by_{2i},\\
\by'_{2i} = \by_{2i+1}\by_{2i} (\by_{2i-1}+\by_{2i})^{-1},\\
\by'_{2i+1} = \by_{2i+1}\by_{2i-1}(\by_{2i-1}+\by_{2i})^{-1}.
\end{array} 
\right.\label{caseonenc}\\
\hbox{\bf Case (ii):} & \left\{\begin{array}{l} 
\by'_{2i-1} = \by_{2i-1}+\by_{2i},\\
\by'_{2i} = \by_{2i+1}\by_{2i} (\by_{2i-1}+\by_{2i})^{-1},\\
\by'_{2i+1} = \by_{2i+1}\by_{2i-1}(\by_{2i-1}+\by_{2i})^{-1},\\
\by'_{2i+2} = \by_{2i+2}\by_{2i-1}(\by_{2i-1}+\by_{2i})^{-1}.\\
\end{array} 
\right.\label{casetwonc}
\end{eqnarray}
with all other $\by'_j=\by_j$ unchanged.
\end{defn}

Our main goal in this section is the proof of Theorem \ref{defymutnc},
concerning the action of mutations on the generating functions $\bF_\bm(t)$.
Towards this end, we prove some non-commutative analogues of the
continued fraction rearrangement lemmas of Section \ref{sectiontwo}.

\subsubsection{Mutations as non-commutative rearrangements}

In addition to the two identities \eqref{rearponc} and \eqref{otherearnc},
we have the two following local rearrangement lemmas for 
non-commutative continued fractions.

\begin{lemma}\label{rear1lemnc}
For all $\ba,\bb\in \cA$ we have
\begin{equation}\label{rear1nc}
{1}+({1}-\ba -\bb)^{-1} \ba =\left( {1}-({1}-\bb)^{-1} \ba \right)^{-1}
\end{equation}
\end{lemma}
\begin{proof} By direct calculation. This is the ``level zero" application 
of identity \eqref{otherearnc}.
\end{proof}

\begin{lemma}\label{rear2lemnc}
For all $\ba,\bb,\bc,\bd \in \cA$, with $\ba+\bb$ invertible, we have
\begin{equation}\label{rear2nc}
\ba+\left({1}-(1-\bd)^{-1}\bc \right)^{-1} \bb =
\left({1}-\bd-({1}-\bc')^{-1}\bb'\right)^{-1} \ba'
\end{equation}
where $\ba',\bb',\bc'$ are defined by
\begin{equation}\label{defabcnc}
\ba'=\ba+\bb \qquad \bb'=\bc \bb (\ba+\bb)^{-1} \qquad \bc'=\bc \ba (\ba+\bb)^{-1}
\end{equation}
\end{lemma}
\begin{proof} By direct calculation.
\end{proof}

Together with the identities \eqref{rearponc}-\eqref{otherearnc}, these
two lemmas may be used to rearrange the non-commutative continued
fraction expression of $\bF_\bm(t)\mapsto \bF_{\bm'}(t)$
\eqref{ncJfrac}.


Define the non-commutative version of the fractions $J_{i}$ as in
the proof of Lemma \ref{rear2lem}.
\begin{defn} The non-commutative finite continued fraction
  $\bJ_i(\bm)$ is defined inductively as follows:
$$
\bJ_{i}(\bm) = \widehat{\by}_{2i-1} +
(1-\bJ_{i+1}(\bm))^{-1}\widehat{\by}_{2i},\qquad
\bJ_{r+1} (\bm) = \widehat{\by}_{2r+1}(\bm).
$$
\end{defn}
Note that $\bJ_0(\bm)=\bJ^{(r)}({\widehat{\by}}(\bm))$,
and $\bF_\bm(t) = 1+t \bJ_0(\bm)\by_1$, with the convention that
${\widehat{\by}}_{-1}=0$ and ${\widehat{\by}}_0=1$.

Let us fix $i$, and define:
$$ \ba=t\by_{2i-3},  \bb=t\by_{2i-2},  \bc=t\by_{2i-1},
\bd=t\by_{2i},  \be=t\by_{2i+1}, 
\bbf=t\by_{2i+2}, \bu=(1 -\bJ_{i+2}(\bm))^{-1},$$ 
where $\by_i:=\by_i(\bm)$. Similarly, we use the same notation for the primed
counterparts $\by_i':=\by_i(\bm')$.  

Starting from $\bF_\bm(t)$ in Jacobi-like form
\eqref{ncJfrac}, we apply fraction rearrangements on the part $\bJ_{i-1}$ of
$\bF_\bm(t)$. 
\begin{itemize}
\item{\bf Case (i)}: We first transform the Jacobi-like form into the local canonical form 
for $\bJ_{i-1}(\bm)$ by applying \eqref{otherearnc} at the second level, namely:
\begin{eqnarray*}
\bJ_{i-1}(\bm)&=&\ba +\big(1 -\bc-\bd -(1 -\be -\bu\bbf)^{-1}\be \bd \big)^{-1}\bb \\
&=&\ba +\big(1 -\bc -(1 -(1 -\bu\bbf)^{-1}\be)^{-1}\bd  \big)^{-1}\bb
\end{eqnarray*}
Next, we apply Lemma \ref{rear2lemnc} to the second level, namely:
\begin{equation*}
\bJ_{i-1}(\bm)=\ba +\big(1 -(1 -(1 -\be'-\bu\bbf)^{-1} \bd')^{-1}\bc'  \big)^{-1}\bb
\end{equation*}
with $\bc'=\bc+\bd$, $\bd'=\be\bd(\bc')^{-1}$ and $\be'=\be\bc (\bc')^{-1}$. 
Finally, we put this back in Jacobi-like form, by applying \eqref{otherearnc} backwards to the first level, namely:
\begin{equation*}
\bJ_{i-1}(\bm')=\bJ_{i-1}(\bm)=\ba+\bb +\big(1 -\bc'-(1 -\be'-\bu\bbf)^{-1} \bd'\big)^{-1}\bc' \bb
\end{equation*}
\par
\begin{center}
\begin{tabular}{ccc}
\psset{unit=2mm,linewidth=.4mm,dimen=middle}
\begin{pspicture}(0,0)(25,14)
\psline(0,0)(0,5)(5,10)
\rput(0,0){\pscircle*[linecolor=red](0,0){.7}}
\rput(0,5){\pscircle*[linecolor=red](0,0){.7}}
\rput(5,10){\pscircle*[linecolor=red](0,0){.7}}
\rput(4,0){$i-1$}
\rput(2,5){$i$}
\rput(9,10){$i+1$}
\psline(20,0)(20,12)
\psline[linestyle=dotted](20,12)(20,14)
\rput(24,0){$i-1$}
\rput(22,5){$i$}
\rput(24,10){$i+1$}
\multips(17.75,0)(0,5){3}{\pscircle(0,0){2}}
\multips(20,0)(0,5){3}{\pscircle*[linecolor=red](0,0){.7}}
\rput(15,0){\blue\small $\ba$}
\rput(12,5){\blue \small $\bc+\bd$}
\rput(15.75,5){$\times$}
\rput(15,10){\blue \small $\be$}
\rput(21,2.5){\blue \small $\bb$}
\rput(22,7.5){\blue\small $\be\bd$}
\rput(21,12.5){\blue\small $\bbf$}
\rput(20,7.5){$\times$}
\end{pspicture}
&\psset{unit=2mm,linewidth=.4mm,dimen=middle}
\begin{pspicture}(0,0)(10,14)
\psline{->}(2,6)(7,6)
\end{pspicture}
 &
\psset{unit=2mm,linewidth=.4mm,dimen=middle}
\begin{pspicture}(0,0)(30,14)
\psline(0,0)(5,5)(5,10)
\rput(0,0){\pscircle*[linecolor=red](0,0){.7}}
\rput(5,5){\pscircle*[linecolor=red](0,0){.7}}
\rput(5,10){\pscircle*[linecolor=red](0,0){.7}}
\rput(4,0){$i-1$}
\rput(7,5){$i$}
\rput(9,10){$i+1$}
\psline(20,0)(20,12)
\psline[linestyle=dotted](20,12)(20,14)
\rput(24,0){$i-1$}
\rput(22,5){$i$}
\rput(24,10){$i+1$}
\multips(17.75,0)(0,5){3}{\pscircle(0,0){2}}
\multips(20,0)(0,5){3}{\pscircle*[linecolor=red](0,0){.7}}
\rput(12,0){\blue\small $\ba+\bb$}
\rput(15.75,0){$\times$}
\rput(14.5,5){\blue \small $\bc'$}
\rput(14.5,10){\blue \small $\be'$}
\rput(22.5,2.5){\blue \small $\bc'\bb$}
\rput(20,2.5){$\times$}
\rput(21.5,7.5){\blue\small $\bd'$}
\rput(21,12.5){\blue\small $\bbf$}
\end{pspicture}
\end{tabular}
\end{center}
\vskip .5in

\item{\bf Case (ii)}: We apply directly Lemma \ref{rear2lemnc} to the second level of $\bJ_{i-1}(\bm)$,
namely:
\begin{eqnarray*}
\bJ_{i-1}(\bm)&=&\ba +\big(1 -\bc-\big(1 -\be -\bu\bbf \big)^{-1}\bd \big)^{-1}\bb\\
&=&\ba +\big(1 -(1 - (1 -\be'-\bu\bbf')^{-1}
(\bd'+\bu\bbf'(\be')^{-1}\bd'))^{-1}\bc'\big)^{-1}\bb\\
\end{eqnarray*}
with $\bc'=\bc+\bd$, $\bd'=\be\bd(\bc')^{-1}$, $\be'=\be\bc(\bc')^{-1}$ and $\bbf'=\bbf\bc(\bc')^{-1}$. 
To bring this back to Jacobi-like form, we first apply \eqref{rearponc} backward on the third level, and
then \eqref{otherearnc} backward on the first one:
\begin{eqnarray*}
\bJ_{i-1}(\bm)&=&\ba +\big(1 -(1 +(\be')^{-1}\bd'- 
(1 -\be'-\bu\bbf')^{-1}(\be')^{-1}\bd')^{-1}\bc'\big)^{-1}\bb\\
&=&\ba +\bb +\big(1 -\bc' +(\be')^{-1}\bd'- 
(1 -\be'-\bu\bbf')^{-1}(\be')^{-1}\bd'\big)^{-1}\bc'\bb
\end{eqnarray*}
\par
\begin{center}
\begin{tabular}{ccc}

\psset{unit=2mm,linewidth=.4mm,dimen=middle}
\begin{pspicture}(0,0)(25,18)
\psline(0,0)(0,10)
\multips(0,0)(0,5){3}{\pscircle*[linecolor=red](0,0){.7}}

\rput(4,0){$i-1$}
\rput(2,5){$i$}
\rput(4,10){$i+1$}
\psline(20,0)(20,15)
\psline[linestyle=dotted](20,15)(20,18)
\rput(24,0){$i-1$}
\rput(22,5){$i$}
\rput(24,10){$i+1$}
\multips(17.75,0)(0,5){3}{\pscircle(0,0){2}}
\multips(20,0)(0,5){3}{\pscircle*[linecolor=red](0,0){.7}}
\rput(15,0){\blue\small $\ba$}
\rput(15,5){\blue \small $\bc$}
\rput(15,10){\blue \small $\be$}
\rput(21,2.5){\blue \small $\bb$}
\rput(21,7.5){\blue\small $\bd$}
\rput(21,12.5){\blue\small $\bbf$}
\end{pspicture}
&
\psset{unit=2mm,linewidth=.4mm,dimen=middle}
\begin{pspicture}(0,0)(10,18)
\psline{->}(2,6)(7,6)
\end{pspicture}
&
\psset{unit=2mm,linewidth=.4mm,dimen=middle}
\begin{pspicture}(0,0)(30,18)
\psline(0,0)(5,5)(0,10)
\rput(0,0){\pscircle*[linecolor=red](0,0){.7}}
\rput(5,5){\pscircle*[linecolor=red](0,0){.7}}
\rput(0,10){\pscircle*[linecolor=red](0,0){.7}}
\rput(4,0){$i-1$}
\rput(7,5){$i$}
\rput(4,10){$i+1$}
\psline(26,0)(26,15)
\psline[linestyle=dotted](26,15)(26,18)
\rput(30,0){$i-1$}
\rput(28,5){$i$}
\rput(30,10){$i+1$}
\multips(23.75,0)(0,5){3}{\pscircle(0,0){2}}
\multips(26,0)(0,5){3}{\pscircle*[linecolor=red](0,0){.7}}
\rput(17.5,0){\blue\small $\ba+\bb$}
\rput(21.75,0){$\times$}
\rput(15,5){\blue \small $\bc'-\be'^{-1}\bd'$}
\rput(21.75,5){$=$}
\rput(20.5,10){\blue \small $\be'$}
\rput(28.5,2.5){\blue \small $\bc'\bb$}
\rput(26,2.5){$\times$}
\rput(30.5,7.5){\blue\small $\be'^{-1}\bd'$}
\rput(26,7.5){$=$}
\rput(27,12.5){\blue\small $\bbf'$}
\end{pspicture}\end{tabular}
\end{center}
\vskip .5in
\end{itemize}
Note that in both cases, the relation of the primed variables to the unprimed
variables is the same as the relation of the mutated variables $y_j'$
to $y_j$ in Definition \ref{mutationsonncy}.

Clearly, for $i>1$, we always have $\bJ_0(\bm')=\bJ_0(\bm)$,
or equivalently $\bF_\bm(t)=\bF_{\bm'}(t)$.
When $i=1$, we have $\ba=0$, $\bb=1$ and in all cases $\bc'=\by_1(\bm')$, so the final formula
amounts to $\bJ_0(\bm)=1+t \bJ_0(\bm')\by_1(\bm')$, or equivalently
$\bF_\bm(t)=1+t\bF_{\bm'}(t)\by_1(\bm)$.

We therefore have the following theorem:
\begin{thm}\label{defymutnc}
  Under the mutation $\mu_i$ on $\cA$, $\mu_i: \by(\bm)\mapsto
  \by(\bm')$ of Definition \ref{mutationsonncy}, the function
  $\bF_\bm(t)$ is transformed as:
\begin{eqnarray}
i=1:\qquad \bF_\bm(t)&=&1+t\bF_{m'}(t)\by_1(\bm)\label{rrt}, \\
i>1:\qquad  \bF_\bm(t)&=&\bF_{\bm'}(t)\label{eqf}.
\end{eqnarray}
\end{thm}


\subsection{An explicit expression for the weights}

Although we do not have evolution equations of the form \eqref{Qsys}
in the non-commutative case, we make a definition of the variables
$\bR_n$ as follows.

\begin{defn}
Let $R_nR_0^{-1}$ be defined as the coefficients in the power series
expansion of $F_{\bm_0}(t)$, as follows:
\begin{equation}
\bF_{m_0}(t)=\sum_{n=0}^\infty t^n \bR_n\bR_0^{-1}.
\end{equation}
\end{defn}
In particular, $\by_1(\bm_0) = R_1 R_0^{-1}$. To make the contact with
the commutative case of Section 2, one should compare $\bR_n$ with
$R_{1,n}$ of that section.

In this section, we relate the weights $\by(\bm)$ to the coefficients
of the generating series The main result is Theorem \ref{ncweightcalc}
below, which gives an explicit expression for all $\by(\bm)$ in terms
of the $\bR_n$'s.

\subsubsection{Stieltjes point and quasi-determinants}

The generating series $\bF_{m_0}(t)$ may be written as the
non-commutative ``Stieltjes" continued fraction $\bF_{\bm_1}(t)$ in
its canonical form \eqref{stielncf} at the Stieltjes point
$\bm=\bm_1$.  Since the sequence of forward mutations from $\bm_0$ to
$\bm_1$ only involves mutations $\mu_i$ with $i>1$, the generating
series is preserved at each step, and $\bF_{m_0}=\bF_{m_1}$.

The coefficients $\by_i(\bm_1)$ of the non-commutative Stieltjes
continued fraction \eqref{stielncf} are expressed in terms of the
$\bR_n$'s via a non-commutative generalization of the celebrated
Stieltjes formula (Theorem \ref{stilnc} below). For this, we need a
number of definitions.  In the following, $A$ denotes a square
$n\times n$ matrix with coefficients $a_{i,j}\in \cA$.

We recall the notion
of quasi-determinant \cite{GGRW}.
\begin{defn}[Quasi-determinant]\label{quasidet}
  The quasi-determinant of $A$ with respect to the entry $(p,q)$ is
  defined inductively by the formula:
\begin{eqnarray*}
\vert A\vert_{p,q}&=& 
a_{p,q}-\sum_{i\neq p,j\neq q} a_{p,j} \left(\vert A^{p,q}\vert_{i,j}\right)^{-1} a_{i,q}\\
\vert(a_{1,1})\vert_{1,1}&=&a_{1,1}
\end{eqnarray*}
where $A^{p,q}$ denotes the $n-1\times n-1$ matrix obtained from $A$ by erasing row
$p$ and column $q$, while keeping the row and column labeling of $A$.
\end{defn}

We refer to \cite{GGRW} for a review of the many properties of the quasi-determinant. 
For now we will only need the following:
\begin{prop}[\cite{GGRW}, Proposition 1.4.6]\label{vaniprop}
The three following statements are equivalent:
\begin{itemize}
\item (i) $\vert A\vert_{p,q}=0$ 
\item (ii) the $p$-th row of $A$ is a left linear combination of the other rows
\item (iii) the $q$-th column of $A$ is a right linear combination of the other columns
\end{itemize}
\end{prop}

Note that the quasi-determinant is not a generalization of the determinant, but rather of
the ratio of the determinant by a minor:
\begin{prop}\label{quasicomm}
If $\cA$ is commutative, and the minor $|A^{p,q}|\neq 0$, 
then the quasi-determinant $|A|_{p.q}$
reduces to a ratio of the determinant of $A$ by its $p,q$ minor:
\begin{equation}
\vert A \vert_{p,q} = {\vert A \vert \over \vert A^{p,q}\vert}
\end{equation}
\end{prop}

Finally, we also need:
\begin{defn}[Quasi-Wronskian]\label{ncwdef}
The quasi-Wronskians of the sequence $(\bR_M)_{M\in \Z}$ are defined 
for $i\geq 1$ and $n\in \Z$ as the following $i\times i$ quasi-determinants:
\begin{equation} \label{wronskdefnc}
\Delta_{i,n}=
\left\vert \begin{matrix} 
\bR_{n+i-1} & \bR_{n+i-2} & \cdots & \cdots & \bR_{n} \\
\bR_{n+i-2} & \bR_{n+i-3} & \cdots & \cdots & \bR_{n-1}\\
\vdots & \vdots &  & & \vdots \\
\bR_n & \bR_{n-1} & \cdots & \cdots & \bR_{n-i+1}
\end{matrix} \right\vert_{1,1}
\end{equation}
\end{defn}

\begin{example}\label{initdelt}
We have for $i=1,2,3$:
\begin{eqnarray*}
\Delta_{1,n}&=&\bR_n\\
\Delta_{2,n}&=&\bR_{n+1}-\bR_n \bR_{n-1}^{-1} \bR_n \\
\Delta_{3,n}&=&\bR_{n+2}-\bR_n(\bR_{n-2}-\bR_{n-1}\bR_n^{-1}\bR_{n-1})^{-1}\bR_n
-\bR_{n+1}(\bR_{n-1}-\bR_{n-2}\bR_{n-1}^{-1}\bR_n)^{-1}\bR_n\\
&&-\bR_n(\bR_{n-1}-\bR_n\bR_{n-1}^{-1}\bR_{n-2})^{-1}\bR_{n+1}
-\bR_{n+1}(\bR_n-\bR_{n-1}\bR_{n-2}^{-1}\bR_{n-1})^{-1}\bR_{n+1}
\end{eqnarray*}
\end{example}

We are now ready for the non-commutative Stieltjes formula \cite{GKLLRT}. We start from
the (infinite) non-commutative continued fraction $\bS(t;\bx)$
with coefficients $\bx_i\in \cA$, $i\in \Z_{>0}$.
It is defined iteratively by
\begin{eqnarray}
\bS(t;\bx_1)&=&({1}-t \bx_1 )^{-1}\nonumber \\
\bS(t;\bx_1,\bx_2,...,\bx_k)&=&\bS(t;\bx_1,...,\bx_{k-2},({1}-t \bx_k)^{-1}\bx_{k-1})\nonumber \\
\bS(t;\bx)&=&\lim_{k\to\infty} \bS(t;\bx_1,\bx_2,...,\bx_k)\label{ncst}
\end{eqnarray}

\begin{thm}\label{stilnc}
The coefficients $\bx_i$ of the non-commutative continued fraction $S(t;\bx)$
are related to the coefficients of the corresponding power series 
$\bS(t;\bx)=\sum_{n\geq 0}t^n \bR_n\bR_0^{-1}$ via:
\begin{eqnarray}
\bx_{2i-1}&=&\Delta_{i,i}\Delta_{i,i-1}^{-1} \label{oddxnc} \\
\bx_{2i}&=&\Delta_{i+1,i}\Delta_{i,i}^{-1} \label{evenxnc}
\end{eqnarray}
for all $i>0$.
\end{thm}
\begin{proof} See Ref. \cite{GKLLRT}, Corollary 8.3.
\end{proof}

Comparing $\bS(t;\bx)$ \eqref{ncst} with the finite canonical continued fraction at the Stieltjes point 
$\bF_{\bm_1}(t)$ \eqref{stielncf},
we see that the latter is a {\it finite truncation} of the former,
which may be implemented by imposing that $\bx_{2r+2}=0$, while $\bx_i=\by_i(\bm_1)$ for
$i=1,2,...,2r+1$. From
Theorem \ref{stilnc}, this finite truncation condition is equivalent to the equation
$\Delta_{r+2,r+1}=0$, and Theorem \ref{stilnc} gives us
an evaluation of the non-commutative weights at the 
Stieltjes point in terms of the $\bR_n$'s. In fact, we have:

\begin{thm}\label{maineq}
The coefficients $\bR_n$ of the power series expansion
of the finite continued fraction $\bF_{\bm_1}(t)$ satisfy:
\begin{equation}\label{nctrunc}
\Delta_{r+2,n}=0 \quad (n>r) 
\end{equation}
\end{thm}
\begin{proof}
As a consequence of the definition \eqref{stielncf}, we may write (with $\by_i\equiv\by_i(\bm_1)$):
$$\bF_{\bm_1}(t)=P(t;\by_1,...,\by_{2r+1})^{-1}Q(t;\by_2,...,\by_{2r+1})=P(t)^{-1}Q(t)$$
with polynomials $P,Q\in \cA[t]$, such that $P(t),Q(t)={1}+O(t)$
and deg$(P)=r+1$ and deg$(Q)=r$. 
Indeed, from the inductive definition \eqref{ncst} we get
\begin{equation}\label{recursionP}
  P(t;\by_1,...,\by_{s+1})=P(t;\by_1,...,\by_{s})-t
  \by_{s+1}P(t;\by_1,...,\by_{s-1})
\end{equation} 
while $P(t;\by_1)=1-t \by_1$, and $Q(t;\by_2,...,\by_{s})=P(t;0,\by_2,...,\by_s)$.
Writing $P(t;\by_1,...,\by_{2r+1})=\sum_{i=0}^{r+1} (-t)^i  P_i $,
we conclude that the $\bR_n$'s satisfy a linear recursion relation of the form
\begin{equation}\label{recurnc}
\sum_{i=0}^{r+1} (-1)^iP_i \bR_{n+r+1-i}=0
\end{equation}
for all $n\geq 0$, with $P_0={1}$ and
$P_{r+1}=\by_{2r+1}\by_{2r-1}\cdots \by_1$. 
This implies that the first row of the Wronskian matrix of \eqref{wronskdefnc}
is a left linear combination of its other rows, and the theorem follows from Proposition
\ref{vaniprop}.
\end{proof}

\begin{remark}[Conserved quantities]\label{consrem}
  Note that the coefficients $P_m$, $m=0,1,...,r+1$ of the linear
  recursion relation \eqref{recurnc} are conserved quantities under
  the evolution of the weights of Definition \ref{mutationsonncy}. Indeed,
  throughout this evolution, the denominator $D_{\bm}^{-1}$ of the
  generating function $\bF_\bm=D_\bm^{-1}N_\bm$ (where both factors
  are equal to 1 modulo $t$) remains invariant. Therefore,
  $D_\bm=P(t;\by_1,...,\by_{2r+1})$ defined above, is the generating
  function for the $r+1$ conserved quantities $P_m$ (with
  $P_0=1$). Note, however, that in general these quantities do not
  commute with each other. 
\end{remark}

\begin{remark}\label{hardparticles} The conserved quantities
  $P_m$ have an interpretation in terms of non-commutative hard particle
  partition function on a chain, as in \cite{DFK3} in the commutative
  case. That is, $P_m$ is the sum over configurations of $m$ hard
  particles on the integer segment $[1,2r+1]$, with a weight
  $y_{i_m}y_{i_{m-1}}\cdots y_{i_1}$ per
  particle configuration occupying vertices $i_1,i_2,...,i_m$ such
  that $i_{j+1}-i_{j}>1$. This interpretation is a consequence of the
  recursion relation \eqref{recursionP}.

  More generally, one can show that for any Motzkin path $\bm$, the
  conserved quantities, when expressed in terms of the weights
  $\by_i(\bm)$, are (non-commutative) hard-particle partition
  functions on certain graphs $G_\bm$ introduced in \cite{DFK3}.
\end{remark}
For later use, it will be convenient to extend the definition of $\bR_n$ to $n<0$. This is 
easily done by extending the linear recursion relation \eqref{recurnc} to any $n\in \Z$.
For instance, this allows to define $\bR_{-1}=P_{r+1}^{-1}\sum_{i=0}^{r} (-1)^{r+i}P_i \bR_{r-i}$.
The main equation extends accordingly to 
\begin{equation}\label{maintrunc}
\Delta_{r+2,n}=0 \quad {\rm for}\ {\rm all}\ \ n\in \Z
\end{equation}

\begin{remark}
If $\cA$ is commutative, and $R_{i,n}$ are solutions to the
$A_r$ $Q$-system \eqref{Qsys}, 
we have  $\bR_{n}=R_{1,n}=R_n$.
Then according to Proposition \ref{quasicomm} we have:
$$ \Delta_{i,n}= {\det\left( (R_{n+i+1-a-b})_{1\leq a,b \leq i} \right)\over 
\det\left( (R_{n+i+1-a-b})_{2\leq a,b \leq i} \right)}={R_{i,n}\over R_{i-1,n-1}}$$
The finite truncation condition \eqref{maintrunc} is equivalent to $R_{r+2,n}=0$, for all $n\in\Z$.
\end{remark}

\subsubsection{Properties: from Pl\"ucker to Hirota}

In this section, we derive a recursion relation for the quasi-Wronskians, which
is a non-commutative analogue of the discrete Hirota equation. To this end, we must
first review a few useful properties of quasi-determinants \cite{GGRW}.

We recall the so-called heredity property of quasi-determinant
(see \cite{GGRW}, Theorem 1.4.3). 
This
states that we get the same result if we apply the definition 
\ref{quasidet} to any $r\times r$ decomposition
of the matrix $A$ into possibly rectangular blocks, 
viewing the blocks as some new non-commutative matrix elements, and then
evaluate the quasi-determinant of the resulting matrix. This holds only if
the $(p,q)$ entry of $A$ belongs to a square block, and all expressions
(involving inverses) are well defined.

Using this property, we immediately get the following two lemmas:

\begin{lemma}\label{here1} 
Let $A$ be a $k+1\times k+1$ matrix with last column such that $a_{i,k+1}=0$
for $i=1,2,...,k$. Then for $1\leq p,q\leq k$, we have:
\begin{equation}\vert A \vert_{p,q} =\vert A^{k+1,k+1}\vert_{p,q}\end{equation}
\end{lemma}
\begin{proof} By the heredity property, we decompose the matrix $A$ into a first $k\times k$ 
square block $B=A^{k+1,k+1}$ which contains the $(p,q)$ entry of $A$, 
a zero $k\times 1$ column vector $z$, 
a $1\times k$ row $u=(a_{k+1,1},...,a_{k+1,k})$ and the single matrix element
$x=a_{k+1,k+1}$, to get:
$\left\vert \begin{matrix}B & z \\
u & x \end{matrix}\right\vert_{p,q}=\vert B- z x^{-1} u\vert_{p,q}=\vert B\vert_{p,q}$.
\end{proof}

\begin{lemma}\label{here2} 
Let $A$ be a $k+1\times k+1$ matrix with first column such that $a_{i,1}=0$
for $i=1,2,...,k$. Then we have:
\begin{equation}
\vert A \vert_{k+1,1} =a_{k+1,1} \end{equation}
\end{lemma}
\begin{proof}
By the heredity property, we decompose the matrix $A$ into a zero
$k\times 1$ column $z$, the $k\times k$ block $B=A^{k+1,1}$, the single matrix element 
$x=a_{k+1,1}$, and the $1\times k$ row vector $u=(a_{k+1,2},...,a_{k+1,k+1})$:
$$ \vert A\vert_{k+1,1} = \left\vert \begin{matrix} z & B \\
x & u
\end{matrix}\right\vert_{k+1,1}=\vert x-u B^{-1} z \vert_{k+1,1}=x=a_{k+1,1}$$
\end{proof}

Quasi-determinants and minors are related via the so-called homological relations:
\begin{prop}[Homological Relations \cite{GGRW}, Theorem 1.4.2]
\begin{equation}\label{homo}
\left(\vert A^{k,j} \vert_{i,\ell} \right)^{-1} \vert A\vert_{i,j}
=-\left(\vert A^{i,j} \vert_{k,\ell} \right)^{-1} \vert A\vert_{k,j}
\end{equation}
\end{prop}


Quasi-determinants obey the following non-commutative generalizations of the Pl\"ucker 
relations. Let $A$ be a $k\times m$ matrix with elements in $\cA$, and $m>k$.
Let $A(i_1,...,i_n)$ denote the $k\times n$ matrix made of the columns
of $A$ indexed $i_1,...,i_n$. 

\begin{defn}
For any ordered sequence $I$ of column indices of $A$, we define the
left quasi-Pl\"ucker coordinates
\begin{equation}
q_{i,j}^I=\left(\vert A(i,I)\vert_{s,i}\right)^{-1} \vert A(j,I)\vert_{s,j}
\end{equation}
\end{defn}
Note that the definition is independent of the choice of row index $s$.
Moreover, if $j$ is among $I$, then $q_{i,j}^I=0$.

We may now write the generalized Pl\"ucker relations:

\begin{prop}[Non-commutative Pl\"ucker Relations, \cite{GGRW}, Theorem 4.4.2]
Let $M=m_1,m_2,...,m_{k-1}$ and $L=\ell_1,\ell_2,...,\ell_k$ be two 
(ordered) sequences
of column indices of $A$, then for any $i\not\in L$:
\begin{equation}\label{ncpluck}
\sum_{j\in L \setminus M} q_{i,j}^M \, q_{j,i}^{L\setminus \{j\}} =1
\end{equation}
\end{prop}

With these properties in mind, we now turn to our main theorem:

\begin{thm}[Discrete Non-commutative Hirota equation]\label{dischironc}
The quasi-Wronskians obey the following relations
\begin{equation}\label{nchiro}
\Delta_{i+1,n}=\Delta_{i,n+1}-\Delta_{i,n}(\Delta_{i,n-1}^{-1}-\Delta_{i-1,n}^{-1})\Delta_{i,n}
\end{equation}
for $i\geq 1$ and $n\in \Z$, with the convention $\Delta_{0,n}^{-1}=0$.
\end{thm}
\begin{proof}
Note that using the initial values $i=1,2$ of Example \ref{initdelt}, we get
$\Delta_{2,n}=\bR_{n+1}-\bR_n\bR_{n-1}^{-1}\bR_n
=\Delta_{1,n+1}-\Delta_{1,n}\Delta_{1,n-1}^{-1}\Delta_{1,n}$, which agrees with \eqref{nchiro}
for $i=1$,
given the convention for $\Delta_{0,n}^{-1}=0$. We now turn to the general proof of the identity \eqref{nchiro}
for $i>1$. Let us apply the Pl\"ucker relation \eqref{ncpluck} to the following
$i+1\times i+3$ matrix and column index sets: 
$$A=\begin{pmatrix}
{1} & \bR_{n+i} & \bR_{n+i-1}& \cdots & \bR_n & 0\\
0 & \bR_{n+i-1} & \bR_{n+i-2} & \cdots &\bR_{n-1} & 0\\
\vdots & \vdots & \vdots &  & \vdots & \vdots \\
0 & \bR_{n+1} & \bR_n & \cdots & \bR_{n-i+1} & 0\\
0 & \bR_n & \bR_{n-1} & \cdots & \bR_{n-i} & 1
\end{pmatrix}, \qquad \left\{ \begin{matrix}
M&=&(3,4,...,i+1,i+2)\\
L&=&(2,3,...,i+1,i+3)
\end{matrix}\right.$$
so that $L\setminus M=2,i+3$. We now compute the various
Pl\"ucker coordinates involved in the sum \eqref{ncpluck}.
Introducing
\begin{equation}
\label{defphitheta}\theta_{i+1,n}={\small \left\vert \begin{matrix} 
\bR_{n+i} & \bR_{n+i-1}& \cdots & \bR_{n+1} & 0\\
\bR_{n+i-1} & \bR_{n+i-2}& \cdots & \bR_{n} & 0\\
\vdots         & \vdots         & \vdots & \vdots & \vdots\\
\bR_{n+1} & \bR_n & \cdots & \bR_{n-i+2} & 0 \\
\bR_n & \bR_{n-1} & \cdots & \bR_{n-i+1} & {1}
\end{matrix}\right\vert_{1,1} },
\varphi_{i+1,n}={\small \left\vert \begin{matrix} 
0 & \bR_{n+i-1} & \bR_{n+i-2}& \cdots & \bR_n \\
0 & \bR_{n+i-2} & \bR_{n+i-3} & \cdots &\bR_{n-1}\\
\vdots & \vdots & \vdots &  & \vdots \\
0 & \bR_n & \bR_{n-1} & \cdots & \bR_{n-i+1}\\
{1} & \bR_{n-1} & \bR_{n-2} & \cdots & \bR_{n-i}
\end{matrix}\right\vert_{1,1}}
\end{equation}
we find
$$q_{1,2}^M=\Delta_{i+1,n}\qquad q_{2,1}^{L\setminus \{2\}}=\theta_{i+1,n}^{-1}\qquad 
q_{1,i+3}^M=\varphi_{i+1,n}\qquad 
q_{i+3,1}^{L\setminus \{i+3\}}=\varphi_{i+1,n+1}^{-1}
$$
so the Pl\"ucker relation turns into
\begin{equation}\label{debut}
\Delta_{i+1,n} \theta_{i+1,n}^{-1}+\varphi_{i+1,n}\varphi_{i+1,n+1}^{-1}=1
\end{equation}
By direct application of Lemma \ref{here1}, we compute $\theta_{i+1,n}=\Delta_{i,n+1}$.
To compute $\varphi_{i+1,n}$ we apply the homological relation \eqref{homo}
with $i=j=1$, $k=i+1$, $\ell=2$ and with the matrix $A$ with entries:
$a_{x,1}=\delta_{x,i+1}$, and $a_{x,y}=\bR_{n-x-y+i}$ for $y>1$.
We get $\varphi_{i+1,n}=\vert A\vert_{1,1}=-\vert A^{i+1,1}
\vert_{1,2}\Big(\vert A^{1,1}\vert_{i+1,2}\Big)^{-1}\vert A\vert_{i+1,1}$. 
The first factor is $\vert A^{i+1,1}\vert_{1,2}=\Delta_{i,n}$
by definition, while the last factor is $\vert A\vert_{i+1,1}=a_{i+1,1}={1}$
by applying Lemma \ref{here2} with $k=i$.
We are left with the task of computing $\vert A^{1,1}\vert_{i+1,2}=\psi_{i,n-1}$, where
$$ \psi_{i,n} =\left\vert \begin{matrix} 
\bR_{n+i-1} & \bR_{n+i-2} & \cdots & \cdots & \bR_{n} \\
\bR_{n+i-2} & \bR_{n+i-3} & \cdots & \cdots & \bR_{n-1}\\
\vdots & \vdots &  & & \vdots \\
\bR_n & \bR_{n-1} & \cdots & \cdots & \bR_{n-i+1}
\end{matrix} \right\vert_{i,1}
$$
Again, we use the homological relation \eqref{homo} with $a_{x,y}=\bR_{n-x-y+i+1}$,
$i=i$, $j=k=1$ and $\ell=2$ to get:
$$\psi_{i,n}=
-\vert A^{1,1}\vert_{i,2} \left( \vert A^{i,1}\vert_{1,2}\right)^{-1} \vert A\vert_{1,1}=
-\psi_{i-1,n-1} \Delta_{i-1,n}^{-1}\Delta_{i,n}
$$
Recalling that $\varphi_{i+1,n}=-\Delta_{i,n}\psi_{i,n-1}^{-1}$, this turns into the recursion
relation
\begin{equation}
\label{phiprop}
\varphi_{i+1,n}=-\left( \Delta_{i,n}\Delta_{i,n-1}^{-1}\right) \varphi_{i,n-1}
\end{equation}
Collecting all the results, we re-express \eqref{debut} as:
\begin{eqnarray}
\Delta_{i+1,n}&=& \Delta_{i,n+1}-\varphi_{i+1,n}\varphi_{i+1,n+1}^{-1}\Delta_{i,n+1}
\label{encorehiro} \\
&=& \Delta_{i,n+1}+\varphi_{i+1,n}\varphi_{i,n}^{-1}\Delta_{i,n}\label{otherhiro}
\end{eqnarray}
We deduce that
\begin{eqnarray*}
\Delta_{i,n}(\Delta_{i-1,n}^{-1}-\Delta_{i,n-1}^{-1})&=&
\Delta_{i,n}\Delta_{i-1,n}^{-1}(\Delta_{i,n-1}-\Delta_{i-1,n})\Delta_{i,n-1}^{-1}\\
&=& \Delta_{i,n}\Delta_{i-1,n}^{-1} \varphi_{i,n-1}
\varphi_{i-1,n-1}^{-1}\Delta_{i-1,n-1}\Delta_{i,n-1}^{-1}\\
&=& \varphi_{i+1,n}\varphi_{i,n}^{-1}=(\Delta_{i+1,n}- \Delta_{i,n+1})\Delta_{i,n}^{-1}
\end{eqnarray*}
and the theorem follows.
\end{proof}

\begin{remark}\label{hirodiscont}
Theorem \ref{dischironc} is a discrete analogue of the continuous formula
of Theorem 9.7.2 of \cite{GGRW}. Note that
the continuous formula contains a misprint (there should be no  $-1$ power for the
term in parentheses).
\end{remark}

\begin{remark}\label{inftyq}
In the case when $\cA$ is commutative, we write $\bR_n=R_n$, and define for $n\in \Z$:
$r_{0,n}=1$, $r_{1,n}=R_n$ and
$r_{i,n}=\det_{1\leq i,j\leq i} R_{n+i+j-i-1}$. 
Then, according to Proposition \ref{quasicomm}, we have 
$\Delta_{i,n}=r_{i,n}/r_{i-1,n-1}$.
It is easy to check that in this case the Hirota equation \eqref{nchiro}
is equivalent to the so-called $A_{\infty/2}$
$Q$-system:
$$ r_{i,n+1}r_{i,n-1}=r_{i,n}^2+r_{i+1,n}r_{i-1,n}\qquad (i>0,n\in\Z)$$
Further restricting to $r_{r+2,n}=0$ for all $n\in \Z$ amounts to the finite truncation 
condition \eqref{maintrunc}. We may then pick the $A_r$ boundary condition $r_{r+1,n}=1$
for which $r_{i,n}=R_{i,n}$, the solution of the $A_r$ $Q$-system. 
\end{remark}

\begin{remark}\label{consnc}
The truncation equation $\Delta_{r+2,n}=0$ \eqref{maintrunc} implies the existence of a conserved quantity. 
Indeed, writing the Hirota equation in the form \eqref{encorehiro} for $i=r+1$, we get
that $\varphi_{r+2,n+1}=\varphi_{r+2,n}$ is independent of $n$, due to the truncation equation.
\end{remark}

\subsubsection{General expressions for the weights}

By applying the non-commutative Hirota equation \eqref{nchiro}, we may compute explicitly 
all the weights $\by(\bm)$, $\bm\in\cM$, as given by Theorem
\ref{defymutnc}, in terms of the $\bR_n$'s. 
Let us first define new variables for $i\geq 0$ and $n\in \Z$:
\begin{equation}\label{defCD}
D_{i,n}=\Delta_{i+1,n}\Delta_{i,n}^{-1} \qquad 
C_{i,n}= (-1)^i \varphi_{i+1,n}
\end{equation}
with the initial conditions $D_{0,n}=0$ and $C_{0,n}={1}$. Note that \eqref{phiprop}
implies the relation 
\begin{equation}\label{recuC}
C_{i,n}=\Delta_{i,n}\Delta_{i,n-1}^{-1} C_{i-1,n-1}
\end{equation}
hence we have
\begin{equation}\label{evalC}
C_{i,n}=\prod_{j=0}^{i-1}\Delta_{i-j,n-j}\Delta_{i-j,n-j-1}^{-1}
\end{equation}
Note also that Remark \ref{consnc} implies that 
\begin{equation}\label{conserveC}
C_{r+1,n}=C_{r+1}
\end{equation}
is independent of $n$.
\begin{remark}
In the commutative case $\bR_n=R_n$, with $r_{0,n}=1$, $r_{1,n}=R_n$ and $r_{i,n}$ solution of the
$A_{\infty/2}$ $Q$-system of Remark \ref{inftyq}, we have
$$ \Delta_{i,n}={r_{i,n}\over r_{i-1,n-1}}\quad C_{i,n}=
{r_{i,n}\over r_{i,n-1}}
\quad D_{i,n}={r_{i+1,n}r_{i-1,n+1}\over r_{i,n+1}r_{i,n}}$$
If we further restrict to the case of the $A_r$ $Q$-system 
(by imposing $r_{r+2,n}=0$ and $r_{r+1,n}=1$), we note that the variables $C_{i,m}$
and $D_{i,m}$ reduce to the variables $c_{i,n}$ and $d_{i,n}$ of Theorem \ref{linkmuta}.
\end{remark}

The variables $C,D$ obey the following relations:

\begin{lemma}\label{identCD}
We have
\begin{eqnarray}
D_{i,n+1}&=& C_{i+1,n+1}C_{i,n}^{-1} D_{i,n} C_{i-1,n}C_{i,n+1}^{-1}\label{done}\\
C_{i,n+1}C_{i-1,n}^{-1}&=&D_{i,n}+C_{i,n}C_{i-1,n}^{-1}\label{dtwo}\\ 
C_{i+1,n}C_{i,n-1}^{-1}&=&D_{i,n}+C_{i+1,n}C_{i,n}^{-1}\label{dthree}
\end{eqnarray}
\end{lemma}
\begin{proof} Eq.\eqref{done} holds by definition, as 
a consequence of \eqref{recuC}. Eq.\eqref{dtwo} is equivalent
to the Hirota equation in the form \eqref{otherhiro}, upon multiplication by $\Delta_{i,n}^{-1}$.
Using \eqref{dtwo} and \eqref{recuC} we compute
\begin{eqnarray*}
D_{i,n}+C_{i+1,n}C_{i,n}^{-1}&=&
\Delta_{i+1,n}\Delta_{i+1,n-1}^{-1}(D_{i,n-1}+C_{i,n-1}C_{i-1,n-1}^{-1})
\Delta_{i,n-1}\Delta_{i,n}^{-1}\\
&=&\Delta_{i+1,n}\Delta_{i+1,n-1}^{-1}C_{i,n}C_{i-1,n-1}^{-1}
\Delta_{i,n-1}\Delta_{i,n}^{-1}=C_{i+1,n}C_{i,n-1}^{-1}
\end{eqnarray*}
and the lemma follows.
\end{proof}

We are now ready for the main theorem of this section:

\begin{thm}\label{ncweightcalc}
The weights $\{\by_i(\bm)\}$ are expressed in terms of the $\bR_n$'s as:
\begin{eqnarray}
\by_{2i-1}(\bm)&=&C_{i,m_i+1}C_{i-1,m_{i-1}+1}^{-1}\label{oddync} \\
\by_{2i}(\bm)&=&\left\{  \begin{matrix} 
C_{i+1,m_{i+1}+1}C_{i+1,m_{i}+1}^{-1} & {\rm if} \ m_i=m_{i+1}+1\\
{1} & {\rm otherwise} 
\end{matrix}\right\}\nonumber \\
&\times&
D_{i,m_i+1} \times \left\{ \begin{matrix} 
C_{i-1,m_{i}+1}C_{i-1,m_{i-1}+1}^{-1} & {\rm if} \ m_i=m_{i-1}-1\\
{1} & {\rm otherwise} 
\end{matrix}\right\}\label{evenync} 
\end{eqnarray}
\end{thm}
\begin{proof}
As already pointed out, the recursions \eqref{caseonenc}-\eqref{casetwonc} determine
uniquely the weights $\by(\bm)$ in terms of the initial data $\by(\bm_0)$
at the fundamental point $\bm_0=(0,0,...,0)$. However, 
Theorem \ref{stilnc} with $\bx_{2r+2}=0$ gives expressions for the weights
$\bx_i=\by_i(\bm_1)$ at the Stieltjes point $\bm_1=(0,1,2,...,r-1)$. 
So we will {\it assume} that eqns. \eqref{oddync}-\eqref{evenync} hold at the fundamental point,
and prove by induction that they hold everywhere. We then simply have to check that 
\eqref{oddync}-\eqref{evenync} hold at the Stieltjes point and the theorem will follow.

Using \eqref{oddync}-\eqref{evenync} for $m_i=i-1$, we get:
$$\by_{2i-1}(\bm_1)= C_{i,i}C_{i-1,i-1}^{-1}\quad \by_{2i}(\bm_1)=D_{i,i} $$
These reduce respectively to \eqref{oddxnc} by \eqref{recuC} and to
\eqref{evenxnc} by \eqref{defCD}. So the solution \eqref{oddync}-\eqref{evenync}) holds at
the Stieltjes point. 

We are left with proving the inductive step: assume \eqref{oddync}-\eqref{evenync} 
hold at some point $\bm$, with weights $\by_i\equiv \by_i(\bm)$, 
then show that it does at $\bm'=\mu_i \bm$, with weights $\by'_i\equiv \by_i(\bm')$.
We must distinguish the two usual cases:

\noindent{$\bullet$} {\bf Case (i)}: $m_{i-1}=m_{i}=m$, $m_{i+1}=m+1$. 
By \eqref{oddync}-\eqref{evenync} we have:
$$\by_{2i-1}=C_{i,m+1}C_{i-1,m+1}^{-1}, \  \  \by_{2i}=D_{i,m+1}, 
\ \ \by_{2i+1}=C_{i+1,m+2}C_{i,m+1}^{-1}$$
The evolutions \eqref{caseonenc} give:
\begin{eqnarray*}
\by_{2i-1}'&=&\by_{2i-1}+\by_{2i}=C_{i,m+2}C_{i-1,m+1}^{-1}
\quad \qquad ({\rm by}\ {\rm eq.}\ {\rm \eqref{dtwo})}\\
\by_{2i}'&=&\by_{2i+1}\by_{2i}(\by_{2i-1}')^{-1}
=D_{i,m+2} 
\qquad \qquad  ({\rm by}\ {\rm eq.}\ {\rm \eqref{done})}\\
\by_{2i+1}'&=&\by_{2i+1}\by_{2i}(\by_{2i-1}')^{-1}
=C_{i+1,m+2}C_{i,m+2}^{-1}
\end{eqnarray*}
and \eqref{oddync}-\eqref{evenync} hold for $\bm'$, with $m_{i-1}'=m$ and $m_i'=m_{i+1}'=m+1$.

\noindent{$\bullet$} {\bf Case (ii)}: $m_{i-1}=m_{i}=m_{i+1}=m$. 
Now we have by \eqref{oddync}-\eqref{evenync}:
$$\by_{2i-1}=C_{i,m+1}C_{i-1,m+1}^{-1}, \  \  \by_{2i}=D_{i,m+1}, 
\ \ \by_{2i+1}=C_{i+1,m+1}C_{i,m+1}^{-1}, \ \ \by_{2i+2}=A D_{i+1,m+1}$$
where $A=C_{i+1,m+1}C_{i,m+1}^{-1}$ if $m_{i+2}=m-1$, and $A=1$ otherwise.
The evolutions \eqref{casetwonc} give:
\begin{eqnarray*}
\by_{2i-1}'&=&\by_{2i-1}+\by_{2i}=C_{i,m+2}C_{i-1,m+1}^{-1}
\quad \qquad \qquad\qquad \qquad({\rm by}\ {\rm eq.}\ {\rm \eqref{dtwo})}\\
\by_{2i}'&=&\by_{2i+1}\by_{2i}(\by_{2i-1}')^{-1}
=C_{i,m+1}C_{i+1,m+2}^{-1}D_{i,m+2} 
\qquad \quad  ({\rm by}\ {\rm eq.}\ {\rm \eqref{done})}\\
\by_{2i+1}'&=&\by_{2i+1}\by_{2i-1}(\by_{2i-1}')^{-1}
=C_{i+1,m+1}C_{i,m+2}^{-1}\\
\by_{2i+2}'&=&\by_{2i+2}\by_{2i-1}(\by_{2i-1}')^{-1}
=A D_{i+1,m+1}C_{i,m+1}C_{i,m+2}^{-1}
\end{eqnarray*}
and \eqref{oddync}-\eqref{evenync} hold for $\bm'$, with $m_{i-1}'=m_{i+1}'=m$ and $m_i'=m+1$.
\end{proof}

\begin{example}\label{incmotz}
The particular set of non-decreasing Motzkin paths with $m_1\leq m_2 \leq \cdots \leq m_r$
is obtained by only applying mutations pertaining to the case (i). The corresponding weights are particularly
simple:
$$ \by_{2i-1}(\bm)=C_{i,m_i+1}C_{i-1,m_{i-1}+1}^{-1},\qquad \by_{2i}(\bm)=D_{i,m_i+1}.$$
In particular, for the fundamental point $\bm_0$, we get 
$$\by_{2i-1}(\bm_0)=C_{i,1}C_{i-1,1}^{-1},\qquad \by_{2i}(\bm_0)=D_{i,1}.$$
This can be compared to
the so-called $J$-fraction expressions of
\cite{GKLLRT}, Theorem. 8.12.
\end{example}

\begin{example}\label{atwoex}
In the case $r=2$, we have the three Motzkin paths $(0,0)$, $(0,1)$ and $(1,0)$. We have
$C_{0,n}={1}$ and $C_{3,n}=C_3$. The weights are:
$$ \begin{matrix}
\by_1(0,0)&=&C_{1,1}\\
\by_2(0,0)&=&D_{1,1}\\
\by_3(0,0)&=&C_{2,1}C_{1,1}^{-1}\\
\by_4(0,0)&=&D_{2,1}\\
\by_5(0,0)&=&C_{3}C_{2,1}^{-1}
\end{matrix}
\qquad
\begin{matrix}
\by_1(0,1)&=&C_{1,1}\\
\by_2(0,1)&=&D_{1,1}\\
\by_3(0,1)&=&C_{2,2}C_{1,1}^{-1}\\
\by_4(0,1)&=&D_{2,2}\\
\by_5(0,1)&=&C_{3}C_{2,2}^{-1}
\end{matrix}
\qquad
\begin{matrix}
\by_1(1,0)&=&C_{1,2}\\
\by_2(1,0)&=&C_{2,1}C_{2,2}^{-1}D_{1,2}\\
\by_3(1,0)&=&C_{2,1}C_{1,2}^{-1}\\
\by_4(1,0)&=&D_{2,1}C_{1,1}C_{1,2}^{-1}\\
\by_5(1,0)&=&C_{3}C_{2,1}^{-1}
\end{matrix}
$$
In the formulation as paths on $\Gamma_\bm$ via the canonical form of the continued fractions,
the only ``redundant" weight occurs in the last case $\bm=(1,0)$, where
$\bw_{3,1}=t\by_4\by_3^{-1}\by_2=tD_{2,1}C_{1,1}C_{2,2}^{-1}D_{1,2}$.
Note that weights may involve $\bR_n$'s with $n<0$ in their definition
(e.g. $\Delta_{i,n}$ for $n<i-1$).
\end{example}


\begin{thm}\label{finlut}
We have
$$F_\bm(t)=\sum_{n\geq 0} t^n \bR_{1,n+m_1}\bR_{1,m_1}^{-1}$$
In other words, the quantity $\bR_{1,n+m_1}\bR_{1,m_1}^{-1}$ is the
partition function of paths from vertex $1$ to itself on $\Gamma_\bm$,
and with the non-commutative weights \eqref{noncoweights} associated
to \eqref{oddync}-\eqref{evenync}.  As such it is a Laurent polynomial
of these weights, with non-negative integer coefficients.
\end{thm}
\begin{proof}
The statement is proved by induction on mutations. It holds for $\bm=\bm_0$,
with $m_1=(\bm_0)_1=0$. Any given $\bm\in\cM$ is obtained by an iterative application
of mutations $\mu_i$ of type (i) or (ii). 
Assume the theorem holds for some $\bm\in\cM$.
Let us form $\bm'=\mu_i\bm$. When $i>1$, we have 
$\bF_m(t)=\bF_{m'}(t)=\sum_{n\geq 0} t^n\bR_{n+m_1}\bR_{m_1}^{-1}$ by \eqref{eqf}: hence
the {\it same} quantity $\bR_{n+m_1}\bR_{m_1}^{-1}$
is interpreted as the partition function for paths on $\Gamma_{\bm'}$ with
the weights associated to $\by_i(\bm')$, and the theorem holds for $\bm'$, as $m_1'=m_1$.
If $i=1$, then we have
\eqref{rrt}. Using $\by_1(\bm)=C_{1,m_1+1}=\bR_{m_1+1}\bR_{m_1}^{-1}$, we get
\begin{eqnarray*}
\bF_{\bm}(t)&=&\sum_{n\geq 0}t^n\bR_{n+m_1}\bR_{m_1}^{-1}={1}+t \bF_{\bm'}(t)\,\bR_{m_1+1}\bR_{m_1}^{-1}\\
\bF_{\bm'}(t)&=&\sum_{n\geq 0}t^n\bR_{n+m_1+1}\bR_{m_1+1}^{-1}
\end{eqnarray*}
and the theorem holds for $\bm'$, with $m_1'=m_1+1$.
\end{proof}

\begin{example}\label{atwopf}
Let us apply Theorem \ref{finlut} to the calculation of $\bR_3$ in the case $r=2$,
at the points $\bm=(0,0),(0,1),(1,0)$ respectively.
We have
\begin{eqnarray*}
\bm=(0,0):\ \ \bR_3\bR_0^{-1}&=& (\by_1+\by_2)^2\by_1+(\by_3+\by_4)\by_2\by_1\\
&=& (C_{1,1}+D_{1,1})^2C_{1,1}+(C_{2,1}+D_{2,1}C_{1,1})C_{1,1}^{-1}D_{1,1}C_{1,1}\\
\bm=(0,1):\ \ \bR_3\bR_0^{-1}&=& (\by_1+\by_2)^2\by_1+\by_3\by_2\by_1\\
&=&(C_{1,1}+D_{1,1})^2C_{1,1}+C_{2,2}C_{1,1}^{-1}D_{1,1}C_{1,1}\\
\bm=(1,0):\ \ \bR_3\bR_1^{-1}&=&\by_1^2+\by_2\by_1+\by_{3,1}\by_1\\
&=&(C_{1,2})^2+(C_{2,1}+D_{2,1}C_{1,1})C_{2,2}^{-1}D_{1,2}C_{1,2}
\end{eqnarray*}
respectively equal to the partition functions for non-commutative weighted paths from the root to itself
on the corresponding graphs $\Gamma_\bm$, with weights $\by_i=\by_i(\bm)$
given in Example \ref{atwoex},
and with respectively $3,3,2$ steps towards the root. The identities between these expressions are 
obtained from identities of Lemma \ref{identCD} relating $C$ and $D$:
$$ C_{1,1}+D_{1,1}=C_{1,2},\quad C_{2,1}+D_{2,1}C_{1,1}=C_{2,2}, \quad
C_{1,1}^{-1}D_{1,1}=C_{2,2}^{-1}D_{1,2}C_{1,2} $$
\end{example}

\section{Application I: T-system}\label{section4}
Here, we apply the construction of the previous section to the $A_r$ $T$-system \eqref{Tsystem}, which can be written as a non-commutative $Q$-system as shown in \cite{DFK09a}. The $A_r$ boundary condition
implies a truncation relation $T_{r+2,j,k}=0$ for all $j,k\in\Z$, similar to that studied in Section \ref{section3}, so we can use the results of Section \ref{section3} to give an
alternative derivation of the $T$-system positivity results of \cite{DFK09a}.

\subsection{The $A_r$ $T$-system as non-commutative $Q$-system}
We may view the $T$-system equations \eqref{Tsystem} as matrix elements of a relation in a non-commutative 
algebra in a particular representation. We consider the
non-commuting variables $\bT_{i,k}\in\cA$ which have a faithful representation 
on a Hilbert space $\mathcal H$ with basis $\{ |i\rangle \}_{i\in \Z}$:
\begin{equation}\label{tact}
\bT_{i,k} \, | j+k+i\rangle = T_{i,j,k} \, | j-k-i\rangle  \qquad (i\geq 0; \, j,k\in\Z)
\end{equation}
Introducing the shift operator $\bd$ which acts on the same representation as $\bd \, |t\rangle =|t-1\rangle$ we may rewrite
the $A_r$ boundary conditions $T_{0,j,k}=T_{r+1,j,k}=1$ as:
\begin{equation}\label{bcncQ}
\bT_{0,k}=\bd^{2k},\qquad \bT_{r+1,k}=\bd^{2(k+r+1)},
\end{equation}
while the finite truncation condition implies
$$ \bT_{r+2,k}=0 \qquad (k\in\Z) $$

Note that all variables $\bT_{i,k}$, $0\leq i\leq r+1$, are invertible for generic non-vanishing
matrix elements $T_{i,j,k}$, with 
\begin{equation}\label{tactinv}
(\bT_{i,k})^{-1} \, | j-k-i\rangle = {1\over T_{i,j,k}} \, | j+k+i\rangle  \qquad (0\leq i\leq r+1;\, j,k\in\Z)
\end{equation}

We have the following

\begin{prop}\label{ncQ}
The $A_r$ $T$-system \eqref{Tsystem} is obtained as matrix elements in $\mathcal H$  of the following non-commutative $Q$-system:
\begin{equation}\label{NCQ}
\bT_{i,k+1}(\bT_{i,k})^{-1} \bT_{i,k-1}=\bT_{i,k} + \bT_{i+1,k}(\bT_{i,k})^{-1}\bT_{i-1,k}
\quad (1\leq i\leq r;\, k\in\Z)
\end{equation}
with the boundary conditions \eqref{bcncQ}.
\end{prop}
\begin{proof} Applying \eqref{NCQ} to the vector $|j+k-1+i\rangle$ and using 
\eqref{tact} and \eqref{tactinv} yields \eqref{Tsystem} divided by $T_{i,j+1,k}$ times $|j-k+1-i\rangle$.
\end{proof}

\begin{remark}
Note that the operators $\bT_{i,k}$ are only ``mildly" non-commuting, in the sense
that the operators $\{\bd^m \bT_{i,k} \bd^{-m-2i-2k}\}_{m,k\in\Z}$ are all diagonal when acting on $\mathcal H$, and therefore commute
with each other. 
In a sense the $\bd$ operator is the only source of non-commutation. It was interpreted in 
\cite{DFK09a}
as a time-shift operator, and the non-commutation of the variables is due to their time-ordering.
Nevertheless, this gives us a first non-trivial example of a non-commutative $Q$-system.
\end{remark}

\subsection{Quasi-Wronskian formulation}

The $A_r$ $T$-system is a particular case of the non-commutative finite truncation equation
\eqref{maintrunc}. We start from the solution $T_{i,j,k}$ of the system \eqref{Tsystem}
and the operators $\bT_{i,k}$ \eqref{tact}-\eqref{tactinv}. This gives rise to the sequence
$\bR_n=\bT_{1,n}\in\cA$, $n\in\Z$.
Defining $\Delta_{i,n}$, $n\in\Z$ as the quasi-Wronskians \eqref{wronskdefnc}
of this sequence, we have the following action of the quasi-Wronskians
on the basis $\{|t\rangle\}_{t\in\Z}$:

\begin{thm}\label{actdelt}
\begin{equation}\label{deltactsys}
\Delta_{i,n} \, |j+n+i\rangle= {T_{i,j,n}\over T_{i-1,j,n-1}} \, |j-n-i \rangle
\qquad (1\leq i\leq r+2;\, j,n\in \Z)
\end{equation}
\end{thm}
\begin{proof}
By induction on $i$. For $i=1$, we have $\Delta_{1,n}=\bR_n$, and $\bR_n|j+n+1\rangle
=T_{1,j,n}\, |j-n-1\rangle$, which amounts to \eqref{deltactsys} for $i=1$, as $T_{0,j,k}=1$
for all $j,k$. Assume \eqref{deltactsys} holds for all $\beta\leq i$. Using the Hirota equation
\eqref{nchiro}, and the recursion hypotheses we get:
$\Delta_{i+1,n}|j+n+i+1\rangle=\delta_{i+1,n}\, |j-n-i-1\rangle$, with
\begin{eqnarray*}
\delta_{i+1,n}&=&
{T_{i,j,n+1}\over T_{i-1,j,n}} 
+{T_{i,j-1,n}\over T_{i-1,j,n-1}}\left(
{T_{i-2,j,n-1}\over T_{i-1,j,n}}-{T_{i-1,j,n-2}\over T_{i,j,n-1}}
\right)
{T_{i,j+1,n}\over T_{i-1,j+1,n-1}}\\
&=&{T_{i,j,n+1}\over T_{i-1,j,n}} -{T_{i,j-1,n}T_{i,j+1,n}\over T_{i-1,j,n}T_{i,j,n-1}}\\
&=&{T_{i+1,j,n}\over T_{i,j,n-1}}
\end{eqnarray*}
where the second line follows from \eqref{Tsystem} for $i\to i-1$ and $k=n-1$, while the third follows
from \eqref{Tsystem} with $k=n$.
\end{proof}

\begin{remark}[Relation to the known determinant relation of \cite{BazResh}]
The $T$-system \eqref{Tsystem}  
may be simplified by first eliminating the $T_{i,j,n}$,
$i>1$ and $j,n\in\Z$ in terms of the $T_{j,n}=T_{1,j,n}$ by use of the Desnanot-Jacobi
formula, as the following determinant:
$$ T_{i,j,n}=\det(\Theta_{i,j,n}), \quad 
\Big(\Theta_{i,j,n}\Big)_{a,b}= T_{j-a+b,n+i+1-a-b} \ \ (1\leq a,b \leq i)$$
We note that 
$$ {T_{i,j,n}\over T_{i-1,j,n-1}}=\vert \Theta_{i,j,n} \vert_{1,1} $$
so the diagonal operator $\bd^{-n-i}\Delta_{i,n}\bd^{-n-i}$ 
has the eigenvalue $\vert\Theta_{i,j,n}\vert_{1,1}$
when acting on $|j\rangle$, for all $j\in \Z$.
\end{remark}

Applying Theorem \ref{actdelt} for $i=r+2$, we find:

\begin{cor}
The $T$-system solution $(\bR_k)_{k\in\Z}$ satisfies the finite truncation equation 
\eqref{maintrunc}.
\end{cor}

We may therefore apply the results of Section \ref{section3}. The $C,D$ operators \eqref{defCD}
act on the basis $\{|t\rangle\}_{t\in\Z}$ as:
\begin{eqnarray*} C_{i,n}\,|t\rangle &=&
{T_{i,t+n-i,n}\over T_{i,t+n-i+1,n-1}} \, |t-2i\rangle 
\qquad 
C_{i,n}^{-1}\,|t\rangle =
{T_{i,t+n+i+1,n-1}\over T_{i,t+n+i,n}} \, |t+2i\rangle \\
D_{i,n}\, |t\rangle &=&{T_{i+1,t+n+i-1,n}T_{i-1,t+n+i,n-1}
\over T_{i,t+n+i-1,n-1}T_{i,t+n+i,n}}\, |t-2\rangle 
\end{eqnarray*}
Using the notation
\begin{equation}\label{lamudef}
c_{i,t,m}={T_{i,t,m}\over T_{i,t+1,m-1}} 
\qquad d_{i,t,m}=
{T_{i+1,t,m} T_{i-1,t+1,m-1}\over T_{i,t+1,m}T_{i,t,m-1}}
\end{equation}
and the fact that expressions of the form $C_{i,m}C_{i,p}^{-1}$ act diagonally, 
Theorem \ref{ncweightcalc} yields non-commutative weights $\by_i(\bm)$ which act
on the basis $\{|t\rangle\}_{t\in\Z}$ as:
\begin{eqnarray}
\by_{2i-1}(\bm) |t\rangle &=& {c_{i,t+m_{i}+i-1,m_{i}+1}
\over c_{i-1,t+m_{i-1}+i,m_{i-1}+1}}\, |t-2\rangle  \label{tweight1}\\
\by_{2i}(\bm) |t\rangle &=&\left\{ \begin{matrix} 
{c_{i+1,t+m_{i+1}+i,m_{i+1}+1}\over c_{i+1,t+m_i+i,m_i+1}} & 
{\rm if} \ m_i=m_{i+1}+1\\
1 & {\rm otherwise} 
\end{matrix}\right\}\times d_{i,t+m_i+i,m_i+1} \nonumber \\
&&\qquad \times
\left\{ \begin{matrix} 
{c_{i-1,t+m_i+i,m_i+1}\over c_{i-1,t+m_{i-1}+i,m_{i-1}+1}} & 
{\rm if} \ m_i=m_{i-1}-1\\
1 & {\rm otherwise} 
\end{matrix}\right\}\,  |t-2\rangle \label{tweight2}
\end{eqnarray}

\begin{remark}
To be able to compare the results \eqref{tweight1}-\eqref{tweight2} to those of  \cite{DFK09a}, we note that
here, 
in the non-commutative weighted path interpretation, the up steps have weights ${\bf 1}$,
acting as the identity on the basis $\{|t\rangle\}_{t\in\Z}$, while all the down weights
(including the redundant ones) act as a diagonal operator times $\bd^2$. Interpreting $t$ as
a discrete time variable, we see that up steps use no time, 
while down steps use two units of time. In Ref. \cite{DFK09a},
the time-dependent weighted path formulation of the $T$-system solutions uses skeleton
down steps as well as up steps that all use one unit of time (i.e. are the product of a diagonal operator by $\bd$), while the redundant steps
may go back in time (a step $i\to i-k$, $k>1$ uses $2-k$ units of time, 
hence the corresponding weight is a diagonal operator times $\bd^{2-k}$). The two descriptions
are equivalent up to a ``gauge" transformation, namely a conjugation of the weights 
\eqref{tweight1}-\eqref{tweight2}:
$\by_{2i-1}\to \bd^{-i-1} \by_{2i-1}\bd^{i}$ and $\by_{2i}\to \bd^{-i-1} \by_{2i}\bd^{i}$.
\end{remark}

\section{Application II: quantum Q-system}\label{section5}
\subsection{Quantum cluster algebra}

The solutions of $Q$-system for $A_r$ have been shown to form particular
clusters in a suitable cluster algebra \cite{Ke07}. On the other hand,
cluster algebras are known to have quantum analogues, 
the so-called quantum cluster algebras \cite{BZ}. In particular, there exists a quantum deformation of the $Q$-system relations.

The cluster algebra for the $A_r$ $Q$-system is a rank $2r$ cluster algebra
with trivial coefficients, and with initial cluster $\bx_0$ and exchange matrix $B_0$:
$$ \bx_0=(R_{1,0},R_{2,0},...,R_{r,0},R_{1,1},R_{2,1},...,R_{r,1})\qquad  B_0=\begin{pmatrix}
0 & -C^t \\ C & 0\end{pmatrix} $$
where $C$ is the Cartan matrix for $A_r$.
The subset of (forward) mutations we have considered so far is such that each mutation amounts
to one application of the $Q$-system equation, allowing to update the cluster
via a substitution of the form: 
$R_{i,n-1}\to R_{i,n+1}=(R_{i,n}^2+R_{i+1,n}R_{i-1,n})/R_{i,n-1}$ (mutation $\mu_i$
if $n$ is odd, $\mu_{r+i}$ if $n$ is even). Recall also that the compound mutation
$\mu=\mu_{2r}\mu_{2r-1}\cdots \mu_1$ acts on $\bx_0$ as a global
translation of all the $n$ indices by $+2$ (i.e. $R_{i,n}\to R_{i,n+2}$ for all $i$), preserving $B_0$.
The ``half-translation" of all indices by $+1$ is realized by acting only with 
$\mu_{\rm even}=\mu_r\mu_{r-1}\cdots \mu_1$
or $\mu_{\rm odd}=\mu_{2r}\mu_{2r-1}\cdots \mu_{r+1}$ in alternance,
which reverse the sign of the exchange matrix.

A quantum cluster algebra with trivial coefficients corresponding to this involves 
the {\it same} exchange matrix, subject to the same evolution under mutation, but
the cluster variables are now non-commuting variables
$\bR_{i,n}\in \cA$ whose commutation relations within each cluster are dictated 
by the exchange matrix as follows.
We have for any pair $(x_i,x_j)$ of variables within the same cluster the following $q$-commutation relation:
\begin{equation}\label{xqcomm} 
x_ix_j=q^{\Lambda_{i,j}} x_jx_i
\end{equation}
where $\Lambda$ is a suitably normalized $2r\times 2r$ matrix proportional to the inverse of the 
exchange matrix. Note that in a different cluster the commutation relations 
are different, as they are entirely dictated by the exchange matrix, 
and may involve other cluster variables.

For the initial cluster, in order for $\Lambda$
to have only integer entries, we choose 
$\Lambda=(r+1)B_0^{-1}=\begin{pmatrix} 0 &\lambda \\-\lambda & 0\end{pmatrix}$, where 
$\lambda=(r+1)C^{-1}$. We obtain the following $q$-commutation relations between variables
of the fundamental cluster:

\begin{lemma}
For $1\leq i,j \leq r$, we have:
\begin{eqnarray}
\bR_{i,0}\bR_{j,1}&=&q^{\lambda_{i,j}} \bR_{j,1}\bR_{i,0}
\label{qcomm1} \\
\bR_{i,n}\bR_{j,n}&=&\bR_{j,n}\bR_{i,n} \qquad  (n=0,1)\ \label{qcomm2} 
\end{eqnarray}
where
\begin{equation}\label{lambdamat}
\lambda_{i,j}={\rm Min}(i,j)\Big( r+1- {\rm Max}(i,j)\Big) \qquad (1\leq i,j \leq r)
\end{equation}
\end{lemma}
\begin{proof}
By explicitly inverting the Cartan matrix $C$, with entries  $2\delta_{i,j}-\delta_{|i-j|,1}$.
\end{proof}

Note that the relations \eqref{qcomm1} and \eqref{qcomm2} are preserved under the above-mentioned 
half-translation of all $n$ indices by $+1$, as $B_0\to -B_0$, and therefore $\Lambda\to -\Lambda$
for $\mu_{\rm even}$ and vice versa for $\mu_{\rm odd}$.
Hence the relations \eqref{qcomm1}-\eqref{qcomm2} extend to
\begin{eqnarray}
\bR_{i,n}\bR_{j,n+1}&=&q^{\lambda_{i,j}} \bR_{j,n+1}\bR_{i,n}
\label{qcom1} \\
\bR_{i,n}\bR_{j,n}&=&\bR_{j,n}\bR_{i,n} \ \label{qcom2} 
\end{eqnarray}
for all $n\in \Z$.
In particular, the initial sequence $\bR_n=\bR_{1,n}$, $n\in\Z$, has the $q$-commutation relations:
\begin{equation}
\bR_n \bR_{n+1} =q^r \, \bR_{n+1}\bR_n
\end{equation}

In the following we consider the {\it same} subset of cluster mutations as in the commuting $A_r$
$Q$-system cluster algebra, leading to a subset of cluster
nodes indexed by Motzkin paths $\bm\in \cM$. Each of these clusters $\bx_\bm$
is now made of $2r$ non-commuting
variables $\{\bR_{i,m_i},\bR_{i,m_i+1} \}_{i=1}^r$, ordered as in the commuting case.
As before, these will play the role of possible initial conditions for the quantum $Q$-system below.

\subsection{Quantum $Q$-system}

The particular form of the quantum $Q$-system we choose is borrowed from 
the $T$-system case \eqref{NCQ}, 
and compatible up to multiplicative redefinition of the variables with the mutations
of the quantum cluster algebra.

\begin{defn}
We define the quantum $Q$-system for $A_r$ as:
\begin{equation}\label{Qquant}
\bR_{i,k+1}(\bR_{i,k})^{-1}\bR_{i,k-1}=\bR_{i,k}+\bR_{i+1,k}(\bR_{i,k})^{-1}\bR_{i-1,k}
\qquad (i\geq 1;\, k\in\Z)
\end{equation}
with the boundary conditions
\begin{equation}\label{boundQquant}
\bR_{0,k}={\bf 1}\qquad  \bR_{r+1,k}={\bf 1} \qquad (k\in \Z)
\end{equation}
\end{defn}

Note that the latter imply a relation 
\begin{equation}\label{truncR}
\bR_{r+2,k}=0 \qquad (k\in \Z)
\end{equation}
as $\lambda_{r+1,r+1}=0$. We also have $\bR_{-1,n}=0$ as $\lambda_{0,0}=0$.

Using the 
$q$-commutations \eqref{qcom1}-\eqref{qcom2}, we may rewrite the system \eqref{Qquant}
in a simpler form:

\begin{lemma}
The $A_r$ quantum $Q$-system is equivalent to:
\begin{equation}\label{quantQ}
q^{\lambda_{i,i}} \bR_{i,k+1}\bR_{i,k-1}=(\bR_{i,k})^2+\bR_{i+1,k}\bR_{i-1,k}
\qquad (i\geq 1;\, k\in\Z)
\end{equation}
with $\lambda_{i,j}$ as in \eqref{lambdamat}, and
with the same boundary condition \eqref{boundQquant}.
\end{lemma}

\begin{remark}
The relation \eqref{quantQ} is equivalent to the relevant quantum 
cluster mutations of \cite{BZ} up to a scalar renormalization
of the cluster variables. Indeed, defining the variables
 $\bX_{i,n}=q^{{r+1\over 4}\lambda_{i,i}} \bR_{i,n}$, we find:
\begin{equation}\label{quantQsys}
\bX_{i,k+1}=q^{\lambda_{i,i}}\, (\bX_{i,k-1})^{-1}(\bX_{i,k})^2
+q^{{1\over 2}(\lambda_{i,i-1}+\lambda_{i,i+1})}\, (\bX_{i,k-1})^{-1}\bX_{i-1,k}\bX_{i+1,k}
\end{equation}
in agreement with the definitions of \cite{BZ}. Note that boundary conditions remain the same: 
$\bX_{0,n}=\bX_{r+1,n}={\bf 1}$ for all $n\in \Z$.
\end{remark}

\subsection{Quasi-Wronskian formulation}

The quasi-Wronskians \eqref{wronskdefnc} of the variables $\bR_n=\bR_{1,n}$ 
may be expressed as simple monomials
in terms of the solutions of the quantum $Q$-system \eqref{Qquant}-\eqref{boundQquant}.
We have:

\begin{thm}\label{quantdelt}
\begin{equation}\label{deltaquant}
\Delta_{i,n}=q^{-{1\over 2}(i-1)(2r+2-i)}\, \bR_{i,n}\bR_{i-1,n-1}^{-1} \qquad (i>1,n\in \Z)
\end{equation}
\end{thm}
\begin{proof}
By induction on $i$. Define 
$\Gamma_{i,n}=q^{-{1\over 2}(i-1)(2r+2-i)}\, \bR_{i,n}\bR_{i-1,n-1}^{-1}$, the r.h.s. of \eqref{deltaquant}
and $a_i={1\over 2}(i-1)(2r+2-i)$.
For $i=1$, we have $\Delta_{1,n}=\bR_{1,n}=\Gamma_{1,n}$
from the boundary condition \eqref{boundQquant}. 
Assume \eqref{deltaquant} holds for all $1\leq j \leq i$.
We use the Hirota equation \eqref{nchiro} to express
\begin{equation}\label{hire}
\Delta_{i+1,n}-\Gamma_{i+1,n}=\Gamma_{i,n+1}-\Gamma_{i+1,n}-
\Gamma_{i,n}\Gamma_{i,n-1}^{-1}(\Gamma_{i,n-1}-\Gamma_{i-1,n})\Gamma_{i-1,n}^{-1}\Gamma_{i,n}
\end{equation}
Moreover, defining $B_{i,n}=(\bR_{i,n})^2+\bR_{i+1,n}\bR_{i-1,n}-q^{\lambda_{i,i}}\bR_{i,n+1}\bR_{i,n-1}$
($=0$ due to the quantum Q-system equation \eqref{quantQ}), we find that
$$ \Gamma_{i-1,n}-\Gamma_{i,n-1}=q^{-a_i} \Big((\bR_{i-1,n-1})^2-B_{i-1,n-1}\Big)
\bR_{i-2,n-1}^{-1}\bR_{i-1,n-2}^{-1}$$
due to the $q$-commutation relations \eqref{qcom1}-\eqref{qcom2}
and the identity $\lambda_{i-1,i-1}-\lambda_{i-2,i-1}-a_{i}=-a_{i-1}$.
Substituting this into \eqref{hire}, we get
\begin{eqnarray*}
\Delta_{i+1,n}-\Gamma_{i+1,n}&=&q^{-a_{i+1}}\bR_{i,n}^2
\bR_{i-1,n}^{-1}\bR_{i,n-1}^{-1}\\
&&-q^{-a_i} \Gamma_{i,n}\Gamma_{i,n-1}^{-1}(\bR_{i-1,n-1})^2
\bR_{i-2,n-1}^{-1}\bR_{i-1,n-2}^{-1}
\Gamma_{i-1,n}^{-1}\Gamma_{i,n}\\
&=&0
\end{eqnarray*}
by use of the $q$-commutation relations \eqref{qcom1}-\eqref{qcom2}.
\end{proof}

\begin{remark}\label{quantumwronsk}
Theorem \ref{quantdelt} allows to interpret $\bR_{i,n}$ as a ``quantum Wronskian", namely 
a ``quantum determinant" of the Wronskian matrix used in Definition \ref{ncwdef}.
Indeed, we may write
$$\bR_{i,n}=q^{{i(i-1)\over 6}(3r+2-i)}\, \Delta_{i,n}\Delta_{i-1,n-1}\cdots \Delta_{1,n-i+1}
\qquad (i\geq 1)$$
which is a product of principal quasi-minors of the Wronskian matrix, like in the standard definitions of
quantum determinants \cite{EtingofRetakh} (note however that here the quasi-minors do not commute).
In the cases $i=1,2,3$, using
the $q$-commutations and the quantum $Q$-system, we find
the following expressions for the quantum Wronskian:
\begin{eqnarray*}
\bR_{1,n}&=&\bR_n\\
\bR_{2,n}&=&q^r \bR_{n+1}\bR_{n-1}-\bR_{n}^2\\
\bR_{3,n}&=& q^{3r-1}\bR_{n+2}\bR_n\bR_{n-2}
-q^{2r-1}(\bR_{n+2}\bR_{n-1}^2+\bR_{n+1}^2\bR_{n-2})\\
&&+q^{r-1}(q^{r+1}+1)\bR_{n+1}\bR_n\bR_{n-1}-\bR_n^3
\end{eqnarray*}
\end{remark}

Applying Theorem \ref{quantdelt} for $i=r+2$, and using the equation \eqref{truncR}, we get:

\begin{cor}
The solution of the quantum $Q$-system \eqref{Qquant}-\eqref{boundQquant} satisfies the finite truncation
condition $\Delta_{r+2,n}=0$ for all $n\in \Z$.
\end{cor}

We may therefore apply the results of Section \ref{section3}. Assuming $\bR_{i,n}$ is a solution
of \eqref{Qquant}-\eqref{boundQquant}, we consider the usual subset of clusters $\bx_\bm$,
also viewed as initial data for the quantum $Q$-system.
Recall that the commutation relations between the quantum cluster variables within 
each cluster $\bx_\bm$ are determined by the exchange matrix $B_\bm$.
Using the explicit expressions for the
exchange matrix $B_\bm$ of Ref. \cite{DFK3}, it is an easy exercise to compute the $q$-commutation
relations between the cluster variables within any given cluster $\bx_\bm$. We find that:

\begin{lemma}\label{qcomgen}
\begin{equation}\label{genqcomm}
\bR_{i,n}\bR_{j,n+p}=q^{p\lambda_{i,j}}\, \bR_{j,n+p}\bR_{i,n}
\end{equation}
for all $n,p,i,j$ such that the variables belong to the same cluster $\bx_\bm$. 
\end{lemma}

\begin{remark}
Lemma \ref{qcomgen} does not cover for instance the case $j=i$, $p=2$,
as $\bR_{i,n}$ and $\bR_{i,n+2}$ never belong to the same cluster $\bx_\bm$. Using
the quantum $Q$-system and the valid $q$-commutations, we find for instance
$$ \bR_{i,n}\bR_{i,n+2}=q^{2\lambda_{i,i}}\,
\left( q^{-r-1}\, \bR_{i,n+2}\bR_{i,n}+( q^{-r-1}-1)\bR_{i,n+1}^2\right)$$
However, we may apply Lemma \ref{qcomgen} to compute the commutation of $\Delta_{i,n}$
and $\Delta_{j,m}$, using Theorem \ref{quantdelt}, when the
four variables $\bR_{i,n},\bR_{i-1,n-1},\bR_{j,m},\bR_{j-1,m-1}$ belong to the same
cluster $\bx_\bm$:
$$ \Delta_{i,n}\Delta_{j,m}=q^{(n-m)(\lambda_{i,j-1}+\lambda_{j,i-1})+r(j-i)}\, 
\Delta_{j,m} \Delta_{i,n}$$
\end{remark}

Using the $q$-commutation relations \eqref{genqcomm}
we easily get the following:

\begin{lemma}\label{uselem}
The quantities $C_{i,n}$ and $D_{i,n}$ \eqref{defCD} are expressed as:
$$C_{i,n}=q^{i(i-1)\over 2} \, \bR_{i,n}\bR_{i,n-1}^{-1} \qquad D_{i,n}= \bR_{i+1,n}\bR_{i,n}^{-1}
\bR_{i,n-1}^{-1}\bR_{i-1,n-1}$$
\end{lemma}
In particular, the value of
the conserved quantity \eqref{conserveC} follows from Lemma \ref{uselem} and the boundary condition
\eqref{boundQquant}:
$$ C_{r+1} = q^{r(r+1)\over 2}. $$
Note that in this case, this conserved quantity is central.

\begin{thm}\label{mainquanthm}
Let $\bR_{i,n}$ be a solution of the quantum $Q$-system \eqref{Qquant}-\eqref{boundQquant} 
and $\bR_n=\bR_{1,n}$.
The quantity $\bR_{n+m_1}\bR_{m_1}^{-1}$ for $n\geq 0$,
when expressed in terms of initial data $\bx_\bm$, is the partition function for non-commutative
weighted paths on the graph $\Gamma_\bm$ from the root to itself with $n$ steps towards the root
and with skeleton weights:
\begin{eqnarray}
\qquad \by_{2i-1}(\bm)&=&
q^{i-1} \, \bR_{i,m_i+1}\bR_{i,m_i}^{-1} 
 \bR_{i-1,m_{i-1}}\bR_{i-1,m_{i-1}+1}^{-1}
\label{oddyq} \\
\by_{2i}(\bm)&=&\left\{  \begin{matrix} 
\bR_{i+1,m_{i+1}+1}\bR_{i+1,m_{i+1}}^{-1}
\bR_{i+1,m_{i}}\bR_{i+1,m_{i}+1}^{-1}
& {\rm if} \ m_i=m_{i+1}+1\\
{\bf 1} & {\rm otherwise} 
\end{matrix}\right\}\nonumber \\
&&\times\, 
 \bR_{i+1,m_i+1}\bR_{i,m_i+1}^{-1}
 \bR_{i,m_i}^{-1}\bR_{i-1,m_i} \nonumber \\
&& \times \left\{ \begin{matrix} 
\bR_{i-1,m_{i}+1}\bR_{i-1,m_{i}}^{-1}
\bR_{i-1,m_{i-1}}\bR_{i-1,m_{i-1}+1}^{-1}
& {\rm if} \ m_i=m_{i-1}-1\\
{\bf 1} & {\rm otherwise} 
\end{matrix}\right\}\label{evenyq} 
\end{eqnarray}
that are all non-commutative Laurent monomials of the initial data, times some possible integer powers of $q$.
As such it is a Laurent polynomial of the initial data with
coefficients in $\Z_+[q,q^{-1}]$.
\end{thm}
\begin{proof}
We apply Theorem \ref{finlut} with the weights
\eqref{ncweightcalc}. We then rewrite all the skeleton 
weights using Lemma \ref{uselem}.
Finally, the redundant weights
\eqref{redunc} are also monomials of the initial data, with a possible integer power of $q$ in factor.
Reordering and rearranging the terms may only produce coefficients in $\Z_+[q,q^{-1}]$.
\end{proof}

\begin{example}
In the case of non-decreasing Motzkin paths (see Example \ref{incmotz} above),
we have the following expressions for the skeleton weights:
\begin{eqnarray*}
\by_{2i-1}(\bm)&=&\left\{ \begin{matrix}
\bR_{i,m_i+1}\bR_{i-1,m_{i-1}+1}^{-1}
\bR_{i,m_i}^{-1} \bR_{i-1,m_{i-1}}
& {\rm if} \ m_i=m_{i-1}\\
\bR_{i,m_i+1}
\bR_{i-1,m_{i-1}+1}^{-1}\bR_{i-1,m_{i-1}}\bR_{i,m_i}^{-1} 
& {\rm if} \ m_i=m_{i-1}+1
 \end{matrix}\right.\\
\by_{2i}(\bm)&=& \bR_{i+1,m_i+1}\bR_{i,m_i+1}^{-1}
 \bR_{i,m_i}^{-1}\bR_{i-1,m_i}
 \end{eqnarray*}
where we have used the $q$ commutations \eqref{genqcomm}
and the identity $\lambda_{i-1,i}-\lambda_{i-1,i-1}+i-1=0$ 
to simplify the monomials.
For the fundamental initial condition $\bx_0=\bx_{\bm_0}$, this gives:
\begin{eqnarray*}
\by_{2i-1}(\bm_0)&=&
\bR_{i,1}\bR_{i-1,1}^{-1}\bR_{i,0}^{-1}\bR_{i-1,0}\\
\by_{2i}(\bm_0)&=& \bR_{i+1,1}\bR_{i,1}^{-1}\bR_{i,0}^{-1}\bR_{i-1,0}
\end{eqnarray*}
For the Stieltjes point $\bx_{\bm_1}$, we have:
\begin{eqnarray*}
\by_{2i-1}(\bm_1)&=&\bR_{i,i}\bR_{i-1,i-1}^{-1}\bR_{i-1,i-2}\bR_{i,i-1}^{-1}\\
\by_{2i}(\bm_1)&=&\bR_{i+1,i}\bR_{i,i}^{-1}\bR_{i,i-1}^{-1}\bR_{i-1,i-1}
\end{eqnarray*}
\end{example}

\begin{example}
Let us consider the case $r=2$ of Example \ref{atwoex}. Setting $\bR_n=\bR_{1,n}$ and $\bS_n=\bR_{2,n}$,
we have the following commutation relations:
\begin{equation}\label{coma2}
\begin{matrix}\bR_n\bR_{n+1}=q^2 \, \bR_{n+1}\bR_n  &  \bS_n\bS_{n+1}=q^2 \, \bS_{n+1}\bS_n &
\bR_n\bS_{n}=\bS_{n}\bR_n \\
\bS_n\bR_{n+1}=q \, \bR_{n+1}\bS_n & \bR_n\bS_{n+1}=q \, \bS_{n+1}\bR_n & \ \ \ (n\in \Z)\end{matrix}
\end{equation}
The $A_2$  quantum $Q$-system reads:
\begin{eqnarray*}
q^2\bR_{n+1}\bR_{n-1}&=& (\bR_n)^2 +\bS_n\\
q^2\bS_{n+1}\bS_{n-1}&=& (\bS_n)^2 +\bR_n
\end{eqnarray*}
and the weights at the various Motzkin paths may be expressed as:
$$ \begin{matrix}
\by_1(0,0)&=&\bR_1\bR_0^{-1}\\
\by_2(0,0)&=&\bS_1\bR_1^{-1}\bR_0^{-1}\\
\by_3(0,0)&=&\bS_1\bR_1^{-1}\bS_0^{-1}\bR_0\\
\by_4(0,0)&=&\bS_1^{-1}\bS_0^{-1}\bR_0\\
\by_5(0,0)&=&\bS_1^{-1}\bS_0
\end{matrix}
\qquad
\begin{matrix}
\by_1(0,1)&=&\bR_1\bR_0^{-1}\\
\by_2(0,1)&=&\bS_1\bR_1^{-1}\bR_0^{-1}\\
\by_3(0,1)&=&\bS_2\bR_1^{-1}\bR_0\bS_1^{-1}\\
\by_4(0,1)&=&\bS_2^{-1}\bS_1^{-1}\bR_1\\
\by_5(0,1)&=&\bS_2^{-1}\bS_1
\end{matrix}
\qquad
\begin{matrix}
\by_1(1,0)&=&\bR_2\bR_1^{-1}\\
\by_2(1,0)&=&\bS_0^{-1}(\bS_1)^2\bR_1^{-1}\bR_2^{-1}\\
\by_3(1,0)&=&\bS_1\bR_1\bS_0^{-1}\bR_2^{-1}\\
\by_4(1,0)&=&\bS_0^{-1}\bS_1^{-1}(\bR_1)^2\bR_2^{-1}\\
\by_5(1,0)&=&\bS_1^{-1}\bS_0
\end{matrix}
$$
where we have used \eqref{coma2} to eliminate the $q$-dependent prefactors.
The redundant weight of the last case $\bm=(1,0)$ is 
$\bw_{3,1}=t\by_4\by_3^{-1}\by_2=tq^{-1}\bR_2^{-1}\bS_0^{-1}$.
Note that the conservation law $\by_5(\bm)\by_3(\bm)\by_1(\bm)=C_3=q^3$ is satisfied for the three 
Motzkin paths.
\end{example}

\begin{remark} One can show that the conserved quantities $P_m$ of
Remark \ref{consrem} commute with each other in this case. This suggests an
analogy with the definition of classical integrability, where the
integrals of the motion are in involution. In fact, in the $A_1$ quantum $Q$-system case,
the Hamiltonian structure and integrability were shown in \cite{FV}. The case of the $A_1$ $T$-system
was treated in \cite{FKV}.
\end{remark}

\begin{remark}
  The quasi-Wronskians $\Delta_{i,n}$ are ``bar-invariant" quantities.
  Let us define an anti-automorphism (called usually the bar
  involution), which we denote by $*$, by $(ab)^*=b^* a^*$ for all
  $a,b\in\cA$, ${\bf 1}^*={\bf 1}$ and $q^*=q^{-1}$, while
  $(\bR_n)^*=\bR_n$ for all $n\in\Z$.  Then, as a trivial consequence
  of the Hirota equation \eqref{nchiro} and of the initial condition
  $\Delta_{1,n}=\bR_n$, we have that
$$(\Delta_{i,n})^*=\Delta_{i,n}$$ 
\end{remark}

\section{Other non-commutative systems}

In this section, we consider further examples of non-commutative recursions which
have the positive Laurent property. In the following,
all the variables we consider belong to some skew field of rational fractions $\cA$.

\subsection{Affine rank $2$ non-commutative systems}

The non-commutative $A_1$ $Q$-system \eqref{aoneNCQsystem} corresponds to one of the so-called
rank two affine cluster algebras, namely with initial exchange matrix $B_0$
related to affine Cartan matrices
of rank 2. These read $B_0=\begin{pmatrix} 0 & -c \\ b & 0 \end{pmatrix}$ with $bc=4$, $b,c\in\Z_{>0}$.
While the $A_1$ case corresponds to $b=c=2$, the two other cases $(b,c)=(1,4)$ or $(4,1)$
were shown in Ref. \cite{DFK09b}
to be related to the following non-commutative system for $\bR_n\in\cA$, $n\in\Z$:
\begin{equation}\label{onefour}
\bR_{n+1}(\bR_n)^{-1}\bR_{n-1} \bR_n= \left\{ \begin{matrix} (\bR_n)^b +1 \hfill & {\rm if} & n\, {\rm odd}\\
(\bR_n)^c +1 \hfill & {\rm if} & n\, {\rm even}
\end{matrix}\right.
\end{equation}
Let us fix $b=1,c=4$. 

\begin{thm}\label{onefourthm}
The solution $\bR_n$ of the system \eqref{onefour} is a positive Laurent polynomial of any initial
data $(\bR_i,\bR_{i+1})$, for any $i,n\in\Z$.
\end{thm}

We sketch the proof below.
As in the $A_1$ case, one finds two conserved quantities and a linear recursion relation for $\bR_n$
with constant coefficients.
Finally, the generating function for $\bR_{2n}\bR_0^{-1}$ reads:
$$\bF_0(t)=\sum_{n\geq 0} t^n \, \bR_{2n} \bR_0^{-1}
=\left(1-t \bz_1 -t^2(1-t \bz_3)^{-1} \bz_2\right)^{-1}$$
where 
\begin{eqnarray*}
\bz_1&=&\bu_1 \bu_0^{-1}=(\bR_1^{-1}+1)\bR_0^{-1}\bR_1\bR_0^{-1}\\
\bz_2&=&(\bK-\bz_1)\bz_1 -\bC= \big( \bR_1^{-1}\bR_0^2\bR_1^{-1}+
(1+\bR_1^{-1})\bR_0^{-2}(1+\bR_1^{-1})
\big)\bR_0^{-1}\bR_1\bR_0^{-1}\\
\bz_3&=&\bK-\bz_1= \bR_1^{-1}\bR_0^2+(1+\bR_1^{-1})\bR_0^{-2}
\end{eqnarray*}
and the Laurent positivity follows for $(\bR_0,\bR_1)$. 
The mutated expressions in terms of $(\bR_i,\bR_{i+1})$
are again obtained by manifestly positive non-commutative rearrangements of $\bF_0(t)$,
and the general Laurent positivity follows from symmetries of the system.

\subsection{Some rank $2k+1$ cases}

For any fixed $k\in\Z_+$,
we consider the ``rank $2k+1$" recursion relation for $n\in\Z$, for variables $\bu_n\in\cA$:
\begin{eqnarray}
\bu_{2n+2k+1} \bu_{2n}&=& \bu_{2n+1}\bu_{2n+2k}+{\bf 1}\label{rk3one}\\
\bu_{2n+1}\bu_{2n+2k+2}&=& \bu_{2n+2k+1}\bu_{2n+2} +{\bf 1}\label{rk3two}
\end{eqnarray}

Admissible initial data for this system are $\bx_i=(\bu_i,\bu_{i+1},\ldots, \bu_{i+2k})$ for all $i\in\Z$.
We have the following positivity theorem.

\begin{thm}\label{positrkthree}
The general solution $\bu_n$ to the system \eqref{rk3one}-\eqref{rk3two} is a positive Laurent
polynomial of any initial data $\bx_i$, for all $n,i\in\Z$.
\end{thm}

This is proved in the Appendix \ref{appendixb}, by  using the integrability of the system
and constructing non-commutative weighted path models for $u_n$. The commutative version of
the system \eqref{rk3one}-\eqref{rk3two} corresponds to the affine $A_{2k}$ cluster algebra \cite{AssemReit}.
Note that the exchange matrix of the fundamental seed is singular, so there is no quantum
cluster algebra version of the system \eqref{rk3one}-\eqref{rk3two}.

\subsection{A non-coprime rank 2 case related to rank 3}

The following non-commutative recursion relation
also satisfies the positive Laurent property:
\begin{equation}\label{firla}\bR_{n+1}(\bR_n)^{-1} \bR_{n-1}=\bR_n+1 \end{equation}
Note that we still have the quasi-commutation relation $\bR_n\bR_{n+1}=\bR_{n+1}\bC\bR_n$
from the conserved quantity $\bC$, and that we may equivalently write 
$\bR_{n+1}\bC\bR_{n-1}=\bR_n(\bR_n+1)$,
or 
\begin{equation}\label{lastla}\bR_{n-1}(\bR_n)^{-1}\bR_{n+1}=\bR_n+\bC^{-1}\end{equation}
In the commuting case, it would correspond to the recursion relation $R_{n+1}R_{n-1}=R_n(R_n+1)$
which does not correspond to a cluster algebra mutation, as the two terms on the r.h.s. are not coprime.
However, this relation may be connected to a non-commutative rank 3 system, with the
positive Laurent property,  as follows.
Introduce the change of variables:
$$ \bR_{2n+1}=\bu_{2n+1}\bu_{2n}, \qquad \bR_{2n}=\bu_{2n-1}\bu_{2n}$$
for some variables $\bu_n$, $n\in \Z$. We have:
\begin{eqnarray*} 
\bu_{2n+1}\bu_{2n-2}&=&\bu_{2n-1}\bu_{2n}+\bone\\
\bu_{2n-1}\bu_{2n+2}&=&\bu_{2n+1}\bu_{2n}+\bC^{-1}
\end{eqnarray*}
This is nothing but the $k=1$ system \eqref{rk3one}-\eqref{rk3two} with an additional 
fixed coefficient $\bC^{-1}$. Choosing the initial data $\bu_0=\bR_1$, $\bu_1=\bone$ and $\bu_2=\bR_2$,
with $\bC^{-1}=\bu_0\bu_2^{-1}\bu_0^{-1}\bu_2$, we may repeat the proof of Appendix 
\ref{appendixb} so as to include the fixed coefficient $\bC^{-1}$, and deduce the 
Laurent positivity property in terms of $(\bR_1,\bR_2)$.

\section{Towards a non-commutative Lindstr\"om-Gessel-Viennot Theorem}

All the cases we have been investigating so far are expressible in
terms of non-commutative weighted path models. More precisely, we have
been able to represent the generating function $\sum_{n\geq 0} t^n\bR_n\bR_0^{-1}$ as the partition
function for paths with non-commutative weights. In the application of the $A_r$ 
case to the quantum $A_r$ $Q$-system of Section \ref{section5}, we also have variables
$\bR_{\al,n}$, $n>1$, expressed as ``quantum Wronskians" (see Remark \ref{quantumwronsk}).
This is actually a non-commutative generalization of the ``Wronskian" expressions
for the $Q$-system solutions in the commuting case \eqref{discwronskco}, which take the form:
$${R_{\al,n+m_1}\over (R_{1,m_1})^\al}
=\det_{1\leq i,j\leq \al}\left( {R_{1,i+j-\al-1+n+m_1}\over R_{1,m_1} }\right) $$
In the cases $\bm=\bm_0$ or $\bm_1$ for instance
(fundamental initial data and Stieltjes point, both with $m_1=0$)
the Lindstr\"om-Gessel-Viennot theorem \cite{LGV1} \cite{LGV2} allows to interpret directly the l.h.s. as the partition
function for a family of $\al$ non-intersecting lattice paths,
namely paths traveled by walkers
starting at the root of the graph $\Gamma_\bm$ at times $0,2,...,2\al-4,2\al-2$ and
returning at times $2n+2\al-2,2n+2\al-4,...,2n+2,2n$, 
and such that no two walkers meet at any vertex of
$\Gamma_\bm$. In the case of general $\bm$, we have a more subtle application of the theorem,
in which the paths must be strongly non-intersecting (see Ref. \cite{DFK3} for details).

It would be extremely interesting to extend the classical 
Lindstr\"om-Gessel-Viennot theorem
to the case of paths with non-commutative weights.
Such a generalization was introduced in \cite{FG} in the context 
of non-commutative Schur functions,
but does not seem to apply directly to the present framework.

So far, we have found such an interpretation only in the fully non-commutative 
$A_1$ $Q$-system case. We wish to interpret the $Q$-system relation
$\bR_{n+1}\bC\bR_{n-1}-\bR_n^2=1$
in terms of pairs of non-intersecting paths on the graph $\Gamma_{(0)}$, i.e. the integer segment $[1,4]$.
Let $\bP_n=\bR_n\bR_0^{-1}$ denote the partition function of paths from vertex $0$ to itself
on $\Gamma_{(0)}$ with the non-commutative weights $\by_1,\by_2,\by_3$ respectively for the steps
$2\to 1$, $3\to 2$ and $4\to 3$, while the up steps $1\to2$, $2\to 3$ and $3\to 4$
receive the trivial weight $1$. Multiplying the $Q$-system relation by $\bC$, we have:
\begin{equation}\label{cnilp}
\bP_{n+1}\bR_0\bC\bP_{n-1}\bR_0\bC-\bP_n\bR_0\bP_n\bR_0\bC=\bC
\end{equation}
To proceed, we note the following important relations between the weights $\by_1,\by_2,\by_3$
and $\bR_0$, $\bC$:
\begin{equation}\label{commutaone}
(\by_3)^n\by_2\by_1\bR_0\bC(\by_1)^n=\by_2\by_1\bR_0\bC=\bR_0^{-1} \qquad 
(\by_1)^n \bR_0 (\by_3)^n =\bR_0 \quad (n\in \Z)
\end{equation}

\begin{figure}
\centering
\includegraphics[width=11.cm]{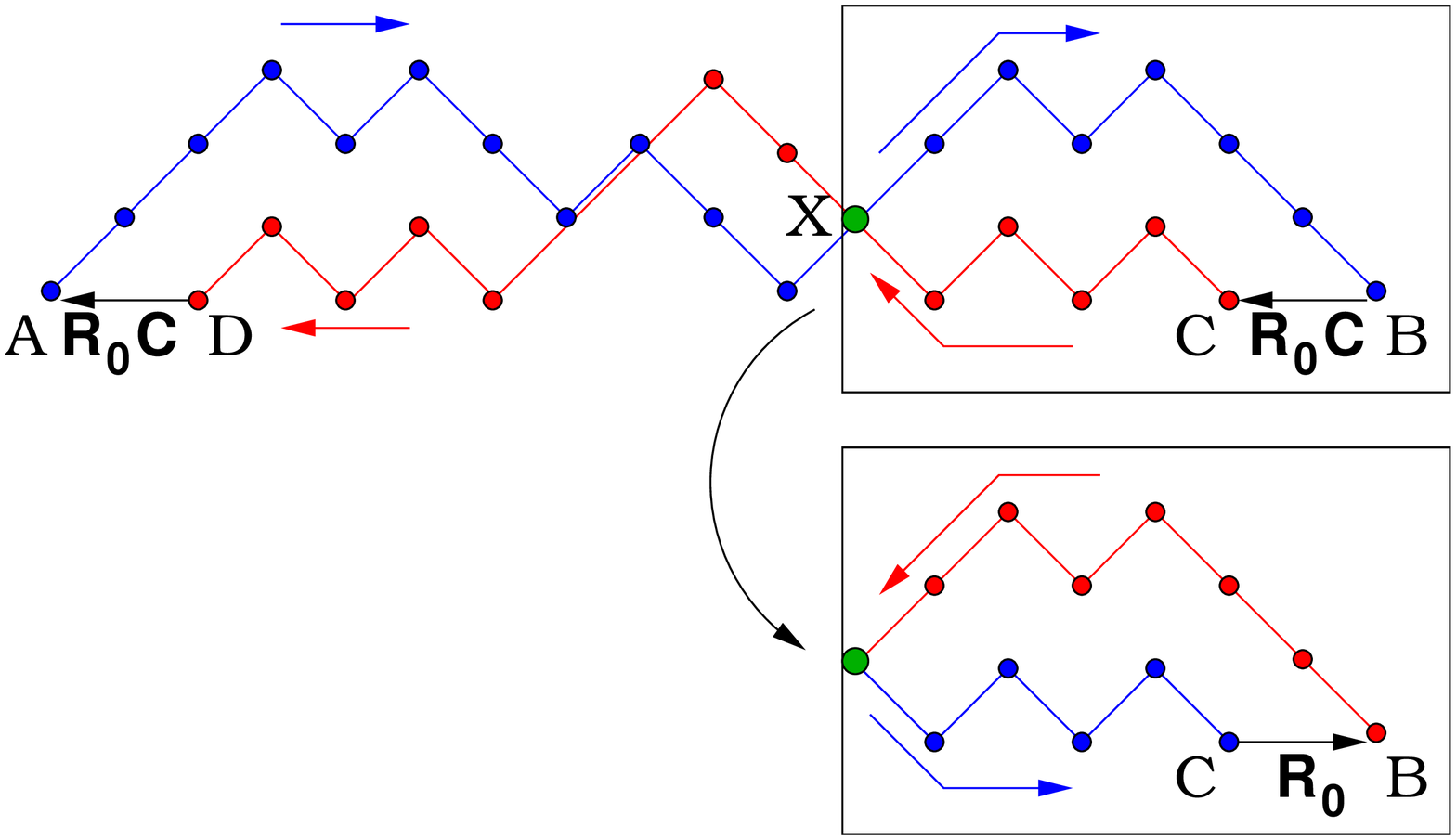}
\caption{\small A typical contribution to $\bP_{n+1}\bR_0\bC\bP_{n-1}\bR_0\bC$ for $n=8$ is made of
a pair of paths $A\to B$ (blue, left to right) and $C\to D$ (red, right to left), 
with two intermediate stations at vertex $1$, $B\to C$ and $D\to A$. 
When paths intersect (like here), this may also be
interpreted as a contribution to $\bP_n\bR_0\bP_n\bR_0\bC$, upon reversing the direction of travel of the 
portions of paths to the right of their rightmost intersection $X$, 
with a weight $\bR_0$ for the step $C\to B$.}
\label{fig:lgv}
\end{figure}

We now interpret the quantity $\bP_{n+1}\bR_0\bC\bP_{n-1}\bR_0\bC$ as the partition function for 
a path loop $A\to B\to C\to D\to A$ (see Fig. \ref{fig:lgv}), where the path $A\to B$ 
(partition function $\bP_{n+1}$) is traveled from left to right
and assigned the usual descending step weights $\by_i$ for a step $i+1\to i$, while the 
portions $B\to C$, $D\to A$ receive the weight $\bR_0\bC$, and the path $C\to D$
(partition function $\bP_{n-1}$)
is traveled from right to left, and receives the usual descending step weights 
$\by_i$ for a step $i+1\to i$.
The contribution for the configuration in the top of
Fig.\ref{fig:lgv} is $\by_3^2\by_2^2\by_1\by_3^2\by_2\by_1\bR_0\bC \by_1^2\by_3\by_2\by_1^3\bR_0\bC$.

\begin{thm} The quantity
\begin{equation}\label{cancelaone}
\bP_{n+1}\bR_0\bC\bP_{n-1}\bR_0\bC-\bP_n\bR_0\bP_n\bR_0\bC
\end{equation}
is the partition function of loops $A\to B\to C\to D\to A$
formed by pairs of {\emph non-intersecting} paths on $[1,4]$ respectively from
$A=(0,1)$ to $B=(2n+2,1)$, then from $C=(2n,1)$ to $D=(2,1)$, with the two 
intermediate stations $B\to C$ and $D\to A$ at the root receiving the weight $\bR_0\bC$.
\end{thm}
\begin{proof}
Recall first that the portions of loop from $B\to C$, and $D\to A$ receive the weight $\bR_0\bC$.
We define the weight for the same portions, but traveled in the other direction ($C\to B$ and $A\to D$)
to be $\bR_0$.
Let us assume the two paths $B\to C$, and $D\to A$ intersect, namely that they have at least a common vertex.
Let us denote by $X$ the rightmost such common vertex.

We may now decompose
the above loop into three pieces: that from $A\to X$, $X\to B\to C\to X$ and that from $X\to A$.
The central portion $X\to B\to C\to X$, with weight $W_{XBCX}$, may be ``flipped" into 
$X\to C\to B\to X$,
with new weight $W_{XCBX}$, by 
reversing the direction of travel, so that 
the new configuration now corresponds to pairs of paths $A\to C$ and $B\to D$
of equal length $2n$, now separated by the two intermediate stations $C\to B$ and $D\to A$ 
with respective weights $\bR_0$ and $\bR_0\bC$, thus contributing
to $\bP_n\bR_0\bP_n\bR_0\bC$. 

We now note that $W_{XCBX}=W_{XBCX}$.
Indeed, depending on the height of $X$, we have the two following possibilities
(the first case is illustrated in Fig.\ref{fig:lgv}):
\begin{itemize}
\item (i) $X=(2t+1,2)$ for some $t$, then the blue path $X\to B$
is $3\to 3\to 4 \to 3\to 4\to \cdots \to 4\to 3\to 2$ (weight $\by_3^m \by_2\by_1$, $m=n-t$)
and the red path $C\to X$ is $1\to 2\to 1\to 2\to \cdots \to 1\to 2$ (weight $\by_1^m$).
The weight of the central loop $X\to B\to C\to X$ is:
$$W_{XBCX}= \by_3^m (\by_2\by_1 \bR_0\bC) \by_1^m=\bR_0^{-1} 
=\by_1\bR_0\by_2= \by_1^{m+1}\bR_0  \by_3^m\by_2=W_{XCBX}$$
by use of \eqref{commutaone}.
\item (ii) $X=(2t+2,3)$ for some $t$, then the blue path $X\to B$
is $3\to 4 \to 3\to 4\to \cdots \to 4\to 3\to 2$ (weight $\by_3^m \by_2\by_1$, $m=n-t$)
and the red path $C\to X$ is $1\to 2\to 1\to 2\to \cdots \to 1\to 2\to 3$ (weight $\by_1^{m-1}$).
The weight of the central loop $X\to B\to C\to X$ is:
$$ W_{XBCX}= \by_3^m (\by_2\by_1 \bR_0\bC) \by_1^{m-1} =\by_3\bR_0^{-1}
=\bR_1^{-1}=\by_2\bR_0=\by_2 \by_1^{m}\bR_0 \by_3^m =W_{XCBX}$$
by use of \eqref{commutaone}.
\end{itemize}

The above flipping procedure is a bijection between pairs of paths $A\to B$ and $C\to D$ that
intersect and pairs of paths $A\to C$ and $B\to D$. They cancel each other in the formula
\eqref{cancelaone}, leaving us only with non-intersecting paths, and the theorem follows.
\end{proof}

Note that there is a unique configuration of non-intersecting paths, where $A\to B$ is 
$1\to 2\to 3\to 4\to 3\to \cdots \to 4\to 3\to 2\to 1$ and $C\to D$ is $1\to 2\to 1\to 2\to \cdots \to 2\to 1$.
The corresponding weight is $\by_3^{n-1} (\by_2\by_1 \bR_0\bC) \by_1^{n-1}\bR_0\bC=\bC$ by
\eqref{commutaone}, and we recover \eqref{cnilp}.

\section{Conclusion: Towards non-commutative cluster algebras}

The rank 2 non-commutative cluster algebras are defined as follows.
Let us consider a bipartite
chain ${\mathbb T}_2$ (Cayley tree of degree 2), with alternating black and white vertices
indexed by $v\in\Z$ (black if $v$ is even, white otherwise), and 
edges with alternating labels $1,2$, and a fixed element $\bC\in\cA$.
To each vertex $v$ of ${\mathbb T}_2$, we attach a cluster, i.e. a vector of the form $\bx_v=(\ba,\bb)$ in $\cA^2$,
where the cluster variables $\ba,\bb$ obey the quasi-commutation relations
$$\begin{matrix}\ba\, \bb=\bb\bC \ba \hfill & {\rm if} &v\, {\rm is}\, {\rm black}\\
\bb\, \ba=\ba\bC \bb \hfill & {\rm if} &v\, {\rm is}\, {\rm white} \end{matrix}$$

To each (doubly oriented) edge labeled $i$ connecting a black and a white vertex,
we associate the two following mutations: (i) the forward mutation $\mu_i^{+}$ corresponding to the oriented
edge from the black vertex to the white (ii) the backward mutation  $\mu_i^{-}=(\mu_i^{+})^{-1}$
corresponding to the oriented edge from the white vertex to the black.
Depending on the color of the vertex $v$,
the forward and backward mutations of $x_v=(\ba,\bb)$ are 
defined as follows:
$$\begin{matrix} v\, {\rm black}: &\mu_1^{+}:(\ba,\bb)\mapsto (\ba',\bb)\hfill &{\rm with} & \ba'\bC \ba =(\bb)^b+\bone \hfill\\
&\mu_2^{-}:(\ba,\bb)\mapsto (\ba,\bb')\hfill &{\rm with} & \bb\bC \bb' =(\ba)^c+\bone \hfill\\
v\, {\rm white}: &\mu_2^{+}:(\ba,\bb)\mapsto (\ba,\bb')\hfill &{\rm with} & \bb'\bC \bb =(\ba)^c+\bone \hfill\\
&\mu_1^{-}:(\ba,\bb)\mapsto (\ba',\bb)\hfill &{\rm with} & \ba\bC \ba' =(\bb)^b+\bone \hfill
\end{matrix}$$
These definitions are compatible with an exchange matrix equal to $B_\bullet=\begin{pmatrix} 0 & -c\\ b &0\end{pmatrix}$
for black vertices and $B_\circ=-B_\bullet$ for white ones, evolving according to the rules of the usual rank 2 cluster
algebras.
The general conjecture of \cite{Kon} may be rephrased as follows: {\it for any positive integers $b,c$,
the cluster variables at any node of ${\mathbb T}_2$ may be expressed as positive Laurent polynomials of 
the cluster variables at any other node of ${\mathbb T}_2$}.
This is easily proved in the ``finite cases" for $bc<4$ (leading to only finitely many distinct cluster
variables in the commuting case), where one can show explicitly that there exists a positive integer $m_{b,c}$ 
such that $\bx_{v+m_{b,c}}=\bC\bx_v\bC^{-1}$, with $m_{1,1}=5$, $m_{1,2}=6$, and $m_{1,3}=8$,
leaving us with only finitely many cluster variables to inspect. In the affine case $bc=4$, this was proved in
\cite{DFK09b} (see the $A_1$ case $b=c=2$ in the introduction,
and the case $b=1,c=4$ Theorem \ref{onefourthm} above), 
while the case $bc>4$ is still open.

To generalize this construction, we may think of using the results of the present paper in the following way.
Assuming that the $A_r$ $Q$-system cluster algebra introduced in \cite{Ke07} may be generalized to a
non-commutative setting, we would expect the relevant cluster variables to still be expressible
in terms of non-commutatively weighted paths on target graphs, in terms of given initial data. In this paper,
we have shown how to introduce mutations of non-commutative weights
via non-commutative continued fractions. 
Concentrating first on the Stieltjes point,  we may trace 
the expressions for $\Delta_{1,n}=\bR_{1,n}$ back to the initial variables 
$\Delta_{\al,\al}$ and $\Delta_{\al,\al-1}$, $\al=1,2,...,r+1$ (see Theorem \ref{stilnc}). 
This seems to indicate, in this case at least, that
the quasi-Wronskians are good non-commutative cluster variables.
However, the general expressions of Theorem \ref{ncweightcalc} for $\bR_{1,n}$
in terms of quasi-Wronskians or the more compact $C,D$ variables
(see \eqref{defCD} and \eqref{evalC}) do not clearly
display which initial data they are related to. We clearly need another definition 
of the non-commutative cluster variables.

We hope to report on this issue in a later publication.

\begin{appendix}
\section{Relation of weighted graphs $G_\bm$ to the
  graphs $\Gamma_\bm$ of \cite{DFK08}}\label{appendixa}
In \cite{DFK08}, we introduced a family of weighted graphs $\Gamma_\bm$
constructed from a Motzkin path $\bm$. Here we recall the definition of these graphs, as well as the sequence of bijections which brings them to form of the weighted graphs $G_\bm$ which we use in the main body of the paper. It is important to understand this connection, because the most basic proof of positivity relies on the fact that the graphs $\Gamma_\bm$ are positively weighted with weights which are monomials in the initial data $R(\bm)$.

In fact, we describe a sequence of simple bijections
$$
\bm \leftrightarrow \Gamma_\bm \leftrightarrow \Gamma'_\bm \leftrightarrow G_\bm.
$$
The graph $\Gamma_\bm'$ is the ``compactified" graph of \cite{DFK09}. Note that in general, only the graph $\Gamma_\bm$ has manifestly positive weights.

Below, let $\bm$ be a fixed Motzkin path of length $r$.
\subsection{The graphs $\Gamma_\bm$}
\begin{enumerate}
\item Decompose the Motzkin path $\bm$ into a disjoint union of maximal ``descending" segments, including those of length one, such that $\bm = \underset{i}{\sqcup} (m_i,m_{i+1}, ..., m_{i+k_i-1})$. Here, 
 $m_{i+a} = m_{i}-a$, $0\leq a < k$ in the segment beginning with index $i$ and having length $k$. 
 \begin{example}\label{samplemot}
 If $\bm = (2,1,2,2,3,2,1,1)$ then 
 $\bm = (2,1)\sqcup (2)\sqcup(2) \sqcup (3,2,1) \sqcup (1)$.
 \end{example}
 For notational purposes, we call the maximal descending segment starting with $\bm_i$ $\bm^{(i)}$, so that $\bm = \sqcup_i \bm^{(i)}$.
\item To each maximal descending segment $\bm^{(i)}$ of length, say, $k$, is associated a graph
$\Gamma_{\bm^{(i)}}$. It has $2(k+1)$ vertices, labeled $(i,i+1,...,i+k)$ and
$(i',(i+1)',...,(i+k)')$. It has the following list of 
oriented edges: $(a,a+1), (a+1,a), (a,a'), (a',a)$ (for all
$a$), and $(b,a)$ for all $k+i\geq b>a+1 \geq 2$. Below we give a picture of the graph $\Gamma_{\bm^{(1)}}$ for the example of the descending subpath $\bm^{(1)}=(3,2,1,0)$.
\begin{center}
\begin{tabular}{ccc}
\psset{unit=2mm,linewidth=.4mm,dimen=middle}
\begin{pspicture}(0,-2)(15,22)
\psline(0,17.5)(15,2.5)
\multips(0,17.5)(5,-5){4}{\pscircle*[linecolor=blue](0,0){.7}}
\rput(2,17.5){$m_k$}
\rput(8,12.5){$m_{k-1}$}
\rput(17,2.5){$m_1$}
\rput(7.5,-2){$\mathbf m^{(1)}$}
\end{pspicture}
&
\psset{unit=2mm,linewidth=.4mm,dimen=middle}
\begin{pspicture}(0,-2)(5,22)
\psline[arrowsize=1.5]{->}(0,10)(5,10)
\end{pspicture}&
\psset{unit=2mm,linewidth=.4mm,dimen=middle}
\begin{pspicture}(0,-2)(12,22)
\psline(5,0)(5,20)
\multips(0,0)(0,5){5}{\psline(5,0)}
\multips(0,0)(0,5){5}{\pscircle*[linecolor=red](0,0){.7}}
\rput(0,10){\psarc[linecolor=blue,arrowsize=1.5]{-}{11.5}{-60}{0}}
\rput(0,10){\psarc[linecolor=blue,arrowsize=1.5]{<-}{11.5}{0}{60}}
\rput(1,7.5){\psarc[linecolor=blue]{-}{8.75}{-60}{0}}
\rput(1,7.5){\psarc[linecolor=blue,arrowsize=1.5]{<-}{8.75}{0}{60}}
\rput(2.5,5){\psarc[linecolor=blue]{-}{6}{-60}{0}}
\rput(2.5,5){\psarc[linecolor=blue,arrowsize=1.5]{<-}{6}{0}{60}}
\rput(1,12.5){\psarc[linecolor=blue]{-}{8.75}{-60}{0}}
\rput(1,12.5){\psarc[linecolor=blue,arrowsize=1.5]{<-}{8.75}{0}{60}}
\rput(2.5,10){\psarc[linecolor=blue]{-}{6}{-60}{0}}
\rput(2.5,10){\psarc[linecolor=blue,arrowsize=1.5]{<-}{6}{0}{60}}
\rput(2.5,15){\psarc[linecolor=blue]{-}{6}{-60}{0}}
\rput(2.5,15){\psarc[linecolor=blue,arrowsize=1.5]{<-}{6}{0}{60}}
\multips(5,0)(0,5){5}{\pscircle*[linecolor=red](0,0){.7}}

\rput(7,0){1}
\rput(-2,0){1'}
\rput(7,5){2}
\rput(-2,5){2'}
\rput(9,20){$k+1$}
\rput(-5,20){$(k+1)'$}
\rput(2,-3){$\Gamma_{\mathbf m^{(1)}}$}
\end{pspicture}
\end{tabular}
\end{center}
\vskip.1in
(In this picture, an unoriented edge on the graph $\Gamma$ is considered to be doubly-oriented, for ease of reading.)
\item Consecutive maximal descending segments 
 are separated by either a ``flat'' step of the form $m_a=m_{a+1}$, or an
``ascending'' step, $m_a+1=m_{a+1}$.
\begin{example}
For the Motzkin path in Example \ref{samplemot}, the step separating $\bm^{(1)}=(2,1)$ and $\bm^{(3)}=(2)$ is ascending, and the step separating $\bm^{(3)}=(2)$ from $\bm^{(4)}=(2)$ is flat.
\end{example}

We ``glue" consecutive graphs $\Gamma_1=\Gamma_{\bm^{(i)}}$ and $\Gamma_2=\Gamma_{\bm^{(j)}}$ where $m^{(i)}$ and $m^{(j)}$ are neighboring descending segments in $\bm$ with $j>i$. The gluing procedure is dictated by whether the segments are separated by an ascending or flat step. We assume $\bm^{(i)}$ has length $k_i$ and $\bm^{(j)}$ has length $k_j$.

\begin{enumerate}
\item If the separation between the segments is flat, we identify the edge and vertices $(j',j)\in \Gamma_2$ with $((k_i+i)',k_i+i)\in \Gamma_1$.
 
\item If the separation between the segments is ``ascending'', we
  identify the vertices and edges $(j,j')\in\Gamma_2$ with $((k_i+i)',k_i+i)\in \Gamma_1$ (i.e., the identification is in the opposite sense). 

\end{enumerate}

\end{enumerate}
We define $\Gamma_\bm$ to be the result of gluing the graphs
corresponding to consecutive strictly descending pieces of $\bm$
according to the procedure given in (4) above. 
\begin{example}\label{exzero}
The purely ascending Motzkin path $\bm=\bm_1=(0,1,...,k)$ maps to the chain with $2k+2$ vertices, connected by doubly-oriented edges. The flat Motzkin path $\bm_0=(0,...,0)$ of length $k$ maps to the graph
\begin{center}

\psset{unit=2mm,linewidth=.4mm,dimen=middle}
\begin{pspicture}(0,-2)(12,22)
\psline(5,0)(5,20)
\multips(0,0)(0,5){5}{\psline(5,0)}
\multips(0,0)(0,5){5}{\pscircle*[linecolor=red](0,0){.7}}
\multips(5,0)(0,5){5}{\pscircle*[linecolor=red](0,0){.7}}

\rput(7,0){1}
\rput(-2,0){1'}
\rput(7,5){2}
\rput(-2,5){2'}
\rput(9,20){$k+1$}
\rput(-5,20){$(k+1)'$}
\rput(-10,10){$\Gamma_{\mathbf m_0}=$}
\end{pspicture}
\end{center}
As an example of the gluing procedure, in the case of $r=3$ and $\bm=(0,1,0)$, we have
\begin{center}

\psset{unit=2mm,linewidth=.4mm,dimen=middle}
\begin{pspicture}(0,-2)(12,22)
\psline(0,0)(5,0)(5,20)(0,20)
\psline(0,15)(5,15)
\rput(2.5,15){\psarc[linecolor=blue]{-}{6}{-60}{0}}
\rput(2.5,15){\psarc[linecolor=blue,arrowsize=1.5]{<-}{6}{0}{60}}
\multips(0,15)(0,5){2}{\pscircle*[linecolor=red](0,0){.7}}
\rput(0,0){\pscircle*[linecolor=red](0,0){.7}}
\multips(5,0)(0,5){5}{\pscircle*[linecolor=red](0,0){.7}}
\rput(-10,10){$\Gamma_{(0,1,0)}=$}
\end{pspicture}
\end{center}
\end{example}
$\Gamma_\bm$ has a root vertex, which is the vertex $1'$ in the picture of $\Gamma_{\bm_0}$, or the bottom-most vertex in $\Gamma_\bm$ in general. 

The weights attached to the graph $\Gamma_\bm$ are as follows.  
Each edge connecting neighboring vertices and pointing towards the root vertex $1'$ carries weight 
$t y_a$, with the edges numbered $a=1,2,...,2r+1$ from the root vertex on 
(with edge $(i',i)$ preceding $(i+1,i)$ in cases where it exists). 
The weights $y_a$ are called the {\it skeleton weights}
of $\Gamma_\bm$.

Let $w_{a,b}$ be the weight corresponding to the edge connecting vertex $a$ to $b$. Then $w_{a,a+1}=w_{a,a'}=1$, $w_{a+1,a}=ty_{2a} $ and $w_{a',a}=t y_{2a-1}$.
The weights $w_{a,b}$ of the edges $(a,b)$ with $(a-b>1)$ are called {\it redundant weights},
and are expressed solely in terms of skeleton weights between the vertices $a,b$ as:
\begin{equation}\label{redunco}
w_{a,b}= t\prod_{i=b+1}^{a}y_{2i-2}
\prod_{i=b+1}^{a-1}y_{2i-1}^{-1}, (a-b\geq 1).
\end{equation}

These graphs have manifestly positive weights.  Let $\bm$ be a Motzkin path with first component $m_1$. Then:
\begin{thm}\cite{DFK08}
The ratio $R_{1,n+m_1}/R_{1,m_1}$
is the coefficient of $t^n$ in the partition function of paths on
$\Gamma_\bm$ from the vertex $1'$ to itself. 
\end{thm}

This theorem can be stated in terms of the resolvent of the transfer matrix 
associated with the graph in the following way. Let $T_{\Gamma_\bm}$ be the matrix with entries 
$(T_{\Gamma_\bm})_{a,b} = w_{a,b}$. 
Then the generating function for paths from vertex $1'$ to itself on $\Gamma_\bm$ is the resolvent
\begin{equation}\label{resolvent}
F_\bm(t) = \left ((\mathbb I - T_{\Gamma_\bm})^{-1}\right)_{1',1'} =\sum_{n\geq 0} t^n R_{1,n+m_1}/R_{1,m_1} .
\end{equation}

\subsection{Compactified graphs}
We now define the graphs $\Gamma'_\bm$, which have the property that 
$$
((\mathbb I - T_{\Gamma_\bm})^{-1})_{1,1}=((\mathbb I - T_{\Gamma'_\bm})^{-1})_{1,1}.
$$
Note that $Z_{1',1'}(t) = 1 + t y_{1} ((\mathbb I - T_{\Gamma'_\bm})^{-1})_{1,1}$.

The compactified graphs $\Gamma'_\bm$ are obtained from $\Gamma_\bm$
as explained in Section 6 of \cite{DFK09}. This process involves an identification of pairs of vertices in the graph $\Gamma_\bm$, such that the final number of vertices is reduced to $r+1$. This  produces loops and also, in certain cases, extra directed edges. 

\begin{enumerate}
\item For subgraphs arising from descending or flat segments of $\bm$, compactified graphs are obtained simply from identification of the vertices $i$ and $i'$ for each $i$, resulting in a loop at vertex $i$ with weight $w_{i',i}$.  All other edges retain their weights.

\item Consider the vertices and edges arising from each ascending segment of length $k$ in $\bm$. This is a chain of $2k+2$ vertices. 
\begin{enumerate}
\item Name the vertices from bottom to top, $1,...,2k+2$ and identify the vertices $2a-1,2a$ ($a = 1,...,k+1$). This leaves $k+1$ vertices.
\item Rename the vertices from bottom to top again, $1,...,k+1$. 
\item Identification of vertices results in a loop at each vertex $a$ with weight inherited from the weight associated with the edge connecting $2a$ to $2a-1$, $w_{2a,2a-1}$. The spine retains weights $w_{2a+1,2a}$ along each edge $(a+1,a)$.
\item In addition we adjoin ascending directed edges, connecting vertex $a$ to $b$ for all $1\leq a<b\leq k+1$, with weight $(-1)^{b-a+1}$.
\end{enumerate}
\end{enumerate}

\begin{remark}\label{appremark}
In fact, the construction of the compactified graph directly from the Motzkin path is easy to describe. The Motzkin path is decomposed into maximal ascending, descending and flat segments, the corresponding compactified graphs and their weights are constructed as above, and the graphs are glued: The top two vertices corresponding to a segment are identified with the bottom two vertices corresponding to the next segment. Incident edges between these vertices are identified. (See pictures below.)
\end{remark}

\begin{lemma}\cite{DFK09} We have the equality of partition functions
$$
\left((\mathbb I - T_{\Gamma_\bm})^{-1}\right)_{1,1}=\left((\mathbb I - T_{\Gamma'_\bm})^{-1}\right)_{1,1}
$$
\end{lemma}

In light of remark \ref{appremark}, we can focus on graph segments corresponding to segments of $\bm$ which do not change direction.

\subsection{The graph $G_\bm$}
We now turn to the equivalence of resolvents of the graph $G_\bm$, defined in Section 2,  to those on the graph $\Gamma'_\bm$. The difference is the presence of long directed edges in $\Gamma'_\bm$ (connecting vertices along the spine of the graph with indices which differ by 2 or more).

On the level of the transfer matrix $T_{\Gamma'_\bm}$, these edges are represented by non-zero elements above or below the second diagonals. 
These entries can be eliminated using row (column) elimination for descending (ascending) edges, respectively. Conveniently,
\begin{lemma}
The matrices $E$ and $F$ such that $E(\mathbb I-T_{\Gamma'_\bm})F$ has no non-trivial entries away from the second diagonals are such that $E$ is upper unitriangular and $F$  is lower unitriangular.
\end{lemma}
The $(1,1)$-resolvent 
$$
((\mathbb I-T_{\Gamma'_\bm})^{-1})_{1,1}
$$
is obviously unchanged under such transformations.
\begin{proof}
It is sufficient to show this for each block of $T'=T_{\Gamma'_\bm}$ corresponding to maximal descending and maximal ascending  pieces of $\bm$. The case of flat pieces is trivial as it has no entries off the second diagonals.

\begin{enumerate}
\item {\bf Descending pieces.} The matrix $\mathbb I - T'$ related  to the vertices in the graph $\Gamma'_{\bm}$ corresponding to a descending path $(k-1,k-2,...,0)$ is the $(k+1)\times (k+1)$-matrix
$$\mathbb I - T' =
\left( \begin{array}{cccccc}
1-x_{1} & -1 &  0 & \cdots &  & 0\\
-x_2 & 1-x_3 & -1 & 0 & \cdots & \\
-\frac{x_2 x_4}{x_3} & -x_4 & 1-x_5 & -1 & 0 & \\
-\frac{x_2 x_4 x_6}{x_3 x_5} & -\frac{x_4 x_6}{x_5}& -x_6 &1-x_7 & -1 & \vdots\\
\vdots & & \ddots & \ddots & \ddots &\\
 & \cdots & & -x_{2k-2} & 1-x_{2k-1} & -1\\
\vdots & \cdots & &-\frac{x_{2k-2}x_{2k}}{x_{2k-1}} & -x_{2k} & 1-x_{2k+1}
\end{array}
\right).
$$
Here, $x_{2a-1} =w_{a',a}$ and $x_{2a} = w_{a+1,a}$. Clearly, there is an lower triangular matrix which eliminates all of the entries below the second diagonal when multiplied on the right:
$$
F_{i,i} = 1, (1\leq i \leq k+1), \quad F_{i+1,i} = -\frac{x_{2i}}{x_{2i+1}} (1\leq k-1),\quad F_{i,j}=0 \hbox{ otherwise}.
$$
The resulting matrix $M=(\mathbb I- T')F$ has non-vanishing entries $M_{i,i+1}=-1$, and for $i\leq k-1$,
\begin{equation}\label{equals}
M_{ii} = 1- (x_{2i-1} - \frac{x_{2i}}{x_{2i+1}}),\quad
M_{i+1,i} = -\frac{x_{2i}}{x_{2i+1}}.
\end{equation}
In addition, the last two columns of the matrix $M$ are those of $\mathbb I-T'$.

Note that we have lost the positivity property of the weights in this case, although the generating function of paths is positive.  Note also that the weights attached to the edges connecting the last two vertices have not changed under this transformation.
\item {\bf Ascending pieces.} The matrix $T'$, attached  to the vertices in the graph $\Gamma'_{\bm}$ corresponding to an ascending path $(0,...,k-1)$ is the $(k+1)\times (k+1)$-matrix
$$\mathbb I - T' =
\left( \begin{array}{ccccccc}
1-x_{1} & -1 & 1 & -1 & 1&\cdots & \pm 1 \\
-x_2 & 1-x_3 & -1 & 1& -1 & \cdots &\\
0 & -x_4 & 1-x_5 & -1 & 1 & & \\
\vdots & 0 & -x_6 & 1-x_7 & -1& \cdots & \\
& & \ddots & \ddots & \ddots& \ddots & \vdots\\
\vdots & & & &-x_{2k-2} & 1-x_{2k-1} & -1\\
0 & \cdots & & & 0 & -x_{2k} & 1-x_{2k+1}
\end{array}
\right).
$$
There is obviously an upper triangular matrix $E$, corresponding to row operations on this matrix, or multiplication on the left, such that $M=E(\mathbb I -T')$ has no entries away from the second off-diagonals. The resolvent is equal to the resolvent of the matrix with non-vanishing entries $M_{i,i+1}=-1$, and if $i\leq k-1$,
\begin{equation}\label{crosses}
M_{ii} = 1-(x_i + x_{i+1}) , \quad M_{i+1,i} = -x_{2i} x_{2i+1}.
\end{equation}
The last two rows of $M$ are those of $\mathbb I-T'$. Note that in this case, we retain the positivity property of the weights.
\end{enumerate} 
\end{proof}
To summarize, we can draw the graphs corresponding to the sequence of transformations from $\bm$ (an ascending, descending or flat path) to $G_\bm$ in graphical notation:
\begin{enumerate}
\item
For ascending paths,
\begin{center}
\begin{tabular}{ccccccc}
\psset{unit=2mm,linewidth=.4mm,dimen=middle}
\begin{pspicture}(0,-2)(15,22)
\psline(15,17.5)(0,2.5)
\multips(15,17.5)(-5,-5){4}{\pscircle*[linecolor=blue](0,0){.7}}
\rput(12,17.5){$m_k$}
\rput(7,12.5){$m_{k-1}$}
\rput(2.5,2.5){$m_1$}
\rput(7.5,-2){$\mathbf m$}
\end{pspicture}
&
\psset{unit=2mm,linewidth=.4mm,dimen=middle}
\begin{pspicture}(0,-2)(5,22)
\psline[arrowsize=1.5]{->}(0,10)(5,10)
\end{pspicture}&
\psset{unit=2mm,linewidth=.4mm,dimen=middle}
\begin{pspicture}(0,-2)(12,22)
\psline(5,0)(5,20)
\psline(0,0)(5,0)
\psline(0,20)(5,20)
\multips(5,0)(0,2.5){9}{\pscircle*[linecolor=red](0,0){.5}}
\rput(0,0){\pscircle*[linecolor=red](0,0){.7}}
\rput(0,20){\pscircle*[linecolor=red](0,0){.7}}
\rput(7,0){2}
\rput(-2,0){1}
\rput(7,2.5){3}

\rput(9,20){$2k+1$}
\rput(-4,20){$2k+2$}
\rput(2,-3){$\Gamma_{\mathbf m}$}
\end{pspicture}
& \psset{unit=2mm,linewidth=.4mm,dimen=middle}
\begin{pspicture}(0,-2)(5,22)
\psline[arrowsize=1.5]{->}(0,10)(5,10)
\end{pspicture}
&
\psset{unit=2mm,linewidth=.4mm,dimen=middle}
\begin{pspicture}(0,-2)(12,22)
\psline(5,0)(5,20)
\multips(.25,-2.5)(0,5){5}{\pscircle(2.5,2.5){2}}
\rput(0,10){\psarc[linecolor=blue,arrowsize=1.5]{->}{11.5}{-60}{0}}
\rput(0,10){\psarc[linecolor=blue,arrowsize=1.5]{-}{11.5}{0}{60}}
\rput(1,7.5){\psarc[linecolor=blue,arrowsize=1.5]{->}{8.75}{-60}{0}}
\rput(1,7.5){\psarc[linecolor=blue,arrowsize=1.5]{-}{8.75}{0}{60}}
\rput(2.5,5){\psarc[linecolor=blue,arrowsize=1.5]{->}{6}{-60}{0}}
\rput(2.5,5){\psarc[linecolor=blue,arrowsize=1.5]{-}{6}{0}{60}}
\rput(1,12.5){\psarc[linecolor=blue,arrowsize=1.5]{->}{8.75}{-60}{0}}
\rput(1,12.5){\psarc[linecolor=blue,arrowsize=1.5]{-}{8.75}{0}{60}}
\rput(2.5,10){\psarc[linecolor=blue,arrowsize=1.5]{->}{6}{-60}{0}}
\rput(2.5,10){\psarc[linecolor=blue,arrowsize=1.5]{-}{6}{0}{60}}
\rput(2.5,15){\psarc[linecolor=blue,arrowsize=1.5]{->}{6}{-60}{0}}
\rput(2.5,15){\psarc[linecolor=blue,arrowsize=1.5]{-}{6}{0}{60}}
\multips(5,0)(0,5){5}{\pscircle*[linecolor=red](0,0){.7}}

\rput(7,0){1}
\rput(7,5){2}
\rput(9,20){$k+1$}
\rput(4,-4){$\Gamma'_{\mathbf m}$}

\end{pspicture}
&\psset{unit=2mm,linewidth=.4mm,dimen=middle}
\begin{pspicture}(0,-2)(5,22)
\psline[arrowsize=1.5]{->}(0,10)(5,10)
\end{pspicture}&
\psset{unit=2mm,linewidth=.4mm,dimen=middle}
\begin{pspicture}(0,-2)(10,22)
\psline(5,0)(5,20)
\multips(.25,-2.5)(0,5){5}{\pscircle(2.5,2.5){2}}
\rput(.75,0){$\times$}
\rput(.75,5){$\times $}
\rput(.75,10){$\times $}
\rput(5,2.5){$\times $}
\rput(5,7.5){$\times $}
\rput(5,12.5){$\times $}

\multips(5,0)(0,5){5}{\pscircle*[linecolor=red](0,0){.7}}

\rput(7,0){1}
\rput(7,5){2}
\rput(9,20){$k+1$}
\rput(4,-4){$G_{\mathbf m}$}

\end{pspicture}
\end{tabular}
\vskip.1in
\end{center}
\item
For descending paths,
\begin{center}
\begin{tabular}{ccccccc}
\psset{unit=2mm,linewidth=.4mm,dimen=middle}
\begin{pspicture}(0,-2)(15,22)
\psline(0,17.5)(15,2.5)
\multips(0,17.5)(5,-5){4}{\pscircle*[linecolor=blue](0,0){.7}}
\rput(2,17.5){$m_k$}
\rput(8,12.5){$m_{k-1}$}
\rput(17,2.5){$m_1$}
\rput(7.5,-2){$\mathbf m$}
\end{pspicture}
&
\psset{unit=2mm,linewidth=.4mm,dimen=middle}
\begin{pspicture}(0,-2)(5,22)
\psline[arrowsize=1.5]{->}(0,10)(5,10)
\end{pspicture}&
\psset{unit=2mm,linewidth=.4mm,dimen=middle}
\begin{pspicture}(0,-2)(12,22)
\psline(5,0)(5,20)
\multips(0,0)(0,5){5}{\psline(5,0)}
\multips(0,0)(0,5){5}{\pscircle*[linecolor=red](0,0){.7}}
\rput(0,10){\psarc[linecolor=blue,arrowsize=1.5]{-}{11.5}{-60}{0}}
\rput(0,10){\psarc[linecolor=blue,arrowsize=1.5]{<-}{11.5}{0}{60}}
\rput(1,7.5){\psarc[linecolor=blue]{-}{8.75}{-60}{0}}
\rput(1,7.5){\psarc[linecolor=blue,arrowsize=1.5]{<-}{8.75}{0}{60}}
\rput(2.5,5){\psarc[linecolor=blue]{-}{6}{-60}{0}}
\rput(2.5,5){\psarc[linecolor=blue,arrowsize=1.5]{<-}{6}{0}{60}}
\rput(1,12.5){\psarc[linecolor=blue]{-}{8.75}{-60}{0}}
\rput(1,12.5){\psarc[linecolor=blue,arrowsize=1.5]{<-}{8.75}{0}{60}}
\rput(2.5,10){\psarc[linecolor=blue]{-}{6}{-60}{0}}
\rput(2.5,10){\psarc[linecolor=blue,arrowsize=1.5]{<-}{6}{0}{60}}
\rput(2.5,15){\psarc[linecolor=blue]{-}{6}{-60}{0}}
\rput(2.5,15){\psarc[linecolor=blue,arrowsize=1.5]{<-}{6}{0}{60}}
\multips(5,0)(0,5){5}{\pscircle*[linecolor=red](0,0){.7}}

\rput(7,0){1}
\rput(-2,0){1'}
\rput(7,5){2}
\rput(-2,5){2'}
\rput(9,20){$k+1$}
\rput(-5,20){$(k+1)'$}
\rput(2,-3){$\Gamma_{\mathbf m}$}
\end{pspicture}
& \psset{unit=2mm,linewidth=.4mm,dimen=middle}
\begin{pspicture}(0,-2)(5,22)
\psline[arrowsize=1.5]{->}(0,10)(5,10)
\end{pspicture}
&
\psset{unit=2mm,linewidth=.4mm,dimen=middle}
\begin{pspicture}(0,-2)(12,22)
\psline(5,0)(5,20)
\multips(.25,-2.5)(0,5){5}{\pscircle(2.5,2.5){2}}
\rput(0,10){\psarc[linecolor=blue,arrowsize=1.5]{-}{11.5}{-60}{0}}
\rput(0,10){\psarc[linecolor=blue,arrowsize=1.5]{<-}{11.5}{0}{60}}
\rput(1,7.5){\psarc[linecolor=blue]{-}{8.75}{-60}{0}}
\rput(1,7.5){\psarc[linecolor=blue,arrowsize=1.5]{<-}{8.75}{0}{60}}
\rput(2.5,5){\psarc[linecolor=blue]{-}{6}{-60}{0}}
\rput(2.5,5){\psarc[linecolor=blue,arrowsize=1.5]{<-}{6}{0}{60}}
\rput(1,12.5){\psarc[linecolor=blue]{-}{8.75}{-60}{0}}
\rput(1,12.5){\psarc[linecolor=blue,arrowsize=1.5]{<-}{8.75}{0}{60}}
\rput(2.5,10){\psarc[linecolor=blue]{-}{6}{-60}{0}}
\rput(2.5,10){\psarc[linecolor=blue,arrowsize=1.5]{<-}{6}{0}{60}}
\rput(2.5,15){\psarc[linecolor=blue]{-}{6}{-60}{0}}
\rput(2.5,15){\psarc[linecolor=blue,arrowsize=1.5]{<-}{6}{0}{60}}
\multips(5,0)(0,5){5}{\pscircle*[linecolor=red](0,0){.7}}

\rput(7,0){1}
\rput(7,5){2}
\rput(9,20){$k+1$}
\rput(4,-4){$\Gamma'_{\mathbf m}$}

\end{pspicture}
&\psset{unit=2mm,linewidth=.4mm,dimen=middle}
\begin{pspicture}(0,-2)(5,22)
\psline[arrowsize=1.5]{->}(0,10)(5,10)
\end{pspicture}&
\psset{unit=2mm,linewidth=.4mm,dimen=middle}
\begin{pspicture}(0,-2)(10,22)
\psline(5,0)(5,20)
\multips(.25,-2.5)(0,5){5}{\pscircle(2.5,2.5){2}}
\rput(.75,0){$=$}
\rput(.75,5){$=$}
\rput(.75,10){$=$}
\rput(5,2.5){$=$}
\rput(5,7.5){$=$}
\rput(5,12.5){$=$}

\multips(5,0)(0,5){5}{\pscircle*[linecolor=red](0,0){.7}}

\rput(7,0){1}
\rput(7,5){2}
\rput(9,20){$k+1$}
\rput(4,-4){$G_{\mathbf m}$}

\end{pspicture}
\end{tabular}
\vskip.1in
\end{center}
\item 
For flat paths,
\begin{center}
\begin{tabular}{ccccccc}
\psset{unit=2mm,linewidth=.4mm,dimen=middle}
\begin{pspicture}(0,-2)(7,22)
\psline(0,17.5)(0,2.5)
\multips(0,17.5)(0,-5){4}{\pscircle*[linecolor=blue](0,0){.7}}
\rput(2,17.5){$m_k$}
\rput(3,12.5){$m_{k-1}$}
\rput(2,2.5){$m_1$}
\rput(4,-2){$\mathbf m$}
\end{pspicture}
&
\psset{unit=2mm,linewidth=.4mm,dimen=middle}
\begin{pspicture}(0,-2)(5,22)
\psline[arrowsize=1.5]{->}(0,10)(5,10)
\end{pspicture}&
\psset{unit=2mm,linewidth=.4mm,dimen=middle}
\begin{pspicture}(0,-2)(12,12)
\psline(5,0)(5,20)
\multips(0,0)(0,5){5}{\psline(5,0)}
\multips(0,0)(0,5){5}{\pscircle*[linecolor=red](0,0){.7}}
\multips(5,0)(0,5){5}{\pscircle*[linecolor=red](0,0){.7}}

\rput(7,0){1}
\rput(-2,0){1'}
\rput(7,5){2}
\rput(-2,5){2'}
\rput(9,20){$k+1$}
\rput(-5,20){$(k+1)'$}
\rput(2,-3){$\Gamma_{\mathbf m}$}
\end{pspicture}
& \psset{unit=2mm,linewidth=.4mm,dimen=middle}
\begin{pspicture}(0,-2)(5,12)
\psline[arrowsize=1.5]{->}(0,10)(5,10)
\end{pspicture}
&
\psset{unit=2mm,linewidth=.4mm,dimen=middle}
\begin{pspicture}(0,-2)(12,22)
\psline(5,0)(5,20)
\multips(.25,-2.5)(0,5){5}{\pscircle(2.5,2.5){2}}
\multips(5,0)(0,5){5}{\pscircle*[linecolor=red](0,0){.7}}

\rput(7,0){1}
\rput(7,5){2}
\rput(9,20){$k+1$}
\rput(4,-4){$\Gamma'_{\mathbf m}$}

\end{pspicture}
&\psset{unit=2mm,linewidth=.4mm,dimen=middle}
\begin{pspicture}(0,-2)(5,12)
\psline[arrowsize=1.5]{->}(0,10)(5,10)
\end{pspicture}&
\psset{unit=2mm,linewidth=.4mm,dimen=middle}
\begin{pspicture}(0,-2)(10,22)
\psline(5,0)(5,20)
\multips(.25,-2.5)(0,5){5}{\pscircle(2.5,2.5){2}}

\multips(5,0)(0,5){5}{\pscircle*[linecolor=red](0,0){.7}}

\rput(7,0){1}
\rput(7,5){2}
\rput(9,20){$k+1$}
\rput(4,-4){$G_{\mathbf m}$}

\end{pspicture}
\end{tabular}
\end{center}
\end{enumerate}
\vskip.1in
In the final graphs $G_{\bm}$ for ascending and descending Motzkin paths, 
the symbols $\times$ and $=$ along certain edges indicates that the weights 
have been renormalized from the initial values $x_{2i-1}$ and $x_{2i}$ to 
the values $x_{2i-1}-x_{2i}/x_{2i+1}$ and $x_{2i}/x_{2i+1}$ for $=$ (see Eq.\eqref{equals})
and to the values $x_{2i-1}+x_{2i}$ and $x_{2i}x_{2i+1}$ for $\times$ (see Eq.\eqref{crosses}).
For $x_i=ty_i$ with the appropriate labeling of skeleton weights,
these are precisely the weights $\widehat y$ of Definition \ref{weightdefco} of Section \ref{sectiontwo}. 

Gluing two path pieces, $\bm$ followed by $\bm'$ into a longer Motzkin path  corresponds to identifying the two top vertices (and incident edges) of $G_\bm$, the lower graph, with the bottom two of $G_{\bm'}$, the top graph. In the cases where the weights have been renormalized in the lower part of the upper path, the renormalized weights appear in the glued graph.

\section{Proof of Theorem \ref{positrkthree}}\label{appendixb}
Due to the invariances of the system \eqref{rk3one}-\eqref{rk3two}, we only need
to express the solutions $\bu_n$ in terms of the initial data 
$\bx_0=(\bu_0,\bu_1,\bu_2,...,\bu_{2k})$. Indeed, introducing the anti-automorphism $\varphi$
such that $\varphi({\bf 1})={\bf 1}$, $\varphi(a b)=\varphi(b)\varphi(a)$ for
all $a,b\in \cA$, and $\varphi(\bu_i)=\bu_{i+1}$, for $i=0,1,...,2k$,
we easily find that $\varphi(\bu_n)=\bu_{n+1}$ for all $n\in \Z$. So if we have
an expression $\bu_n=f_n(\bu_0,\bu_1,...,\bu_{2k})$ for the general solution of \eqref{rk3one}-\eqref{rk3two}
in terms of the initial data $\bx_0$, we may apply $\varphi$ iteratively on it to get
$\bu_{n}=\varphi^i(f_{n-i}(\bu_0,\bu_1,...,\bu_{2k}))=g_{n,i}(\bu_i,\bu_{i+1},\ldots,\bu_{i+2k})$, namely an expression
for $\bu_n$ in terms of any other initial data $\bx_i$, for all $n,i\in\Z$.
We may also restrict ourselves to $\bu_n$ with $n\in\Z_+$. Indeed, let $\psi$ be the anti-automorphism
such that $\psi(\bu_i)=\bu_{2k-1-i}$, $i=0,1,...,2k$, then we have $\psi(\bu_n)=\bu_{2k-1-n}$
for all $n\in \Z$. Hence if we have an expression $\bu_n=f_n(\bu_0,\bu_1,...,\bu_{2k})$ for $n\geq 2k+1$,
then applying $\varphi\psi$ to it gives $\bu_{2k-n}=h_n(\bu_{0},\bu_1,...,\bu_{2k})$, 
namely an expression for $\bu_n$,
$n\leq 0$, in terms of $\bx_0$. Moreover all of the above expressions are positive Laurent polynomials
iff $f_n$ is a Laurent polynomial for $n\geq 0$.

By the above remarks, we only need prove the statement of the theorem
for $n\geq 0$ and $i=0$.
We first need the following lemma, expressing the integrability of the system \eqref{rk3one}-\eqref{rk3two}.

\begin{lemma} \label{constint}
There exist two linear recursion relations of the form:
\begin{eqnarray*}\bu_{2n+2k+1}- \bu_{2n+1}\bK +\bu_{2n-2k+1} &=&0 \\
\bu_{2n+2k}-\bK \bu_{2n} +\bu_{2n-2k} &=&0 
\end{eqnarray*}
for all $n\in \Z$, and with $K$ expressed in terms of $\bx_0$ as:
\begin{equation}\label{valK}
\bK=\bu_0\bu_{2k}^{-1}+\bu_{2k}\bu_0^{-1}+\sum_{j=1}^k \bu_{2j-1}^{-1}(\bu_{2j}^{-1}+\bu_{2j+2}^{-1})
\end{equation}
\end{lemma}
\begin{proof}
Introduce 
$$ \bK_n=\bu_{2n+1}^{-1}(\bu_{2n+2k+1}+\bu_{2n-2k+1})\qquad \bL_n=(\bu_{2n+2k}+\bu_{2n-2k})\bu_{2n}^{-1} $$
Subtracting \eqref{rk3one} for $n\to n-k$ from that for $n$, we get 
$(\bu_{2n+2k+1}+\bu_{2n-2k+1})\bu_{2n}=\bu_{2n+1}(\bu_{2n-2k}+\bu_{2n+2k})$, hence $\bK_n=\bL_n$.
Moreover, 
\begin{eqnarray*}\bu_{2n+1}K_n\bu_{2n+2}&=&\bu_{2n+1}\bu_{2n+2k+2}-{\bf 1}+\bu_{2n-2k+1}\bu_{2n+2}\\
&=& \bu_{2n+1}(\bu_{2n+2k+2}+\bu_{2n-2k+2})=\bu_{2n+1} L_{n+1} \bu_{2n+2}
\end{eqnarray*}
hence $\bK_n=\bL_{n+1}=\bK_{n+1}=\bK$ is a conserved quantity of the system \eqref{rk3one}-\eqref{rk3two}.
Eq. \eqref{valK} is easily derived by induction on $k$, and the lemma follows.
\end{proof}

Note that $\varphi(\bL_n)=\bK_n$, hence $\varphi(\bK)=\bK$. Analogously, $\psi(\bL_n)=\bK_{k-1-n}$,
hence $\psi(\bK)=\bK$ as well.

Introducing the generating functions 
$$ \bF_i(t)=\sum_{n=0}^\infty t^n \bu_{2i+2kn} \qquad G_i(t)=\sum_{n=0}^\infty t^n \bu_{2i+1+2kn} $$
for $i=0,1,...,k-1$ and the weights:
\begin{eqnarray*}
\bx_{2i}&=&\bu_{2i+2k}\bu_{2i}^{-1}\\  
\bz_{2i}&=&\bK-\bx_{2i}=\bu_{2i-2k}\bu_{2i}^{-1}\\ 
\bbw_{2i}&=&\bz_{2i}\bx_{2i}-{\bf 1}\\
\bx_{2i+1}&=&\bu_{2i+1}^{-1}\bu_{2i+1+2k}\\  
\bz_{2i+1}&=&\bK-\bx_{2i+1}=\bu_{2i+1}^{-1}\bu_{2i+1-2k}\\ 
\bbw_{2i+1}&=&\bx_{2i+1}\bz_{2i+1}-{\bf 1}
\end{eqnarray*}
Lemma \ref{constint} implies the following continued fraction expressions:
\begin{eqnarray*}
\bF_i(t)&=& ({\bf 1}-(\bx_{2i}+\bz_{2i})t+t^2)^{-1}({\bf 1}-\bz_{2i} t) \bu_{2i}\\
&=&\left( {\bf 1}-\bx_{2i} t-({\bf 1}-\bz_{2i}t)^{-1}\bbw_{2i}t^2\right)^{-1}\bu_{2i}\\
\bG_i(t)&=& \bu_{2i+1} ({\bf 1}-\bz_{2i+1} t)({\bf 1}-(\bx_{2i+1}+\bz_{2i+1})t+t^2)^{-1}\\
&=&\bu_{2i+1}\left( {\bf 1}-t\bx_{2i+1}-t^2\bbw_{2i+1}({\bf 1}-\bz_{2i+1} t)^{-1}\right)^{-1}
\end{eqnarray*}
whose series expansions in $t$ have coefficients
that are manifestly positive Laurent polynomials of $(\bx_i,\bz_i,\bbw_i)$. 
Finally, we see that the $\bx$'s are defined recursively by $\bx_0=\bu_{2k}\bu_0^{-1}$
and $\bx_{2j}=\bx_{2j-1}+\bu_{2j-1}^{-1}\bu_{2j}^{-1}$, $\bx_{2j+1}=\bx_{2j}+\bu_{2j+1}^{-1}\bu_{2j}^{-1}$,
hence they form subsums of $\bK$, and so do the $\bz$'s. Also, all $\bx$'s contain the term
$\bu_{2k}\bu_0^{-1}$, while all $\bz$'s contain $\bu_0\bu_{2k}^{-1}$, hence both $\bx_i\bz_i$ and $\bz_i\bx_i$
contain $\bf 1$: we conclude that the $\bx,\bz,\bbw$'s are all positive Laurent polynomials of $\bx_0$.
This completes the proof of Theorem \ref{positrkthree}.

The existence of the linear recursion relations of Lemma \ref{constint} implies by
Proposition \ref{vaniprop} that all $3\times 3$ quasi-Wronskians
$\bW_{2kn+i}=\left\vert (\bu_{2kn+i-2a-2b+8})_{1\leq a,b\leq 3}\right\vert_{1,1}$ must vanish.
Each function $\bF_i(t)$, $i=0,1,...,k-1$, is actually a particular instance of the general non-commutative
case with $r=1$,
discussed in Section \ref{section3}. It is the generating function of paths on the graph $G_1$,
with respective weights $\bw_{1,2}=\bone$ and:
$$(\bw_{1,1},\bw_{2,1},\bw_{2,2})=(t(\by_1+\by_2),t^2\by_3\by_2,t\by_3)=(t\bx_{2i},t^2\bbw_{2i},t\bz_{2i})$$
As we also have a condition $\bbw_{2i}=\bz_{2i}\bx_{2i}-{\bf 1}$, we deduce that the conserved quantity
$C_{2,n}=\bC=\by_3\by_1=\bf 1$. This provides us with examples in the general $A_1$ setting
that are different from the non-commutative $A_1$ $Q$-system of 
Section \ref{cononcosec}. In the latter indeed,
$\bC$ is kept arbitrary, but the value of $\Delta_{2,n}=\bR_{n+1}-\bR_n\bR_{n-1}^{-1}\bR_n$ is fixed by the $Q$-system
relation $\Delta_{2,n}=\bR_n^{-1}\bR_{n-1}\bR_n$. Here we have $\bC=\bf 1$, hence
the recursion relation $\Delta_{2,n+1}=\bR_{n-1}\bR_n^{-1}\Delta_{2,n}$ 
(with $\bR_n=\bu_{2i+2kn}$), but there is no such additional relation. Note that $\bC=\bf 1$ also
implies that there is no quantum version in which $\bC$ would be a non-trivial
central element of $\cA$. This is in agreement with the above-mentioned fact that the exchange matrix
of the associated cluster algebra is singular.

\end{appendix}


\begin{thebibliography}{10}
\bibitem{AssemReit} I. Assem, C. Reutenauer and D. Smith, \emph{Frises}. 
Preprint {\tt arXiv:0906.2026 [math.RA].}

\bibitem{BazResh} V. Bazhanov and N. Reshetikhin, 
Restricted solid-on-solid models connected with simply laced algebras and conformal field theory.
J. Phys. A {\bf 23} (1990) 1477--1492. 

\bibitem{BZ}A. Berenstein, A. Zelevinsky, \emph{Quantum Cluster Algebras},
Adv. Math. {\bf 195} (2005) 405--455. 
{\tt arXiv:math/0404446 [math.QA]}.



    
\bibitem{CZ} P. Caldero and A. Zelevinsky \emph{Laurent 
expansions in cluster algebras via quiver representations}. 
Mosc. Math. J., \textbf{6}  No. 3 (2006), 411-429. {\tt arXiv:math/0604054
    [math.RT]}.
    
        
\bibitem{DFT} P. Di Francesco, \emph{The solution of the $A_r$ $T$-system with
arbitrary boundary}, preprint (2010)
{\tt  arXiv:1002.4427 [math.CO]}.
    
    
\bibitem{DFK08} P. Di Francesco and R. Kedem, \emph{Q-systems as
    cluster algebras II: Cartan matrix of finite type and the polynomial 
    property}, Lett. Math. Phys. {\bf 89} No 3 (2009) 183-216. 
 {\tt arXiv:0803.0362 [math.RT]}.
    

\bibitem{DFK3} P. Di Francesco and R. Kedem, \emph{Q-systems, heaps,
paths and cluster positivity}, Comm. Math. Phys. {\bf 293} No. 3 (2009) 727--802,
DOI 10.1007/s00220-009-0947-5.
{\tt arXiv:0811.3027 [math.CO]}.

\bibitem{DFK09}P. Di Francesco and R. Kedem, 
\emph{$Q$-system cluster algebras, paths and total positivity}, 
SIGMA {\bf 6} (2010) 014, 36 pages,
{\tt arXiv:0906.3421 [math.CO]}.

\bibitem{DFK09a}P. Di Francesco and R. Kedem, \emph{Positivity of the
$T$-system cluster algebra},  Elec. Jour. of Comb. Vol. {\bf 16(1)}
(2009) R140, Oberwolfach preprint OWP 2009-21, 
{\tt arXiv:0908.3122 [math.CO]}.

\bibitem{DFK09b}P. Di Francesco and R. Kedem, 
\emph{Discrete non-commutative integrability: proof of a 
conjecture by M. Kontsevich}, Int. Math. Res. Notices (2010),  
doi:10.1093/imrn/rnq024.
{\tt arXiv:0909.0615 [math-ph]}.


\bibitem{EtingofRetakh} P. Etingof and V. Retakh, \emph{Quantum determinants and 
quasideterminants}, Asian J. Math {\bf 3} (1999) 345--351.

\bibitem{FKV} L.D. Faddeev, R.M. Kashaev and A.Y. Volkov, {\em Strongly coupled quantum discrete Liouville theory. I. Algebraic approach and duality}, 
Commun. Math. Phys. {\bf 219} (2001), pp. 199Ð219.
{\tt arXiv:hep-th/0006156}.


\bibitem{FV} L.D. Faddeev and A.Y. Volkov {\em
Discrete evolution for the zero modes of the quantum Liouville model.}
J. Phys. A {\bf 41} (2008), no. 19, 194008, 12 pp. 
{\tt arXiv:0803.0230 [hep-th]}.




\bibitem{FG} S Fomin and C. Green, {\em Noncommutative Schur functions 
and their applications}, Discrete Math. {\bf 306} No.10-11 (2006) 1080--1096.


\bibitem{FZI} S. Fomin and A. Zelevinsky Cluster Algebras I.
  J. Amer. Math. Soc.   \textbf{15}  (2002),  no. 2, 497--529 {\tt
    arXiv:math/0104151 [math.RT]}.

\bibitem{FZLaurent} S. Fomin And A. Zelevinsky \emph{The Laurent
    phenomenon}.   Adv. in Appl. Math.   \textbf{28}  (2002),  no. 2,
  119--144. {\tt arXiv:math/0104241 [math.CO]}. 
  
\bibitem{FR} E. Frenkel and N. Reshetikhin, {\em The
  $q$-characters of representations of quantum affine algebras and
  deformations of $W$-algebras}. In { Recent developments in quantum affine
  algebras and related topics (Raleigh, NC 1998)}, Contemp. Math. {\bf
  248} (1999),
  163--205.
  
\bibitem{GGRW} I. Gelfand, S. Gelfand, V. Retakh, and R.L. Wilson, 
{\em Quasideterminants.} Adv. Math. {\bf 193} No.1 (2005), 56--141.
{\tt arXiv:math/0208146v4 [math.QA]}.


\bibitem{GKLLRT} I. Gelfand, D. Krob, A. Lascoux, B. Leclerc, 
V. Retakh, and J.-Y. Thibon, {\em Noncommutative symmetric functions.}  
Adv. Math. {\bf 112} No. 2. (1995), 218--348.
{\tt arXiv:hep-th/9407124}



\bibitem{LGV2} I. M. Gessel and X. Viennot, \emph{Binomial
    determinants, paths and hook formulae}, Adv. Math. \textbf{58}
  (1985) 300-321.  

	

\bibitem{Ke07} R. Kedem, \emph{$Q$-systems as cluster algebras}.  J.
  Phys. A: Math. Theor. \textbf{41} (2008) 194011 (14 pages). {\tt
    arXiv:0712.2695 [math.RT]}.

\bibitem{KR} A.~N. Kirillov and N.~Yu. Reshetikhin,
  \emph{Representations of Yangians 
and multiplicity of occurrence of the irreducible components of the
tensor product 
of representations of simple {L}ie algebras}, J. Sov. Math. {\bf 52}
(1990) 3156-3164. 

\bibitem{Kon} M. Kontsevich, private communication.



\bibitem{KNS} A. Kuniba, T. Nakanishi and J. Suzuki, Functional relations in solvable lattice models. 
I. Functional relations and representation theory.
Internat. J. Modern Phys. A {\bf 9} (1994) 5215--5266. 

\bibitem{LGV1} B. Lindstr\"om, \emph{On the vector representations of
	induced matroids}, Bull. London Math. Soc. \textbf{5} (1973)
	85-90. 

\bibitem{PM} G. Musiker and J. Propp \emph{Combinatorial interpretations for rank-two 
cluster algebras of affine type},
Elec. Jour. of Combinatorics {\bf 14} (2007), R15.
{\tt arXiv:math/0602408.}

\bibitem{Nakajima} H. Nakajima, $t$-Analogs of $q$-characters of Kirillov-Reshetikhin modules of quantum affine algebras. 
Represent. Theory {\bf 7} (2003), 259--274.


             
\bibitem{SZ} P. Sherman and A. Zelevinsky,  \emph{Positivity and canonical bases 
in rank 2 cluster algebras of finite and affine types}, 
Mosc. Math. J. \textbf{4}, (2004), no. 4, 947-974, {\tt
arXiv:math/0307082 [math.RT]}.








\end{thebibliography}
\end{document}